\pdfoutput=1
\documentclass[11pt,reqno]{amsart}
\usepackage[T1]{fontenc}
\usepackage{lmodern}
\usepackage[letterpaper,textwidth=6.25in,textheight=9in,centering]{geometry}
\usepackage{amssymb,mathtools,array,needspace}
\usepackage{microtype}
\usepackage[hidelinks,hypertexnames=false]{hyperref}
\theoremstyle{plain}
\newtheorem{theorem}{Theorem}[section]
\newtheorem{lemma}[theorem]{Lemma}
\newtheorem{proposition}[theorem]{Proposition}
\newtheorem{corollary}[theorem]{Corollary}
\theoremstyle{definition}
\newtheorem{definition}[theorem]{Definition}
\newtheorem{convention}[theorem]{Convention}
\newtheorem{problem}[theorem]{Problem}
\theoremstyle{remark}
\newtheorem{remark}[theorem]{Remark}
\hypersetup{
  pdftitle={Width-Bounded Equational Derivations for Finite Graph Expressions},
  pdfauthor={Antonios Kalampakas},
  pdfsubject={Complete preprint with supplementary proofs},
  pdfkeywords={finite graphoids, pattern width, bounded equational coherence, derivational space, graph width, resource-bounded rewriting}
}

\newcommand{\bsum}{\mathbin{\Box}}
\DeclareMathOperator{\val}{val}
\DeclareMathOperator{\tw}{tw}
\DeclareMathOperator{\bw}{bw}
\DeclareMathOperator{\pw}{pw}
\DeclareMathOperator{\lw}{lw}
\DeclareMathOperator{\patw}{patw}
\newcommand{\patwstar}{\patw^{*}}
\newcommand{\patwlin}{\patw_{\mathrm{lin}}}
\DeclareMathOperator{\rank}{rank}
\newcommand{\GR}{\mathbf{GR}}
\newcommand{\Gr}{\mathbf{Gr}}
\newcommand{\DISC}{\mathbf{DISC}}
\newcommand{\PERM}{\mathbf{PERM}}

\newcommand{\mg}{\mathrm{mag}}
\newcommand{\Free}{\mathrm{F}}
\newcommand{\E}{\mathcal{E}}
\newcommand{\NN}{\mathbb{N}}
\newcommand{\io}[2]{i_{#1 #2}}

\title[Width-bounded derivations for graph expressions]{Width-Bounded Equational Derivations for Finite Graph Expressions}
\author[A. Kalampakas]{Antonios Kalampakas}
\address{Department of Mathematics, College of Engineering, American University of the Middle East, Egaila 54200, Kuwait}
\email{antonios.kalampakas@aum.edu.kw}
\date{}
\begin{document}
\begin{abstract}
Completeness of an equational presentation guarantees an equality path but
need not control the resources used along it. For finite graph expressions we
measure derivational space by the largest input-output interface of an
intermediate raw term. For every finite doubly ranked edge alphabet $\Sigma$,
we prove that equal closed expressions of pattern width at most $k$ are joined
by a derivation in which every step applies an equation in either direction and
every intermediate width is at most a computable $B_\Sigma(k)$, independently
of graph size. The derivation uses only the structural magmoid laws and the
fifteen finite-graphoid schemes. The construction compiles each expression through protected
cores and finite routing windows to an encoding-canonical representative
$\operatorname{NF}_k$ of width at most $4k+4$. We call this property bounded
equational coherence. We also prove
$\bw(F)\le\patw(F)\le4(\tw(F)+1)$ whenever the underlying simple graph has at
least two edges, separate layered linear from branching expressions on
cliques, and provide machine-checkable witnesses and finite invariants for the
first five nontrivial clique values.

\end{abstract}
\maketitle
\section{Introduction}\label{sec:intro}

An equality proof can require more space than either of its endpoints. This
phenomenon is hidden by ordinary completeness, which asserts the existence of
a derivation but places no bound on its intermediate interfaces. Global normal
forms for finite graph expressions illustrate the problem sharply because they
may expose every edge at once. We ask whether two narrow descriptions of the
same graph can always be connected without an interface whose size depends on
the graph.

The syntax is the free magmoid over a finite doubly ranked edge alphabet
$\Sigma$ and the five discrete generators
$D=\{\io21,\io01,\io12,\io10,\pi\}$. Bozapalidis and Kalampakas proved that
the quotient by fifteen displayed equation schemes is the magmoid of finite
multi-pointed hypergraphs \cite{BK04}. Thus the equations are sound and
complete for equality of graph values. Engelfriet and Vereijken had already
shown that sequential concatenation and parallel sum generate the relevant
graph algebra \cite{EV97}. In the source-algebra tradition, Bauderon and
Courcelle give complete ground equations for a related one-sided calculus
\cite{BC87}. None bounds derivational space.

The closest precedents are finite presentations for bounded-treewidth source
graphs. Cosme Ll\'opez and Pous treated treewidth two \cite{CLP17}, Doczkal and
Pous gave a convergent reformulation \cite{DP18}, and Doumane, Humeau, and Pous
reached treewidth three with at most four sources \cite{DHP24}. Those systems
build a fixed interface bound into bound-specific syntax. We work in the
unrestricted finite-graphoid syntax and ask uniformly in the endpoint bound
$k$ whether its fixed fifteen schemes admit bounded-width equality paths.

The largest input-output rank sum of a raw subexpression is its \emph{pattern
width}. In the singleton fragment of Peterseim and
Lopuha{\"a}-Zwakenberg's hypergraph terms, the structural maps correspond to
$e_1,\pi,\io12,\io10,\io21,\io01$ up to composition conventions, and the
widths coincide \cite[Defs.~23 and~29]{PLZ26}. Their algorithm algebraises a
dot diagram of primal treewidth $h$ at width $O(h)$ in the free strict
hypergraph category \cite[Def.~28 and Thm.~32]{PLZ26}. It selects one
representative, while we bound an elementary-axiom path between any two
supplied narrow terms. Copy and merge make pattern width tree-like because a
vertex can remain on one wire and be re-exposed for later edges.

For fixed $k$ and finite $\Sigma$, every closed width-$k$ expression reaches a
representative selected from its graph value alone through terms of width at
most $B_\Sigma(k)$, independently of graph size. We call this \emph{bounded
equational coherence}. It concerns connected components of the width-filtered
undirected derivation graph, not Mac Lane or polygraphic coherence
\cite{Mim25}.

The proof protects live vertices, extracts and locally normalizes a bounded
edge-partition tree, and compares trees by four-cell switches. The containment
$\partial(X\cap Y)\subseteq\partial X\cup\partial Y$ prevents interface
accumulation. All coherence claims concern closed raw expressions.

For graphs whose underlying simple graph has at least two edges, the parameter
is calibrated by
\[
  \bw(F)\le\patw(F)\le4(\tw(F)+1).
\]
Layered linear cliques have width between $2n-3$ and $2n$, whereas branching
terms have width at most $n+1$. The machine-verified values for $K_3$ through
$K_7$ are $4,5,6,6,6$.

Connection-algebra and wiring-diagram normal forms do not bound derivations
from supplied expressions \cite{BC87,BGM02,CG99,PSV21}. Finite-graphoid
automata and varieties address recognizability and equational classes
\cite{BK04,BK17,BK06,BK08}. Cospan, nominal, and tree-decomposable approaches
relate syntax to graph structure \cite{BBFK13,HMS15,BMS20}. Unlike Courcelle's
injective sources, finite-graphoid interfaces may repeat vertices
\cite{CE12,EV97}.

Monoidal width minimizes over decomposition trees and charges composition
interfaces \cite{DLS23}, while our width belongs to a chosen raw tree and we
bound its peak along an equality path. The closer proof-complexity analogue is
Olshanskii's derivational space function \cite{Ols13}. We replace word length
by typed interface rank and obtain a bound from endpoint width even for
arbitrarily large terms. Small term reachability instead bounds target size
\cite{BG24}. Further rewriting context appears in
\cite{FS19,BGKSZ22,BGKSZ22II,BGKSZ22III,MPZ25,GK26,TBGH25}.

The supplementary sections are included in this preprint. Sections~\ref{sec:calibration-details}
through \ref{sec:cliques} give the expanded graph-width calibration and clique arguments,
Section~\ref{sec:normal-details} gives the complete local constructions for bounded equational
coherence, Section~\ref{supp:sec:gi-encoding} details the graph-isomorphism encoding, and Section~\ref{sec:saturation-details}
proves the finite-saturation and certificate semantics.

\section{Preliminaries: the axiomatization of graphs}\label{sec:prelim}

We briefly recall the setting of \cite{BK04}, to which we refer for all
details.

A \emph{doubly ranked alphabet} is a set $X$ with a function
$\rank\colon X\to\NN\times\NN$. A \emph{raw expression} is a fully
parenthesized term over $X$, the unit symbols $e_n$, and the two binary
operations below. We write $\Free_{m,n}(X)$ for the raw expressions of rank
$(m,n)$. A \emph{magmoid} is a doubly ranked set $M=(M_{m,n})$ with operations
$\circ\colon M_{m,n}\times M_{n,k}\to M_{m,k}$ and
$\bsum\colon M_{m,n}\times M_{m',n'}\to M_{m+m',n+n'}$ satisfying the
following fully parenthesized structural schemes whenever the terms are well
typed:
\begin{align*}
 ((f\circ g)\circ h)&=f\circ(g\circ h),&
 ((f\bsum g)\bsum h)&=f\bsum(g\bsum h),\\
 e_m\circ f&=f,& f\circ e_n&=f,\\
 e_0\bsum f&=f,& f\bsum e_0&=f,\\
 e_m\bsum e_n&=e_{m+n},&
 (f\circ f')\bsum(g\circ g')&=(f\bsum g)\circ(f'\bsum g').
\end{align*}
These raw equation schemes are denoted by $\mathcal B$. Equivalently, a
magmoid is a strict monoidal category with objects $\NN$
\cite{AD78a,AD79,Hot65a,Hot65b}. We write
$\mg(X)$ for the free magmoid, which is the quotient of raw expressions by
$\mathcal B$. Its elements are called \emph{patterns}. Pattern width below is
a property of a raw expression and can differ between raw representatives of
the same pattern.

Let $\GR(\Sigma)$ consist of boundary-preserving isomorphism classes of finite
$\Sigma$-labeled hypergraphs with ordered begin and end sequences. Product
glues the end sequence of the first representative positionally to the begin
sequence of the second. Sum is disjoint union with concatenated interfaces.
These operations are independent of representatives and make $\GR(\Sigma)$ a
magmoid. Equality in $\GR(\Sigma)$ therefore means labeled pointed-hypergraph
isomorphism. When a later construction names vertices or edge occurrences, it
chooses a representative and is transported by isomorphisms. The symbols
$\DISC$ and $\PERM$ denote the discrete and permutation submagmoids. Let
$D=\{\io21,\io01,\io12,\io10,\pi\}$ have ranks
$(2,1),(0,1),(1,2),(1,0),(2,2)$, and assume $\Sigma\cap D=\varnothing$.
Interpret these symbols by the elementary discrete graphs
$I_{2,1},I_{0,1},I_{1,2},I_{1,0},\Pi$, and let
$\val\colon\mg(\Sigma\cup D)\to\GR(\Sigma)$ be the unique magmoid morphism
sending $\sigma\mapsto G(\sigma)$ and $D$ to the graphs above. The map $\val$ is
surjective by \cite{EV97}. We also write $\val$ for its composite with the
quotient map from raw expressions to $\mg(\Sigma\cup D)$.

We fix raw representatives for all permutation notation. Put
$\widehat e_0=e_0$, $\widehat e_1=e_1$, and
$\widehat e_{n+1}=(\widehat e_n\bsum e_1)$ for $n\ge1$. An iterated sum or
composition is left associated, with the empty sum equal to $e_0$ and the
empty composition on $n$ wires equal to $\widehat e_n$. For $1\le j<r$, let
$\tau_{r,j}$ be the fixed left-associated sum whose literal summand list is
$j-1$ copies of $e_1$, then $\pi$, then $r-j-1$ copies of $e_1$.
Order $\tau_{r,1},\ldots,\tau_{r,r-1}$ by their indices.
For $r=m+n$, define $\pi_{m,n}^{\rm raw}$ to be the left-associated
composition of the lexicographically first shortest word in
$\tau_{r,1},\ldots,\tau_{r,r-1}$ whose value moves the last $m$ positions
before the first $n$ positions while preserving both internal orders. Put
$\pi_{m,0}^{\rm raw}=\pi_{0,m}^{\rm raw}=\widehat e_m$. This fixes a fully
parenthesized raw term for every block permutation, of width exactly $2(m+n)$.

For every doubly ranked alphabet $X$, define $\E_X$ in the raw term algebra
over $X\cup D$ by the following schemes:
\begingroup
\allowdisplaybreaks[4]
\begin{align}
&\pi\circ\pi=e_2 \tag{A1}\\
&((e\bsum\pi)\circ(\pi\bsum e))\circ(e\bsum\pi)
 =((\pi\bsum e)\circ(e\bsum\pi))\circ(\pi\bsum e) \tag{A2}\\
&(e\bsum\io21)\circ\io21=(\io21\bsum e)\circ\io21 \tag{A3}\\
&(e\bsum\io01)\circ\io21=e \tag{A4}\\
&\pi\circ\io21=\io21 \tag{A5}\\
&((\pi\bsum e)\circ(e\bsum\pi))\circ(\io21\bsum e)
 =(e\bsum\io21)\circ\pi \tag{A6}\\
&(e\bsum\io01)\circ\pi=\io01\bsum e \tag{A7}\\
&\io12\circ(e\bsum\io12)=\io12\circ(\io12\bsum e) \tag{A8}\\
&\io12\circ(e\bsum\io10)=e \tag{A9}\\
&\io12\circ\pi=\io12 \tag{A10}\\
&((\io12\bsum e)\circ(e\bsum\pi))\circ(\pi\bsum e)
 =\pi\circ(e\bsum\io12) \tag{A11}\\
&\pi\circ(e\bsum\io10)=\io10\bsum e \tag{A12}\\
&(\io12\bsum e)\circ(e\bsum\io21)=\io21\circ\io12 \tag{A13}\\
&\io12\circ\io21=e \tag{A14}\\
&\pi_{p,1}^{\rm raw}\circ(\sigma\bsum e)
 =(e\bsum\sigma)\circ\pi_{q,1}^{\rm raw} \tag{A15}
\end{align}
\endgroup
Here $\sigma\in X_{p,q}$ and $e=e_1$. Every displayed side is the literal raw
term with the shown parentheses. The structural laws together with (A1) and
(A2) identify adjacent-transposition words with the same permutation, so
these representatives induce exactly the pattern equations of
\cite[Sect.~3]{BK04}. Thus $\E_X$ is a set of equations between fixed raw
terms.
Write $p\xleftrightarrow{*}_{\mathcal B\cup\E_X}q$ for the congruence on raw
expressions generated by contextual uses of the structural laws and these
equations. After quotienting by $\mathcal B$, write
$p\xleftrightarrow{*}_{\E_X}q$ for the induced congruence on patterns.
For the fixed edge alphabet $\Sigma$, write $\E=\E_\Sigma$ and set
$\Gr(\Sigma)=\mg(\Sigma\cup D)/\xleftrightarrow{*}_{\E}$, the
magmoid of \emph{graph patterns}. We also call $\Gr(\Sigma)$ the finite
graphoid over $\Sigma$.

\begin{theorem}[{\cite[Cor.~3]{BK04}}]\label{thm:bk}
The evaluation morphism on $\mg(\Sigma\cup D)$ factors through an isomorphism
$\val\colon\Gr(\Sigma)\to\GR(\Sigma)$. Thus, for raw expressions $p,q$,
$\val(p)=\val(q)$ if and only if
$p\xleftrightarrow{*}_{\mathcal B\cup\E}q$.
\end{theorem}

\begin{convention}\label{conv:enc}
Unless stated otherwise, $\Sigma=\{a:1\to1\}$. Closed patterns denote finite
directed multigraphs, with wire-identification classes as vertices and
$a$-occurrences as arcs. For a simple graph $F$, let $\operatorname{Ori}(F)$
contain every one-orientation-per-edge digraph. Then $\val(t)=F$ means
$\val(t)\in\operatorname{Ori}(F)$, so $\patw(F)$ minimizes over orientations.
Classical widths of a directed value refer to the underlying simple graph
obtained by deleting loops and collapsing each nonempty parallel class. The
clique verifier uses the same convention by seeding both orientations of each
edge. All graphs are finite. Complexity statements use expanded arity, under
which $e_n$ and an explicit length-$n$ port list contribute $\Theta(n)$ input
symbols. This encoding does not affect the equational results.
\end{convention}

The following proposition records the complexity context needed later.

\begin{proposition}\label{prop:gicomplete}
The word problem of\/ $\Gr(\Sigma)$ takes raw expressions $s,t$. It rejects
them when their ranks differ and otherwise asks whether
$s\xleftrightarrow{*}_{\mathcal B\cup\E}t$. Under the expanded-arity encoding
and the default $\Sigma=\{a\}$ of Convention~\ref{conv:enc}, this problem is
polynomial-time many-one equivalent to graph isomorphism.
\end{proposition}

\begin{proof}[Proof sketch]
Structural evaluation is polynomial, and labeled pointed-hypergraph
isomorphism reduces to GI by standard gadgets. Conversely, encode each
undirected edge by two opposite arcs, maintain one wire per vertex, route arc
endpoints by adjacent transpositions, and close the wires. This gives a raw
expression of size $O(|E||V|^2+|V|)$, including isolated vertices.
Theorem~\ref{thm:bk} completes the reduction. Section~\ref{supp:sec:gi-encoding}, gives
the construction and audit.
\end{proof}

\section{Pattern width}\label{sec:width}

\begin{definition}[pattern width]\label{def:width}
For an expression $t\in\Free_{m,n}(X_\Sigma)$, $X_\Sigma=\Sigma\cup D$,
define
\[
\patw(t)\;=\;\max\{\,a+b \;:\; t' \text{ is a subterm occurrence of } t
\text{ of rank } (a,b)\,\},
\]
where subterms include the leaves (letters of $X_\Sigma$ and units $e_k$)
and $t$ itself. For a graph $F$ (under Convention~\ref{conv:enc}) set
\[
\patw(F)\;=\;\min\{\,\patw(t)\;:\;t\in\Free_{0,0}(X_\Sigma),\ \val(t)=F\,\},
\]
which is finite by the surjectivity of $\val$ \cite{EV97}.
\end{definition}

\begin{definition}[layered linear expressions]\label{def:linear}
A closed raw expression is \emph{layered linear} if there is an integer
$L\ge1$, together with integers $c_0,\ldots,c_L$ satisfying $c_0=c_L=0$ and blocks
$B_j\in\Free_{c_{j-1},c_j}(X_\Sigma)$ such that each $B_j$ is the fixed
left-associated $\bsum$-term on a nonempty ordered list of generators and unit
symbols, and
\[
 t=(((B_1\circ B_2)\circ B_3)\circ\cdots)\circ B_L.
\]
For $L=1$, this means $t=B_1$, so $e_0$ is admitted by taking $B_1=e_0$.
Thus layering includes its types and its raw
parenthesization. No reassociation is performed when width is computed or when
the minimum is taken. The only composition subterms not already inside a block
are the prefixes of type $(0,c_j)$. Write $\patwlin(F)$ for the minimum width
over these fixed left-comb expressions with value $F$. Call such an expression
\emph{elementary} if every block is the fixed-bracketing term
$e_r\bsum x\bsum e_s$ for one $x\in X_\Sigma$ and $r,s\ge0$, and put
\[
 \patw^{\mathrm{el}}(F)=\min\{\patw(t):t\text{ is elementary and }\val(t)=F\}.
\]
\end{definition}

Copy and merge make the measure tree-like: directed paths have width at most
$2$, out-stars and directed cycles at most $4$ despite unbounded out-star
cutwidth, and
$\patw(F_1\sqcup F_2)\le\max\{\patw(F_1),\patw(F_2)\}$. Also
$\patw(K_n)=\Theta(n)$ by Corollary~\ref{cor:calib}.

The saturation below uses reduced primitive leaves. Replace $e_r$ for $r\ge1$
by $\widehat e_r$, preserving its largest rank sum $2r$, and remove internal
$e_0$ leaves by the structural unit laws. Thus a target with a vertex or edge
has an optimal expression without $e_0$, while $e_0$ represents the
vertexless target.

\begin{lemma}[internal nodes stay internal]\label{lem:internal}
Let $t$ be closed, let $t'$ be a subterm occurrence, and let
$\theta\colon\val(t')\to\val(t)$ be the canonical map induced by its context.
If a node $v$ of $\val(t')$ occurs in neither boundary sequence, then
$\theta^{-1}(\theta(v))=\{v\}$ and every edge incident with $\theta(v)$ is the
image of an edge already in $\val(t')$.
\end{lemma}

\begin{proof}
Induct on the depth of $t'$ below the root. If the parent is $t'\circ u$, only
end nodes of $\val(t')$ are identified with begin nodes of $\val(u)$. The
class of $v$ is therefore a singleton and no edge of $\val(u)$ meets it. Its
image remains internal at the parent. The cases $u\circ t'$, $t'\bsum u$, and
$u\bsum t'$ are analogous. Apply the induction hypothesis at the parent.
\end{proof}

For a target having a vertex or edge and a bound $W$, saturation works either
with a directed multigraph $\vec F=(V,A)$ or with a simple undirected graph
$H=(V,E)$ whose edge leaves may receive either orientation. Write $E^\sharp$
for $A$ or $E$ respectively. A state
$s=(P,C,\kappa,b,e)$ records the placed occurrences $P\subseteq E^\sharp$,
the finalized vertices $C\subseteq V$, the target labels $\kappa$ of live
boundary pieces, and ordered boundary words $b,e$ of total length at most
$W$. Its realization maps graph letters bijectively to $P$, maps its internal
nodes bijectively onto $C$, and maps its live nodes by $\kappa$. Directed mode
preserves tails and heads. Undirected mode preserves unordered endpoint pairs
and thereby records the orientation chosen for each placed edge.

Transitions require disjoint placed and finalized sets, exclude finalized
labels from the other live boundary, and retain legal states within the cap.
Sum concatenates. Composition also matches middle arity and labels, takes the
positional quotient, and finalizes a label only when its last live class
vanishes after all incident edges are placed. The accepting state is
$(E^\sharp,V,(),(),())$. Definitions~\ref{supp:def:raw-state} through \ref{supp:def:composition-successor},
give these guards, and Section~\ref{supp:sec:symmetry-certificates} gives the symmetry reduction.

\begin{proposition}[finite saturation]\label{prop:finiteclosure}
Under the fixed rank-$(1,1)$ alphabet of Convention~\ref{conv:enc}, saturation
under these transitions terminates for every finite target and width bound
$W$. For a directed target $\vec F$ having a vertex or edge, its accepting
state is reached
if and only if $\patw(\vec F)\le W$. For a simple undirected target $H$ having
a vertex or edge, both orientations of every edge are seeded, and the accepting state is
reached if and only if
\[
 \min_{\vec H\in\operatorname{Ori}(H)}\patw(\vec H)=\patw(H)\le W.
\]
For the vertexless target, $e_0$ is accepted directly.
\end{proposition}

\begin{proof}
A width-$W$ expression induces these states at all subterms by
Lemma~\ref{lem:internal}. Conversely, induction on a retained construction
tree shows that every reached state realizes exactly its recorded data. An
accepting expression therefore has value $\vec F$ in directed mode and lies in
$\operatorname{Ori}(H)$ in undirected mode. There are finitely many choices of
$P,C$ and at most $W$ boundary occurrences up to renaming. Lemma~\ref{supp:lem:endpoint-incidence} and
Proposition~\ref{supp:prop:finiteclosure} in the supplementary sections give the two inductions guard by
guard.
\end{proof}

For a stored set $R$ of canonical representatives, the lower-bound invariant
is its orbit union: this union contains every seed, is transition closed, and
omits the accepting orbit. In undirected mode both oriented seeds for every
edge exclude all orientations. The supplied checker verifies representative-level
conditions implying this invariant independently of the search engine, as
formalized in Definition~\ref{supp:def:certificate} and Proposition~\ref{supp:prop:certificate-soundness}.

\section{Bounded normal forms and bounded coherence}\label{sec:normal}

Throughout this section, $\Sigma$ is an arbitrary fixed finite doubly ranked
edge alphabet and $\E=\E_\Sigma$.

The global form in Theorem~\ref{thm:bk} is not width-sensitive. We instead
compile along an edge-partition tree. Core-frame factorization protects exposed
vertices, finite local schemes reduce routing, four-cell interpolation compares
bounded trees, and fixed serialization selects a representative. The finite
renderer facts used below are typed linear-use templates and value checks
(Definition~\ref{supp:def:renderer} and Proposition~\ref{supp:prop:renderer-validity}), terminating coverage (Section~\ref{supp:sec:coverage-audit} and
Lemma~\ref{supp:lem:renderer-coverage}), endpoint bounds (Proposition~\ref{supp:prop:window-arity}), exact routing (Lemma~\ref{supp:lem:routing-endpoint}),
quadrant endpoints (Lemma~\ref{supp:lem:quadrant-degeneracies}), and literal carrier exposure (Lemma~\ref{supp:lem:carrier-exposure})
in the supplementary sections.

\subsection{Derivation width and boundary control}

\begin{definition}[derivation width]\label{def:derw}
For raw expressions $s,t$ of the same type, a derivation is a sequence
$s=t_0,t_1,\ldots,t_N=t$ of raw expressions of that type. At one step, choose
a fully parenthesized scheme in $\mathcal B\cup\E$, substitute well-typed raw
expressions for all its metavariables, choose one literal occurrence of the
resulting left or right side as a subterm of $t_i$, and replace it by the other
side. Either direction is allowed. Its width is $\max_i\patw(t_i)$. We write
$s\equiv_{\mathcal B\cup\E}^{w}t$ when such a derivation of width at most $w$
exists.
\end{definition}

\begin{proposition}[the global-form obstruction]\label{prop:obstruction}
Let $F$ be a closed graph with $m\ge1$ edge occurrences. Every backfold normal
form of \cite[Thm.~4a]{BK04} for $F$ contains a subexpression of rank sum
\[
 \sum_{i=1}^{m}(p_i+q_i),
\]
where the $i$th edge label has rank $(p_i,q_i)$. Under
Convention~\ref{conv:enc} this is $2m$. Thus routing a derivation through that
normal form gives no bound in terms of the pattern widths of the endpoints on
families with an unbounded number of positive-arity edge occurrences, including
the ordinary graph encoding of Convention~\ref{conv:enc}.
\end{proposition}

\begin{proof}
The closed backfold form is
\[
 \gamma\circ\bigl(\operatorname{bf}(\sigma_1)\bsum\cdots\bsum
 \operatorname{bf}(\sigma_m)\bigr),
\]
where $\gamma$ has rank $(0,s)$ for $s=\sum_i(p_i+q_i)$ and the $i$th
backfolded letter has rank $(p_i+q_i,0)$. The right factor is a subexpression
of rank $(s,0)$.
\end{proof}

We use edge occurrences as atoms. For a graph value $G$, let $E(G)$ denote
its finite set of labeled edge occurrences and let $P(G)$ be the set of
vertices appearing in its begin or end sequence. If $A\subseteq E(G)$, write
$V_G(A)$ for the vertices incident with an edge in $A$ and set
\begin{equation}\label{eq:atomboundary}
 \partial_G A=
 V_G(A)\cap\bigl(V_G(E(G)\setminus A)\cup P(G)\bigr).
\end{equation}
Thus a vertex belongs to $\partial_G A$ if the $A$-part must retain it for an
edge outside $A$ or for an outer port. Repetitions in the begin and end
sequences are counted once in $P(G)$. Their ordered multiplicities will be
restored by a separate port frame.

We also record the largest arity of an edge occurrence,
\[
 \alpha(G)=\max\bigl(\{p+q:\text{an edge occurrence of $G$ has a
 label of rank $(p,q)$}\}\cup\{0\}\bigr).
\]
This parameter is necessary because a one-edge graph can have empty atom
boundary while its edge letter has large rank. If $G$ is the value of a
width-$k$ expression, then $\alpha(G)\le k$.

A rooted edge-partition tree of $G$ is a rooted plane full binary tree whose leaves
are the edge occurrences. A node $x$ represents the set $A_x$ of leaves below
it. Its boundary width is
\[
  \operatorname{ewd}_G(T)=\max_x|\partial_G A_x|.
\]
Components containing no edge are excluded from the tree. They are discrete
components and will be handled by the discrete normal form in
Lemma~\ref{lem:relative-spider}.

\begin{lemma}[intersection boundary]\label{lem:intersection}
For all $A,B\subseteq E(G)$,
\[
 \partial_G(A\cap B)\subseteq\partial_G A\cup\partial_G B.
\]
\end{lemma}

\begin{proof}
Let $v\in\partial_G(A\cap B)$ and choose an incident edge
$e\in A\cap B$. If $v\in P(G)$, then $v$ belongs to both boundaries on the
right. Otherwise $v$ has an incident edge $f\notin A\cap B$. If $f\notin A$,
then $v\in\partial_G A$. If $f\in A$, then $f\notin B$ and
$v\in\partial_G B$.
\end{proof}

The next lemma is the width-sensitive replacement for an unrestricted spider
normalization. The holes in the statement are important. They prevent the
normalization from collecting all graph letters into one parallel layer.

\subsection{Protected cores and local normalization}

\begin{lemma}[naturality under substitution]\label{lem:substitution}
Let $H$ be a finite doubly ranked alphabet of formal boxes and let $\Delta$ be
a finite $\mathcal B\cup\E_{\Sigma\cup H}$-derivation between two linear-use
contexts. There is an effectively computable monotone
$s_\Delta:\NN\to\NN$ such that every rank-preserving substitution
$h\mapsto u_h$ gives a $\mathcal B\cup\E_\Sigma$-derivation of width at most
\[
 s_\Delta\!\left(\max\bigl(\{\patw(u_h):h\in H\}\cup\{0\}\bigr)\right).
\]
\end{lemma}

\begin{proof}
Only an (A15) step for a formal box needs expansion. For every raw
$u:p\to q$, structural induction derives
\[
 \pi_{p,1}^{\rm raw}\circ(u\bsum e_1)
 \xleftrightarrow{*}_{\mathcal B\cup\E_\Sigma}
 (e_1\bsum u)\circ\pi_{q,1}^{\rm raw}.
\]
The base cases use (A1), (A2), (A5) through (A7), (A10) through (A12), and
(A15). The two
induction clauses use interchange and sequential child derivations. Since
$\Sigma$ is finite, this gives a computable monotone bound $x_\Sigma(n)$ for
the displayed exchange when $\patw(u)\le n$. If $c_\Delta$ bounds the fixed
context ranks and endpoint widths in $\Delta$, replacing its formal (A15)
steps gives
$s_\Delta(n)=\max\{n,c_\Delta,x_\Sigma(n)+c_\Delta\}$ after monotone closure.
Same-type substitution changes no ancestor rank. The induction templates and
the bound $x_\Sigma(n)\ge4n+4$ are in Section~\ref{supp:sec:naturality-audit} and
Equation~\ref{supp:eq:naturality}.
\end{proof}

We first fix the raw discrete ingredients used by every router.
Put $\mu_0=\io01$, $\mu_1=e_1$, $\delta_0=\io10$, and
$\delta_1=e_1$. For $j\ge1$, define the fixed left-associated terms
\[
 \mu_{j+1}=(\mu_j\bsum e_1)\circ\io21,
 \qquad
 \delta_{j+1}=\io12\circ(\delta_j\bsum e_1).
\]
Thus $\mu_j:j\to1$ merges $j$ wires and $\delta_j:1\to j$ copies one
wire to $j$ wires, including the nullary cases.
Put $z=\io01\circ\io10$, $Z_0=e_0$, and
$Z_{c+1}=z\bsum Z_c$, with right association.

\begin{definition}[cores, frames, and canonical routers]
\label{def:core-frame-router}
A \emph{named core} is a raw expression $C:0\to r$ whose end ports carry
vertex names. It is \emph{reduced} when the names are distinct and
\emph{protected} when a surrounding discrete context uses $C$ once as an
opaque hole. A graph-bearing core of rank zero remains opaque. Only a
port-free component created by the discrete context is represented by
$z$ and contributes to the fixed suffix $Z_c$.

Let $\mathbf x,\mathbf y$ be ordered input and output position lists and let
$\eta$ partition their union. Order the blocks by first occurrence, put
$p_B=|B\cap\mathbf x|$ and $q_B=|B\cap\mathbf y|$, and let
$P_{\rm in},P_{\rm out}$ be the fixed permutation terms for block order and
the prescribed output order. The canonical router is
\begin{equation}\label{eq:canonical-router}
 \operatorname{Can}(\eta,c)=
 \left(P_{\rm in}\circ
 \Bigl(\mathop{\bsum}_{B\in\eta}\mu_{p_B}\Bigr)\circ
 \Bigl(\mathop{\bsum}_{B\in\eta}\delta_{q_B}\Bigr)\circ
 P_{\rm out}\right)\bsum Z_c.
\end{equation}
All sums are left associated, units are expanded, and the empty partition
uses $e_0$. For a protected hole $h:0\to r$ with ordered names
$\boldsymbol\nu$, and prescribed begin and end tuples $\mathbf b,\mathbf e$,
let $\eta(\mathbf b,\boldsymbol\nu\mid\mathbf e)$ identify equally named
positions. Its fixed port frame is
\begin{equation}\label{eq:raw-frame}
 \operatorname{Frame}_{\mathbf b,\mathbf e,\boldsymbol\nu,c}[h]
 = (e_{|\mathbf b|}\bsum h)\circ
 \operatorname{Can}(\eta(\mathbf b,\boldsymbol\nu\mid\mathbf e),c).
\end{equation}
\end{definition}

For later compilation, let $u$ be a subterm of a fixed closed expression with
value $G$, let $A_u$ be its edge occurrences, and let
$\theta_u:\val(u)\to G$ be its context map. On the
boundary positions, with represented vertices $x_i$, put
\[
 i\mathrel{\beta_u}j\Longleftrightarrow x_i=x_j,
 \qquad
 i\mathrel{\gamma_u}j\Longleftrightarrow
 \theta_u(x_i)=\theta_u(x_j).
\]
Here $\beta_u$ is local equality, $\gamma_u$ is equality in $G$, and
$q_u:\beta_u\twoheadrightarrow\gamma_u$ is the quotient. For $L\in\gamma_u$,
put $v_L=\theta_u(x_i)$ for any $i\in L$. Let
\begin{equation}\label{eq:state-rule}
 \lambda_u=\{L\in\gamma_u:
 v_L\in V_G(A_u)\cap
 (V_G(E(G)\setminus A_u)\cup P(G))\}.
\end{equation}
At $u=v\mathbin\star w$, the parent partition is the disjoint union of the
child partitions, followed for $\star=\circ$ by exactly the positional
gluings. Classes absent from the parent boundary are closed. Distinct
$\beta_u$-blocks are not merged merely because they have the same image in
$\gamma_u$. A \emph{parse-aligned framed realization} uses reduced cores and
the fixed routers above at every parse node. It is in \emph{completed-step
form} when every inactive region is protected and all extracted scalars form
a fixed right-associated suffix ending in $Z_c$. All names, partitions, and
activity marks are metasyntactic.
The fully expanded frames and edge-free normalization are in Equations~\ref{supp:eq:name} through \ref{supp:eq:canonical-frame} and Lemma~\ref{supp:lem:edge-free}.

The local moves are rendered with the fixed raw routers of
Definition~\ref{def:core-frame-router}. Their expanded generator-level
formulas, permutation padding, and scalar cases are given in Definition~\ref{supp:def:renderer} and Equations~\ref{supp:eq:binary-router} through \ref{supp:eq:scalar-compose}. We retain here the typed four-cell endpoint used in the
interpolation. Let $I\subseteq\{0,1\}^2$ record the occupied edge cells. For
$(i,j)\in I$, let $h_{ij}:0\to r_{ij}$ be a distinct linear-use hole, where
$r_{ij}$ is allowed to be zero. For $(i,j)\notin I$, put $r_{ij}=0$ and write
$h_{ij}=e_0$ only as a unit placeholder, with no formal hole. Thus an occupied
cell with empty protected interface remains an opaque rank-zero graph-bearing
core. Let $\eta$ be the partition of the protected ports of the occupied holes
by final vertex, and let $\mathcal O$ be the ordered blocks surviving in the
common outer interface. The partitions $\eta_{i*}$ and $\eta_{*j}$ first
assemble row $i$ and column $j$. The partitions $\eta_R$ and $\eta_C$ merge
their two staging banks and output $\mathcal O$. With the fixed plane orders,
set
\begin{align}
 \mathsf{Row}_\eta[h]&=
 \Bigl(((h_{00}\bsum h_{01})\circ\operatorname{Can}(\eta_{0*},0))
 \bsum((h_{10}\bsum h_{11})\circ\operatorname{Can}(\eta_{1*},0))\Bigr)
 \circ\operatorname{Can}(\eta_R,c),\label{eq:row-template}\\
 \mathsf{Col}_\eta[h]&=
 \Bigl(((h_{00}\bsum h_{10})\circ\operatorname{Can}(\eta_{*0},0))
 \bsum((h_{01}\bsum h_{11})\circ\operatorname{Can}(\eta_{*1},0))\Bigr)
 \circ\operatorname{Can}(\eta_C,c),\label{eq:col-template}
\end{align}
where $c\in\{0,1\}$. Every occupied formal hole occurs once in each term,
including every hole of type $0\to0$, and the two terms induce the same
partition, ordered outer interface, and scalar tag.

\Needspace{8\baselineskip}
\begin{definition}[admissible local schemes]\label{def:admissible-schemes}
A $q$-local datum records one of six kinds, namely leaf frame, binary router,
fold, prune, four-cell switch, or two-scalar move. It records at most four
ordered linear-use hole types, the outer type, the relevant partitions and
orders, and at most two exposed auxiliary scalars. Four-cell data also record
occupancy, so an occupied rank-zero core remains a hole. The scalar increment
is $c\in\{0,1,2\}$. A completed suffix of any length is one opaque rank-zero
carrier occurring once on each side and is not counted by $c$.

The deterministic renderer uses $\operatorname{Can}$,
$\operatorname{Frame}$, the fixed adjacent-transposition words, and
Equations~\eqref{eq:row-template} and \eqref{eq:col-template}. Absent cells contribute
$e_0$. Scalar data record the two raw orders and parenthesizations.

Port load counts named positions in hole interfaces, the outer interface, and
internal banks, including repetitions caused by staging. A datum is valid when
its endpoints have the recorded type and load at most $q$, use every formal
hole once with the same typed incidence list, and have the same pointed value
when holes are distinct letters. Let $\mathcal M_q$ be the valid data up to
order-preserving renaming. Definition~\ref{supp:def:renderer}, gives the
expanded formulas.
\end{definition}

Fix a discrete context $Q:p\to q$, linear in reduced named holes
$h_i:0\to r_i$, including opaque graph-bearing holes with $r_i=0$. Suppose
the holes and instantiated context have width at most $w$, and at most $r$
protected and outer positions occur. Put
\[
 B_D(r,w)=2\bigl(r+4\max\{r,w\}\bigr)+4,
 \qquad Q_D(r,w)=12\bigl(B_D(r,w)+1\bigr).
\]
For a completed subcontext $K:p_K\to q_K$, its normalization record is
\[
 \mathsf N(K)=
 \bigl((e_{p_K}\bsum A_K)\circ\operatorname{Can}(\eta_K,0)\bigr)
 \bsum S_K:p_K\longrightarrow q_K.
\]
$A_K$ aggregates the positive-rank holes, $\eta_K$ is the exact port
partition, and $S_K$ is the fixed suffix of rank-zero graph-bearing holes and
$D$-only factors $z$. Every original hole occurs once in $A_K$ or $S_K$.

\begin{lemma}[typed coverage of discrete normalization]
\label{lem:typed-coverage}
Deterministic bottom-up normalization terminates at the unique record
$\mathsf N(Q)$. Every rendered local step belongs to $\mathcal M_{Q_D(r,w)}$,
uses at most four aggregate holes once, and is interleaved with only the
literal-address carrier exposures of Lemma~\ref{supp:lem:carrier-exposure}. Every
graph-bearing rank-zero hole is preserved. A parse or routing window with at
most six slots of load $k$ has rendered load at most $12(k+1)$.
\end{lemma}

\begin{proof}
Syntax induction uses disjoint union, positional quotient, and first-occurrence
order. It turns only port-free discrete classes into $z$ and keeps graph-bearing
scalars opaque. Direct evaluation verifies the retained data. Carrier load
$\max\{r,w\}$ gives $B_D$, and the six-slot audit gives $Q_D$ and $12(k+1)$.
Exposure only reassociates. Fixed postorder makes
$\Theta=(U,N,\phi,I,S_{\rm out},S_{\rm inv},S_{\rm ass})$ decrease.
Lemmas~\ref{supp:lem:renderer-coverage} and \ref{supp:lem:carrier-exposure}, give the full state proof.
\end{proof}

\begin{lemma}[finite-template invariant]\label{lem:template-invariant}
If every substituted hole, exposed scalar, and opaque carrier of a valid
$q$-local datum has width at most $w$, its instantiated endpoints have width at most
\[
 \max\{w,4q+4\}.
\]
No carrier is inspected or duplicated, and same-type replacement
preserves every ancestor rank.
\end{lemma}

\begin{proof}
Proposition~\ref{supp:prop:window-arity} audits the arities. Recorded routing contributes at most
$4q+4$, linear substitution contributes at most $w$, and the suffix has type
$(0,0)$. Lemma~\ref{lem:substitution} bounds the instantiated derivations.
\end{proof}

Call two discrete contexts $(r,w)$-compatible when they use the same ordered
list of linear opaque holes once, substitute the same reduced named cores,
have at most $r$ protected and outer ports, and all cores and instantiated
endpoints have width at most $w$. They must also induce the same
protected-and-outer partition, ordered outer interface, and number of
$D$-only closed components. Rank-zero graph-bearing holes are included among
the opaque cores.

\begin{lemma}[relative spider and port straightening]
\label{lem:relative-spider}
Every $(r,w)$-compatible pair is connected within width $a_\Sigma(r,w)$ for
a computable function $a_\Sigma$, without changing or duplicating a core.
\end{lemma}

\begin{proof}
Put
\[
 L(r,w)=\max\{r,w,4Q_D(r,w)+4\}.
\]
Lemma~\ref{lem:typed-coverage} normalizes every such context by data in
$\mathcal M_{Q_D(r,w)}$ and the nonwidening exposures of Lemma~\ref{supp:lem:carrier-exposure}. The
assumed data fix its endpoint, so both contexts reach it. For each retained
$M$, Theorem~\ref{thm:bk} and proof search select a formal-hole derivation
$\Delta_M$. Lemma~\ref{lem:template-invariant} bounds substitution by
$L(r,w)$ while keeping the scalar suffix opaque. Hence the
monotone closure of
\[
 a_\Sigma(r,w)=\max\Bigl(\{L(r,w)\}\cup
 \{s_{\Delta_M}(L(r,w)):M\in\mathcal M_{Q_D(r,w)}\}\Bigr)
\]
bounds both directions. Lemmas~\ref{supp:lem:renderer-coverage}, \ref{supp:lem:relative-spider}, and \ref{supp:lem:carrier-exposure}
and Proposition~\ref{supp:prop:window-arity} provide the full coverage and audit.
\end{proof}

Put $\operatorname{cup}=\io01\circ\io12:0\to2$ and
$\operatorname{cap}=\io21\circ\io10:2\to0$.

\begin{lemma}[derived bending identities]\label{lem:bending}
For $u:p\to q$ and $C:0\to p+q$, let
$\operatorname{Name}_{p,q}[u]:0\to p+q$ and
$\operatorname{Unname}_{p,q}[C]:p\to q$ be the fixed cup and cap wrappers.
Naming changes the boundary but preserves the underlying unpointed labeled
hypergraph. Yanking gives
\[
 u\xleftrightarrow{*}_{\mathcal B\cup\E}
 \operatorname{Unname}_{p,q}[\operatorname{Name}_{p,q}[u]].
\]
If $p+q\le k$, either wrapper has width at most $\max\{w,4k+4\}$ on an
inserted term of width $w$, and the derivation is computably bounded in $k,w$.
\end{lemma}

\begin{proof}
The yanking diagrams evaluate to $e_1$, so Theorem~\ref{thm:bk} supplies fixed
discrete derivations. Tensor them wirewise, insert $u$, and restore order by
the fixed permutations. Direct evaluation proves the unpointed-hypergraph
claim, and Lemma~\ref{lem:substitution} gives the derivation bound.
\end{proof}

For $u:a\to b$, let $S(u)$ be the set of distinct vertices in the two
boundary sequences of $\val(u)$, ordered by first occurrence.

\begin{corollary}[typed core-frame factorization]
\label{cor:coreframe}
If $\patw(u)\le k$, then $u$ reaches a frame $Q_u[\overline u]$ within
computable width $f_\Sigma(k)$, where the reduced core
$\overline u:0\to|S(u)|$ has width at most $4k+4$. The frame restores the
ordered boundary with repetitions, and every edge occurrence stays in the
core.
\end{corollary}

\begin{proof}
$\mathcal L_{\Sigma,k}=\bigcup_{p+q\le k}(\Sigma\cup D)_{p,q}
\cup\{e_j:j\ge0,\ 2j\le k\}$.
Naming and merging the boundary-equality classes gives a reduced core for
each $\ell\in\mathcal L_{\Sigma,k}$. Completeness and dovetailed proof search
select a finite leaf derivation $\Delta_\ell$ to its fixed frame. Define
\[
 \beta_\Sigma(k)=\max\bigl(\{\patw(v):v\text{ occurs in some }
 \Delta_\ell,\ \ell\in\mathcal L_{\Sigma,k}\}\cup\{k,1\}\bigr),
\]
with monotone closure, and put
$r_k=4(k+1)$, $g(k)=16k+20$, and
\[
 F_\Sigma(k)=\max\{k,g(k),\beta_\Sigma(k),
 a_\Sigma(r_k,g(k))\}.
\]
Induct on the raw syntax. The base cases use the fixed frame. At a sum, correct
the concatenated child orders to the parent order by $P_\oplus$ and put
\[
 \overline u=(\overline v\bsum\overline w)\circ P_\oplus.
\]
At a composition, let $J_\theta$ perform the positional quotient, close
classes absent from the parent boundary, and order the survivors, and put
\[
 \overline u=(\overline v\bsum\overline w)\circ J_\theta.
\]
In both cases the old and new frames have the same protected partition and
outer tuple, so Lemma~\ref{lem:relative-spider} connects them. Router stages
give $\patw(\overline u)\le4k+4$, their protected load is at most $r_k$, and
their width is at most $g(k)$. Sequential child derivations preserve ancestor
ranks, so the monotone closure of $F_\Sigma$ defines $f_\Sigma$, independently
of depth. The full typed induction and arity audit are in Lemmas~\ref{supp:lem:renderer-coverage} and \ref{supp:lem:frame-straightening} and Proposition~\ref{supp:prop:window-arity}.
\end{proof}

\subsection{Parse-aligned compilation}

\begin{lemma}[parse-cut extraction]\label{lem:extract}
Let $t$ be a closed raw expression with $\patw(t)\le k$ and put
$G=\val(t)$. If $E(G)\ne\emptyset$, then after the non-edge leaves are pruned
from the parse tree and unary nodes are suppressed, the remaining tree $T_t$
is an edge-partition tree satisfying
\[
  \operatorname{ewd}_G(T_t)\le k.
\]
\end{lemma}

\begin{proof}
Consider a parse-tree edge above a subexpression $u$. If a vertex has an edge
generated in $u$ and another edge generated outside $u$, then some begin or
end position of $u$ represents that vertex. A vertex internal to $\val(u)$
cannot later be identified with an outside vertex or receive a new incident
edge. This follows by induction through the four parent contexts
$u\circ v$, $v\circ u$, $u\bsum v$, and $v\bsum u$. Hence every atom-boundary
vertex at that cut is represented by one of the at most $k$ ports of $u$.
Pruning and suppressing nodes creates no new cut.
\end{proof}

This lemma only extracts a combinatorial tree. It does not authorize deletion
of a discrete parse piece from the raw expression. Lemma~\ref{lem:typed-coverage}
first normalizes every maximal $D$-only piece, retains its positive-rank router,
and extracts each rank-zero component as a scalar. Pruning occurs only after
that normalization. Lemmas~\ref{supp:lem:edge-free} and \ref{supp:lem:renderer-coverage}, prove this
edge-free normalization and its typed converse coverage.

After discrete preprocessing, let $S$ be the degree-three parse skeleton and
let each edge $e$ carry its bank $P_e$. Give each node $x$ its prescribed
boundary order $\rho_x$, and put $T=S$ and $\rho=(\rho_x)_x$. For a final vertex $v$, the router blocks and
the $v$-pieces form a multigraph $\mathcal R_v$ with its natural map
$\varpi_v:\mathcal R_v\to S$. Atom incidences are flags. Write
$m_v(e)=|\varpi_v^{-1}(e)|$ and let $H_v$ be the minimal subtree spanning the
flags, with $H_v=\varnothing$ when there is no flag. Let $A_e$ be the atom set
below $e$.

\begin{lemma}[parse-routing invariant]\label{lem:routing-invariant}
In a parse-aligned realization, every $\mathcal R_v$ is connected, and
$e\in H_v$ exactly when $v\in\partial_G A_e$. Folding and pruning it to one
strand over each edge of $H_v$, together with the completed-child scalar
schedule, yields $R_{\mathrm{sc}}(G,T,\rho)$, including the isolated-vertex
suffix. Reversing that scalar schedule yields $R(G,T,\rho)$.
\end{lemma}

\begin{proof}
The block-bank incidence components are exactly the final vertex classes, so
each $\mathcal R_v$ is connected. Cutting $S$ at $e$ separates its flags by
the associated atom cut, which gives the boundary criterion above. The
remaining bank therefore has one piece for each vertex of $\partial_G A_e$.
The fixed node routers merge child copies, close non-survivors, and order the
parent bank. Induction from the leaves identifies the nonscalar compiler body
and scalar factors. Lemma~\ref{supp:lem:routing-endpoint} identifies their literal endpoint as
$R_{\mathrm{sc}}(G,T,\rho)$ and gives the reassembly.
\end{proof}

\begin{lemma}[routing folds on a tree]\label{lem:routing-folds}
Local folds and prunes reduce $\mathcal R_v$ to $H_v$, or to one tagged
zero-leg block when $H_v=\varnothing$, without increasing any other
multiplicity. Each move is a two-router window of load at most $12(k+1)$ when
the skeleton cuts and atoms have rank sum at most $k$.
\end{lemma}

\begin{proof}
Fold two edges over the same skeleton edge when they meet one router block.
After no fold remains, connectedness and local injectivity embed
$\mathcal R_v$ as a subtree of $S$ containing $H_v$. Prune unflagged leaves
outside $H_v$. Moves for distinct vertices do not interact. Their full raw
rendering, load audit, termination, and endpoint proof are given in Definition~\ref{supp:def:renderer}, Lemmas~\ref{supp:lem:renderer-coverage}, \ref{supp:lem:routing-endpoint}, and \ref{supp:lem:carrier-exposure}, and
Proposition~\ref{supp:prop:window-arity}.
\end{proof}

\begin{lemma}[finite local schemes]\label{lem:finite-local}
$\mathcal M_q$ is finite and effectively enumerable, with computably bounded
endpoints and a selectable derivation for each pair. The parse compiler uses
$\mathcal M_{12(k+1)}$, and an $r$-bounded four-cell switch uses
$\mathcal M_{12(r+1)}$.
\end{lemma}

\begin{proof}
Bounded record data give finiteness up to ordered renaming. Pointed-hypergraph
equality decides validity, and Theorem~\ref{thm:bk} with dovetailed proof
search selects the derivations. Lemmas~\ref{lem:typed-coverage},
\ref{lem:routing-invariant}, and \ref{lem:routing-folds} give coverage.
Definition~\ref{supp:def:renderer}, Proposition~\ref{supp:prop:renderer-validity}, Lemma~\ref{supp:lem:renderer-coverage},
Proposition~\ref{supp:prop:window-arity}, and Lemma~\ref{supp:lem:carrier-exposure} give the full enumeration and audits.
\end{proof}

The compiler uses rooted plane trees, the fixed cup-letter-cap leaf wrapper,
numbered block order, fixed adjacent-transposition words, left-associated
staging, plane-ordered child sums, and a right-associated isolated-vertex
suffix. A final permutation realizes the prescribed bank order $\rho$. These
choices determine one raw expression.

We now define the compiler recursively. Let
$\emptyset\ne C\subseteq E(G)$ and let $T_C$ be a rooted plane binary
partition tree on $C$. For every node $x$, write $A_x$ for its leaf set, choose
an order $\rho_x=(v^x_1,\ldots,v^x_{b_x})$ of $\partial_G A_x$, where
$b_x=|\partial_G A_x|$, and let $\rho=(\rho_x)_x$. If $x$ is the leaf carrying an occurrence
$e$ labeled by $\sigma:p\to q$, let
$\mathsf{Name}[e]=\operatorname{Name}_{p,q}[e]:0\to p+q$ be the fixed raw
term from Lemma~\ref{lem:bending}. On the $p+q$ incidence positions followed by the
$b_x$ output positions, let $\eta_x$ identify precisely positions denoting the
same vertex of $G$. Set
\begin{equation}\label{eq:compiler-leaf}
 R_x=\mathsf{Name}[e]\circ\operatorname{Can}(\eta_x,0):0\longrightarrow b_x.
\end{equation}
Thus repeated incidences and loops are merged, vertices not in
$\partial_G\{e\}$ are closed, and the survivors are output in order $\rho_x$.

If $x$ has plane-ordered children $y,z$, order the input bank first by
$\rho_y$ and then by $\rho_z$. Let $\eta_x$ be the partition of this input bank
and an output bank ordered by $\rho_x$ in which two positions are equivalent
exactly when they name the same vertex of $G$. Define
\begin{equation}\label{eq:compiler-node}
 J_x=\operatorname{Can}(\eta_x,0):b_y+b_z\longrightarrow b_x,
 \qquad R_x=(R_y\bsum R_z)\circ J_x.
\end{equation}
At the root this gives the \emph{relative reduced realization}
$R_G(C,T_C,\rho)=R_{\operatorname{root}(T_C)}:0\to|\partial_G C|$.
If $G$ is closed, $C=E(G)$, and $i(G)$ is the number of isolated vertices,
set $T=T_C$ and put
\begin{equation}\label{eq:closed-compiler}
 R(G,T,\rho)=R_G(E(G),T,\rho)\bsum Z_{i(G)},
\end{equation}
with the fixed unit suppression when $i(G)=0$. This definition includes every
choice previously hidden in the phrase ``the compiler expression.''

Let $R_{\mathrm{sc}}(G,T,\rho)$ be the renderer's scalar-suffix representative,
which reboxes completed graph-bearing $0\to0$ children as separate opaque
factors and right-associates inherited lists. Equations~\ref{supp:eq:scalar-compiler-recursion}--\ref{supp:eq:scalar-compiler} give its recursion and Lemma~\ref{supp:lem:routing-endpoint} its endpoint and
bridge to $R$.

\begin{lemma}[compiler invariant]
\label{lem:compiler-invariant}
The value of $R_x$ has exactly the edge occurrences in $A_x$, contains
every vertex incident with them, and has end sequence $\rho_x$. A vertex is an
end vertex exactly when it lies in $\partial_G A_x$. Consequently
$\val(R(G,T,\rho))=\val(R_{\mathrm{sc}}(G,T,\rho))=G$ and
$R_{\mathrm{sc}}(G,T,\rho)\xleftrightarrow{*}_{\mathcal B\cup\E}R(G,T,\rho)$.
Equation~\eqref{eq:compiler-node} is the literal rule for $R$. The conversion
introduces no positive-rank bank.
\end{lemma}

\begin{proof}
The leaf router merges equal incidences, closes nonboundary vertices, and
orders survivors. Inductively, $\eta_x$ merges the duplicated child copies and
retains $\partial_G A_x$, while \eqref{eq:closed-compiler} adds the isolated
vertices. Lemma~\ref{supp:lem:routing-endpoint} returns each opaque factor once to its recorded child
address by scalar moves and padded interchange, yielding
\eqref{eq:compiler-node}. Reversal gives extraction.
\end{proof}

\begin{corollary}[compiler compatibility]\label{cor:compiler-compatibility}
Substituting four cell compilers into \eqref{eq:row-template} and
\eqref{eq:col-template} gives the recursively compiled row and column trees.
Transport along a label- and incidence-preserving graph isomorphism leaves the
unnamed raw compiler expression unchanged, including the common tagged suffix.
\end{corollary}

\begin{proof}
Apply \eqref{eq:compiler-node} at the row or column nodes and their parent.
Isomorphisms preserve labels, incidence positions, partitions, orders, and
therefore every selected raw router. The generator-level expansion, endpoint
proof, and two-edge example are in Definition~\ref{supp:def:renderer},
Lemma~\ref{supp:lem:routing-endpoint}, and Section~\ref{supp:sec:two-edge-path}.
\end{proof}

\begin{lemma}[reduced-realization width]\label{lem:realizationwidth}
Let $\emptyset\ne C\subseteq E(G)$ and suppose every node of $T_C$ has ambient
boundary size at most $r$. If every edge occurrence in $C$ has label rank sum
at most $r$, then the fixed compiler parenthesization satisfies
\[
 \patw(R_G(C,T_C,\rho))\le4r+4.
\]
The same bound holds for $R(G,T,\rho)$ and
$R_{\mathrm{sc}}(G,T,\rho)$ when $G$ is closed, $C=E(G)$, and
isolated-vertex scalar factors are added.
\end{lemma}

\begin{proof}
At an internal node, the two child boundaries have at most $2r$ ports in
total. The router only permutes, merges, and deletes these ports. A padded
permutation, merge, or deletion block has input and output ranks at most $2r$,
so its rank sum is at most $4r$. The sum of the child expressions has rank
$(0,b_y+b_z)$ with $b_y+b_z\le2r$. A one-edge name and the cup or cap subexpressions have
rank sum at most $4r+4$. Closed discrete factors have rank $(0,0)$. Induction
over $T_C$ gives the result. Moving and reassociating opaque $0\to0$ factors
creates no positive-rank bank and does not increase the bound.
\end{proof}

Let $t$ be closed with $\patw(t)\le k$, put $G=\val(t)$, assume
$E(G)\ne\varnothing$, and let $T_t$ be the tree from
Lemma~\ref{lem:extract}.

\Needspace{6\baselineskip}
\begin{lemma}[parse-aligned compilation]\label{lem:compile}
For some boundary orders $\rho_t$ and a computable $c_\Sigma$,
\[
 t\equiv_{\mathcal B\cup\E}^{c_\Sigma(k)}R(G,T_t,\rho_t).
\]
\end{lemma}

\begin{proof}
First normalize every maximal $D$-only parse subtree, including a nullary one,
by Lemma~\ref{lem:typed-coverage}. Only then prune the edge-free leaves and
suppress unary nodes. A bottom-up pass records at each processed node
$(A_u,\beta_u,\gamma_u,\lambda_u,C_u,S_u)$. Its invariant says that type and
value are unchanged, and every graph letter occurs once in $C_u$ or a
graph-bearing $S_u$-factor, with its hidden occurrence list retained. Its
protected ports are the $\beta_u$-classes, their final images are the
$\gamma_u$-classes, and liveness is exactly $\lambda_u$. If $v$ is a
descendant of $u$, then $A_v\subseteq A_u$, while incomparable nodes have
disjoint occurrence sets. Hence the $A$-sets of the current processing
frontier are pairwise disjoint. Distinct $\beta_u$-classes merge only under a
positional gluing. In recursive order, $S_u$ contains each completed
graph-bearing $0\to0$ child as an opaque aggregate, then each $D$-only
port-free component. Inherited factors stay separate, and none becomes $z$.

The finite leaf derivations and the canonical discrete router establish the
invariant at leaves. At a binary node, the local renderer takes the disjoint
child partitions, adds exactly the positional pairs for $\circ$, closes the
classes absent from the parent boundary, and preserves linear use. Thus the
block-bank components remain the final-vertex classes. Lemmas
\ref{lem:routing-invariant} and \ref{lem:routing-folds} leave one piece for
each $\partial_G A_e$ and end literally at $R_{\mathrm{sc}}(G,T_t,\rho_t)$.
Lemma~\ref{lem:compiler-invariant} then reassembles $R(G,T_t,\rho_t)$.

For the uniform bound put
\[
 L_0(k)=\max\Bigl(\{4k+4\}\cup
 \{\patw(L),\patw(R):(L,R)\in\mathcal M_{12(k+1)}\}\Bigr),
\]
and, using the selected formal-hole derivations $\Delta_M$, set
\[
 c_\Sigma(k)=\max\Bigl(\{f_\Sigma(k),L_0(k)\}\cup
 \{s_{\Delta_M}(L_0(k)):M\in\mathcal M_{12(k+1)}\}\Bigr).
\]
The same scalar data reassemble the endpoint without enlarging $c_\Sigma(k)$.
After monotone closure, Corollary~\ref{cor:coreframe} and Lemmas
\ref{lem:finite-local} and \ref{lem:substitution} bound every transition by
$c_\Sigma(k)$. Same-type replacement preserves ancestor ranks. The full state
invariant, induction, termination measure, and literal endpoint verification
are in Lemmas~\ref{supp:lem:renderer-coverage}, \ref{supp:lem:routing-endpoint}, and \ref{supp:lem:carrier-exposure}.
\end{proof}

\subsection{Tree interpolation and canonical normal forms}

It remains to remove the dependence on $T_t$. This is done by a rectangular
interpolation, not by taking all intersections at once. The latter collection
need not be laminar. If $P$ is a plane binary tree, write $\operatorname{Lf}(P)$
for its leaf set and
$\mathcal C(P)=\{\operatorname{Lf}(P_x):x\text{ is a node of }P\}$, where
$P_x$ is the subtree rooted at $x$. If
$\varnothing\ne D\subseteq\operatorname{Lf}(P)$, write $P|D$ for the minimal
subtree spanning $D$, with empty branches deleted and unary nodes suppressed.
It inherits its plane order, and $(P|D)|D'=P|D'$ whenever
$\varnothing\ne D'\subseteq D$. For $D=\varnothing$, the notation denotes an
absent entry rather than a tree. Brackets below delete absent entries and
suppress unary nodes.

For a finite graph $G$, let $T,U$ be rooted plane binary partition trees on
$\varnothing\ne C\subseteq E(G)$, of ambient boundary widths $p,q$.

\begin{lemma}[quadrant interpolation]\label{lem:quadrants}
Four-cell switches transform $T$ into $U$ through clusters belonging to $T$,
to $U$, or to their pairwise intersections. Every intermediate boundary has
size at most $p+q$.
\end{lemma}

\begin{proof}
Put
\[
 \mathcal P=\mathcal C(T)\cup\mathcal C(U)\cup
 \{X\cap Y:X\in\mathcal C(T),\ Y\in\mathcal C(U),\
             X\cap Y\ne\varnothing\}.
\]
Laminarity of the two original cluster families makes $\mathcal P$ closed
under nonempty intersection. Define $\mathcal I(P,Q)$ by strong induction on
the common leaf set $D$, under the invariant that their clusters lie in
$\mathcal P$. If $|D|=1$, do nothing. Otherwise write
\[
 P=\langle P_0,P_1\rangle,\quad A_i=\operatorname{Lf}(P_i),\qquad
 Q=\langle Q_0,Q_1\rangle,\quad B_j=\operatorname{Lf}(Q_j),
 \qquad Q_{ij}=A_i\cap B_j.
\]
Recursively transform $P_i$ into $Q|A_i$, omitting empty cells, and replace
the resulting row assembly by
\[
\begin{aligned}
 \mathsf{Row}&=\langle\langle Q|Q_{00},Q|Q_{01}\rangle,
          \langle Q|Q_{10},Q|Q_{11}\rangle\rangle,\\
 \mathsf{Col}&=\langle\langle Q|Q_{00},Q|Q_{10}\rangle,
          \langle Q|Q_{01},Q|Q_{11}\rangle\rangle.
\end{aligned}
\]
Then recursively transform each displayed column into $Q_j$, yielding $Q$.
Every call uses a proper $A_i$ or $B_j$, so it terminates. Intersection
closure preserves the $\mathcal P$ invariant, and Lemma~\ref{lem:intersection}
bounds each mixed cluster by $p+q$.

Every row and column contains an occupied cell, so only two-, three-, and
four-cell switches occur. Corollary~\ref{cor:compiler-compatibility}
identifies their literal compiler endpoints. Empty cells alone are suppressed,
while occupied rank-zero cells remain opaque. Lemma~\ref{supp:lem:quadrant-degeneracies}, verifies all degeneracies and inherited plane orders.
\end{proof}

Call a four-cell assembly $r$-bounded when every occupied reduced cell
realization has defining-tree boundaries at most $r$, and when every row,
column, and common outer staging bank and every occurring edge-label rank sum
are at most $r$.

\begin{lemma}[bounded four-cell switch]\label{lem:cell-switch}
For every fixed finite alphabet $\Sigma$ there is a computable $d_\Sigma$ such
that the row and column versions of every $r$-bounded four-cell assembly are
connected within width $d_\Sigma(r)$.
\end{lemma}

\begin{proof}
The row and column contexts have the same occupied cell holes, the same port
partition, and ordered outer interface. Absent cells and their unary routers
are omitted, while an occupied rank-zero cell remains a formal hole. Prescribe
orders through a temporary vertex order. The resulting cell, staging,
restoration, order-correction, and scalar data form an effective finite set
$\mathcal S_r$ of load at most $12(r+1)$.

For $M\in\mathcal S_r$, treat its holes as distinct letters. Direct evaluation
of the row and column endpoints, or of an order-correction pair, gives the same
pointed hypergraph. Theorem~\ref{thm:bk} and dovetailed proof search give a
derivation $\Delta_M$. Define, with monotone closure,
\[
 d_\Sigma(r)=\max\Bigl(\{s_{\Delta_M}(4r+4):M\in\mathcal S_r\}
 \cup\{4r+4\}\Bigr).
\]
Lemma~\ref{lem:substitution} gives this bound after inserting the linear-use
cores. Lemma~\ref{supp:lem:renderer-coverage}, Proposition~\ref{supp:prop:window-arity}, and
Lemmas~\ref{supp:lem:quadrant-degeneracies} and \ref{supp:lem:carrier-exposure} give renderer coverage, the width audit, all
degeneracies, and literal carrier exposure.
\end{proof}

For a closed graph $G$, let $T,U$ be edge-partition trees of boundary widths
$p,q$, with arbitrary internal orders $\rho,\rho'$.

\begin{corollary}[bounded tree interpolation]\label{cor:tree-interpolation}
Then
\[
 R(G,T,\rho)\equiv_{\mathcal B\cup\E}^{
 d_\Sigma(\max\{p+q,\alpha(G)\})}R(G,U,\rho')
\]
for arbitrary internal orders $\rho,\rho'$.
\end{corollary}

\begin{proof}
Use Lemma~\ref{lem:quadrants}, correcting bank orders before and after each
switch. Corollary~\ref{cor:compiler-compatibility} gives the literal compiler
endpoints. Every staging boundary is at most $p+q$, edge ranks are at most
$\alpha(G)$, and same-type steps preserve ancestor ranks.
\end{proof}

Fix total orders for the finite alphabet and raw serialization. Define
$\operatorname{NF}_k$ on closed classes represented by width-$k$ expressions,
which necessarily satisfy $\alpha(G)\le k$. For an edgeless $n$-vertex class,
put $\operatorname{NF}_k(G)=Z_n$. Otherwise, number the vertices and edge
occurrences in every possible way, enumerate every rooted plane edge-partition
tree of ambient boundary width at most $k$ and every bank order, form
$R(G,T,\rho)$, and erase the auxiliary numbers.
Lemma~\ref{lem:extract} makes this candidate set
finite and nonempty. Its least raw serialization is $\operatorname{NF}_k(G)$.
Isomorphisms permute the presentations, so this literal expression depends
only on the isomorphism class, $k$, and the fixed syntax. It preserves labels,
directions, multiplicities, loops, and isolated vertices.

\begin{theorem}[bounded normal form]\label{thm:boundednf}
For every finite doubly ranked edge alphabet $\Sigma$, there is a computable
nondecreasing $B_\Sigma:\NN\to\NN$ such that every closed raw expression $t$
with $\patw(t)\le k$ satisfies
\[
 \val(\operatorname{NF}_k(\val(t)))=\val(t),
 \qquad
 \patw(\operatorname{NF}_k(\val(t)))\le4k+4,
\]
and
\[
  t\equiv_{\mathcal B\cup\E}^{B_\Sigma(k)}
  \operatorname{NF}_k(\val(t)).
\]
\end{theorem}

\begin{proof}
Put $G=\val(t)$. If $G$ has no edge,
Lemma~\ref{lem:relative-spider} connects $t$ to the fixed scalar tensor within
$a_\Sigma(0,\max\{k,1\})$.

Otherwise $\alpha(G)\le k$. Lemmas~\ref{lem:extract} and \ref{lem:compile}
connect $t$ to a realization of an edge tree $T_t$ of width at most $k$.
Choose the numbered presentation and tree used by
$\operatorname{NF}_k(G)$. Transport their vertex names, edge occurrences,
and bank orders along an isomorphism to a representative of $G$.
Corollary~\ref{cor:compiler-compatibility} shows that erasing those names leaves
exactly the selected raw expression and gives a tree $U$ of width at most $k$
on the actual occurrences. Corollary~\ref{cor:tree-interpolation} connects the
two realizations within $d_\Sigma(2k)$, and
Lemma~\ref{lem:realizationwidth} bounds the target width by $4k+4$.

Take $B_\Sigma(k)$ as the monotone closure of
\[
 \max\{c_\Sigma(k),d_\Sigma(2k),4k+4,
 a_\Sigma(0,\max\{k,1\})\}.
\]
The candidate set is determined by the isomorphism class of $G$, the integer
$k$, and the fixed encoding conventions. Its least serialization is therefore
independent of the input expression and the chosen decomposition.
\end{proof}

\begin{corollary}[bounded equational coherence]\label{cor:boundedcoherence}
\leavevmode\newline
If $s,t$ are closed raw expressions and $\patw(s),\patw(t)\le k$, then
\[
 \val(s)=\val(t)
 \quad\Longleftrightarrow\quad
 s\equiv_{\mathcal B\cup\E}^{B_\Sigma(k)}t.
\]
\end{corollary}

\begin{proof}
Soundness gives the reverse implication. For the forward implication, derive
both expressions to the same normal form and reverse the derivation from $t$.
\end{proof}

\section{Calibration with graph width}\label{sec:calib}

For a graph $F$ whose underlying simple graph has at least two edges, a branch
decomposition is a tree whose internal nodes have degree $3$ and whose leaves
are those simple edges. Removing a tree edge partitions the simple edge set.
Its middle set consists of the vertices incident with graph edges on both
sides. The minimum possible maximum middle-set size is the branchwidth $\bw(F)$.
Robertson and Seymour proved
\cite{RS91}
\begin{equation}\label{eq:rs}
 \bw(F)\le\tw(F)+1
 \le\max\left\{\tfrac32\bw(F),2\right\}.
\end{equation}

\begin{lemma}[branchwidth of cliques]\label{lem:bwkn}
For every $n\ge3$,
\[
 \bw(K_n)=\left\lceil\frac{2n}{3}\right\rceil.
\]
\end{lemma}

\begin{proof}[Proof sketch]
The lower bound is the standard balanced-branch argument. For the upper bound,
partition the vertices into three nearly equal classes and graft three edge
decompositions supported on their pairwise unions. Every middle set then has
size at most $\lceil2n/3\rceil$. Section~\ref{sec:calibration-details}, treats all
degeneracies.
\end{proof}

\begin{proposition}[parse-tree lower bound]\label{prop:lower}
If the underlying simple value of a closed expression $t$ has at least two
edges, then
\[
 \bw(\val(t))\le\patw(t).
\]
Consequently $\bw(F)\le\patw(F)$ for every graph $F$ whose underlying simple
graph has at least two edges.
\end{proposition}

\begin{proof}
Retain one representative $a$-occurrence for every edge of the underlying
simple graph, take the minimal subtree spanning their parse-tree leaves, and
suppress degree-two vertices. The
remaining internal vertices have degree three, so the resulting tree is a
branch decomposition. Every one of its cuts has the same retained-edge
bipartition as a cut induced by some subterm $u$. A graph vertex with
an incident edge on each side must occur in the begin or end interface of
$\val(u)$ by Lemma~\ref{lem:internal}. Its middle set therefore has size at
most the rank sum of $u$, which is at most $\patw(t)$.
\end{proof}

Under Convention~\ref{conv:enc}, let $F$ be a directed multigraph represented
by one $a$-occurrence per directed edge, and let a tree decomposition of its
underlying simple graph have width $k$.

\begin{proposition}[tree-decomposition compiler]\label{prop:upper}
Given the supplied width-$k$ tree decomposition, one can construct in polynomial
time a closed expression for $F$ of width at most $4(k+1)$.
\end{proposition}

\begin{proof}[Construction sketch]
Use a nice rooted decomposition and retain one ordered wire for each bag
vertex. Assign every directed edge occurrence to a bag containing its
endpoints, assigning a loop at $v$ to a bag containing $v$, and introduce the
occurrences separately. The introduce-vertex, introduce-edge, forget-vertex,
and join nodes use fixed copy-merge routers. At a join, the running-intersection
property implies that the vertex sets occurring in the two child subtrees
intersect exactly in the join bag, so the compiled pieces share no unintended
vertex outside that bag. A join is therefore the only step using $2(k+1)$
wires, which gives width at most $4(k+1)$. Section~\ref{sec:calibration-details},
gives the expanded invariant and audit.
\end{proof}

\begin{corollary}[calibration with treewidth]\label{cor:calib}
For every graph $F$ whose underlying simple graph has at least two edges,
\[
 \bw(F)\le\patw(F)\le4(\tw(F)+1)
\]
and
\[
 \tw(F)+1\le\max\left\{\tfrac32\patw(F),2\right\}.
\]
If $\tw(F)\ge2$, then
\[
 \tfrac23(\tw(F)+1)\le\bw(F)\le\patw(F)\le4(\tw(F)+1).
\]
\end{corollary}

\Needspace{9\baselineskip}
For a finite simple graph $F$, choose any orientation and one $a$-occurrence
per edge.

\begin{corollary}[path-decomposition compiler]\label{cor:pwupper}
\[
 \patw(F)\le\patwlin(F)\le2(\pw(F)+1)+2.
\]
\end{corollary}

\begin{proof}
A path decomposition has no join shuffle. Compose its bag gadgets in path
order with the fixed left-comb parenthesization. Section~\ref{sec:calibration-details}, gives the exact rank audit.
\end{proof}

For graphs with at least two underlying simple edges, Section~\ref{sec:calibration-details}, proves $\max\{\bw(F),2\}\le\patwstar(F)\le\patw(F)$ for support width,
which counts distinct final boundary vertices, and gives sharpness examples.

\section{Linear and unrestricted constructions}\label{sec:applications}

Let $F$ be a finite simple undirected graph represented by one arbitrarily
oriented $a$-occurrence per edge. Its linear width $\lw(F)$ minimizes, over
orders of these occurrences, the largest number of vertices incident with
edges on both sides of a proper cut \cite{Tho96,Thi00,BT04}. Put $\lw(F)=0$
when no proper cut exists.

\begin{theorem}[linear-fragment bounds]\label{thm:linear-bounds}
For every such graph $F$ with an edge,
\[
 2\lw(F)-1\le\patw^{\mathrm{el}}(F)\le2\lw(F)+4,
 \qquad \patwlin(F)\ge\tfrac13\lw(F).
\]
Moreover $\patwlin/\patw$ is unbounded on complete binary trees.
\end{theorem}

\begin{proof}[Proof sketch]
An elementary chain orders the edges. Two adjacent block ranks at each proper
cut give $2\lw(F)-1$. Conversely, retain one wire per boundary vertex and at
most two temporary wires to obtain $2\lw(F)+4$. Refining any layered term to
an edge order enlarges its boundary by at most a factor of three. Complete
binary trees have pattern width at most $8$ but unbounded pathwidth and hence
unbounded layered width. Section~\ref{sec:linear}, gives the details.
\end{proof}

\begin{theorem}[clique bounds and machine-verified values]\label{thm:clique-values}
For complete graphs, with pattern width minimized over
$\operatorname{Ori}(K_n)$ as in Convention~\ref{conv:enc},
\[
\begin{gathered}
 2n-3\le\patwlin(K_n)\le2n\quad(n\ge1),
 \qquad \patw(K_n)\le n+1\quad(n\ge3),\\
 (\patw(K_3),\ldots,\patw(K_7))=(4,5,6,6,6).
\end{gathered}
\]
The supplied lane-spiral terms for $5\le n\le24$ have directed values in
$\operatorname{Ori}(K_n)$, and they have width $n+1$.
\end{theorem}

\begin{proof}[Proof sketch]
The linear lower bound follows from the first stage at which a clique vertex
is complete, while the standard live-wire construction gives the upper bound
$2n$. The lane-spiral recursion gives unrestricted width at most $n+1$.
Reproducibility package version 1.1.0 checks the upper witnesses for $K_3$
through $K_7$ and verifies
canonical representative sets whose orbit unions contain every seed, are
transition closed, and omit the accepting orbit at the next smaller width.
Sections~\ref{sec:linear} and \ref{sec:cliques}, gives the proofs and replay instructions.
\end{proof}

In particular $\patwlin(K_n)-\patw(K_n)\ge n-4$.

\section{Conclusion}\label{sec:conclusion}

The bounded-normal-form theorem turns ordinary completeness into a local
resource statement. Equal narrow graph expressions can be connected inside the
fixed finite-graphoid presentation without exposing an interface whose size
depends on the graph. The protected-core construction and intersection
interpolation identify the mechanism behind this locality. The calibration and
clique results show that the resulting parameter reflects familiar structural
width while retaining information about the expression tree.

The result is qualitative in three important respects. The computable
function $B_\Sigma$ is obtained by finite proof search rather than a useful
closed formula. No bound on derivation length is asserted, and the
lexicographic normal-form selection is not claimed to be efficient. Extending
bounded coherence to open values, determining the gap between support width
and pattern width, and finding the asymptotic value of $\patw(K_n)$ remain
open.

\clearpage
\section*{Supplementary proofs}
\label{sec:supplementary-proofs}
The following sections contain the full expanded proofs, constructions,
calibrations, and certificate semantics used in the main text.
\setcounter{section}{0}
\setcounter{subsection}{0}
\renewcommand{\thesection}{S\arabic{section}}
\numberwithin{equation}{section}
\numberwithin{table}{section}
\numberwithin{figure}{section}
\subsection*{Notation and conventions}

Let $\Sigma$ be a finite doubly ranked alphabet.

Let
$D=\{\io21,\io01,\io12,\io10,\pi\}$, with respective ranks
$(2,1)$, $(0,1)$, $(1,2)$, $(1,0)$, and $(2,2)$, and put
$X_\Sigma=\Sigma\cup D$. The set $\Free_{p,q}(X_\Sigma)$ consists of raw
rank-$(p,q)$ expressions built from the letters in $X_\Sigma$, the units
$e_j:j\to j$, composition $\circ$, and parallel sum $\bsum$. The evaluation
map $\val$ interprets a raw expression as its finite graph with ordered begin
and end interfaces. Every $\pi_{m,n}$ below abbreviates the fixed raw
representative $\pi_{m,n}^{\rm raw}$ defined in the main text. All iterated
operations and named permutation words use the main text's fixed raw
bracketing. For a raw expression $t$,
\[
 \patw(t)=\max\{a+b:t'\text{ is a rank-}(a,b)\text{ subterm occurrence of }t\}.
\]
The structural magmoid laws, including interchange, are denoted by
$\mathcal B$. The finite-graphoid equations (A1) through (A15) over $\Sigma$ are
denoted by $\mathcal E=\mathcal E_\Sigma$. The graph-width, clique, encoding,
and saturation sections use the default $\Sigma=\{a\}$ with
$\rank(a)=(1,1)$. Section~\ref{sec:normal-details} allows any finite doubly ranked $\Sigma$. Closed
values are finite directed multigraphs and may have loops. For a simple
undirected graph $F$, let $\operatorname{Ori}(F)$ contain all directed graphs
obtained by orienting each edge exactly once. The notation $\val(t)=F$
abbreviates $\val(t)\in\operatorname{Ori}(F)$, and $\patw(F)$ is minimized
over that set rather than over one fixed orientation. For a directed
multigraph, the underlying simple graph is obtained by forgetting directions,
deleting loops, and replacing each nonempty parallel class by one edge while
retaining all vertices. Classical graph-width parameters are applied to that
underlying simple graph.

\section{Expanded calibration audit with branchwidth and treewidth}
\label{sec:calibration-details}

Throughout this section, branchwidth, treewidth, pathwidth, and linear width of
a directed multigraph $F$ mean the corresponding parameters of its underlying
simple graph. This graph has the same vertex set as $F$. It is obtained by
forgetting edge directions, deleting loops, and replacing each nonempty
parallel class by one edge. Thus isolated vertices are retained. All branch,
tree, and path decompositions below are decompositions of this underlying
simple graph.

Recall that a \emph{branch decomposition} of a graph $F$ whose underlying
simple graph has edge set $E$ with $|E|\ge 2$ is an unrooted tree $T$ whose
internal nodes have degree $3$
and whose leaves are in bijection with $E$. Each tree edge $f$ induces a
bipartition $(E_f,E\setminus E_f)$ of $E$. Its middle set
$\mathrm{mid}(f)$ consists of the vertices incident with edges on both sides.
The width of $T$ is $\max_f|\mathrm{mid}(f)|$. The quantity $\bw(F)$ is the minimum
over all branch decompositions. Robertson and Seymour \cite{RS91} proved
\begin{equation}\label{supp:eq:rs}
\bw(F)\ \le\ \tw(F)+1\ \le\ \max\bigl(\tfrac{3}{2}\bw(F),\,2\bigr).
\end{equation}

\begin{lemma}[branchwidth of cliques]\label{supp:lem:bwkn}
For every $n\ge3$,
\[
  \bw(K_n)=\Bigl\lceil\frac{2n}{3}\Bigr\rceil .
\]
\end{lemma}

\begin{proof}
For the lower bound, fix a branch decomposition $T$ of $K_n$. For a vertex
$v$, let $T_v$ be the minimal subtree of $T$ containing the leaves
corresponding to all edges incident with $v$. For every two vertices $u,v$, the
leaf labeled by the edge $uv$ belongs to both $T_u$ and $T_v$. Hence the
subtrees $T_v$ are pairwise intersecting. By the Helly property for subtrees of
a tree, they have a common intersection. If this common intersection contains a
tree edge $f$, then every vertex lies in $\mathrm{mid}(f)$, so the width is at
least $n$.

Otherwise the common intersection contains no tree edge. It also cannot be a
leaf, since $n\ge3$ and no edge leaf belongs to the incident-edge subtrees of
all vertices. Choose an internal node $x$ in the common intersection. It has three
incident tree edges. Since $n\ge3$, each $T_v$ contains at least two incident
branches at $x$: if all leaves of $T_v$ lay in only one branch, then $x$ would
not belong to the minimal subtree $T_v$. Thus the three incident tree edges at
$x$ are contained, with multiplicity, in at least $2n$ of the subtrees $T_v$.
One of them is contained in at least $\lceil 2n/3\rceil$ of the $T_v$. For that
tree edge $f$, the vertices whose subtrees contain $f$ are exactly the vertices
incident with edges on both sides of the cut of $f$, so
$|\mathrm{mid}(f)|\ge\lceil 2n/3\rceil$.

For the upper bound, partition $V(K_n)$ into three parts $A,B,C$ as evenly as
possible. Then each pairwise union has size at most $\lceil2n/3\rceil$. Split
$E(K_n)$ into three classes: edges inside $A$ together with edges between $A$
and $B$, edges inside $B$ together with edges between $B$ and $C$, and edges
inside $C$ together with edges between $C$ and $A$. Each class is supported on,
respectively, $A\cup B$, $B\cup C$, or $C\cup A$. Hence each uses at most
$\lceil2n/3\rceil$ vertices. All three classes are nonempty because the three
parts are nonempty. For a one-edge class, take a new vertex labeled by that
edge. It will be a leaf after attachment. For each class with at least two
edges, take an arbitrary branch decomposition, subdivide one chosen tree edge,
and use the new degree-$2$ vertex as its attachment point. Add one new cubic
node. Join it directly to the labeled vertex for a one-edge class and to the
attachment point for every larger class. The labeled vertex then has degree
$1$, while every attachment point has degree $3$. Hence the resulting tree is
a branch decomposition of $K_n$.

Every cut within a grafted piece has its middle set contained in the supporting
vertex set of that class. This also holds for either edge created by a
subdivision and for the edge joining the piece to the new cubic node. The three
central middle sets are therefore contained in
$A\cup B$, $B\cup C$, and $C\cup A$. Therefore the width is at most
$\lceil2n/3\rceil$.
\end{proof}

\begin{lemma}[internal nodes stay internal]\label{supp:lem:internal}
Let $t\in\Free_{0,0}(X_\Sigma)$, let $t'$ be a subterm occurrence of rank
$(a,b)$, and let $\theta\colon\val(t')\to\val(t)$ be the canonical map on
nodes induced by evaluating the surrounding context. Call a node $v$ of
$\val(t')$ \emph{internal} if it occurs in neither the begin nor the end
sequence of $\val(t')$. If $v$ is internal then
$\theta^{-1}(\theta(v))=\{v\}$, and every edge of $\val(t)$ incident with
$\theta(v)$ is the $\theta$-image of an edge of $\val(t')$.
\end{lemma}

\begin{proof}
We use induction on the depth of $t'$ below the root of the parse tree. The root
case is trivial. Let $s$ be the parent subterm. If $s=t'\circ u$, then
$V(\val(s))$ is the quotient of $V(\val(t'))\sqcup V(\val(u))$ identifying
the $i$th end node of $\val(t')$ with the $i$th begin node of $\val(u)$.
Since $v$ is not in the end sequence, its class is a singleton and no edge of
$\val(u)$ is incident with it. Since
$\mathrm{begin}(\val(s))=\mathrm{begin}(\val(t'))$ and
$\mathrm{end}(\val(s))=\mathrm{end}(\val(u))$, the image of $v$ is again
internal in $\val(s)$. The cases $s=u\circ t'$, $s=t'\bsum u$ and
$s=u\bsum t'$ are analogous ($\bsum$ performs no identification). The
induction hypothesis applied to $s$ concludes.
\end{proof}

\Needspace{6\baselineskip}
\begin{proposition}[lower calibration]\label{supp:prop:lower}
For every expression $t\in\Free_{0,0}(X_\Sigma)$ whose underlying simple value
has at least two edges, we have $\bw(\val(t))\le\patw(t)$. Hence
$\bw(F)\le\patw(F)$ for every graph $F$ whose underlying simple graph has at
least two edges.
\end{proposition}

\begin{proof}
The parse tree of $t$ is binary. Its leaves are the occurrences of letters
of $X_\Sigma$ and of units. Branchwidth is computed on the underlying simple
graph of $\val(t)$, as in the main text. Hence, before forming the
branch decomposition, discard every loop leaf and retain one representative
$a$-occurrence for each unordered edge of the underlying simple graph. If the
value is already simple, this changes nothing. In the unrooted parse tree,
take the minimal subtree spanning the retained $a$-leaves and suppress every
degree-$2$ vertex. If exactly two leaves are retained, the result is their
joining edge. Otherwise every remaining internal vertex has degree $3$. Its
leaves are in bijection with the edge set of the underlying simple graph of
$\val(t)$, so the resulting tree $T'$ is a branch decomposition. The leaf cut
induced by each edge of $T'$ agrees with the cut induced by some parse-tree edge
$(u,\mathrm{parent}(u))$. If a vertex of the underlying simple graph is incident with a
retained edge generated inside the subterm $t_u$ and with a retained edge
generated outside, then by Lemma~\ref{supp:lem:internal} it is the image of a
non-internal node of $\val(t_u)$, of which there are at most
$a_u+b_u\le\patw(t)$. Discarding loop and duplicate leaves cannot create a new
middle vertex, so every middle set has size at most $\patw(t)$.
\end{proof}

\begin{proposition}[upper calibration]\label{supp:prop:upper}\label{supp:prop:tree-compiler}
Let $F$ be a finite directed multigraph, possibly with loops, represented over
the default rank-$(1,1)$ graph alphabet $\Sigma=\{a\}$ with one
$a$-occurrence for every directed edge occurrence. From a width-$k$ tree
decomposition of its underlying simple graph, one can construct in polynomial
time an expression for $F$ of width at most $4(k+1)$. Thus
\[
\patw(F)\ \le\ 4\,(\tw(F)+1).
\]
For the linear-fragment conclusion, suppose that $F$ is an oriented
representation of a finite simple undirected graph, with exactly one
$a$-occurrence in one arbitrarily chosen orientation for each edge. From a
width-$k$ path decomposition, the same construction yields a linear expression
of width at most $4(k+1)$. Hence
$\patwlin(F)\le 4(\pw(F)+1)$ under this convention.
Corollary~\ref{supp:cor:pwupper} sharpens the constant in that case.
\end{proposition}

\begin{proof}
Start with the supplied width-$k$ tree decomposition of the underlying simple
graph of $F$ and put it, in polynomial time, in nice rooted form
$(T,\{B_x\})$ without increasing its maximum bag size \cite{Klo94}. Its root
and leaves have empty bags, and all bags have size at most $k+1$. Assign each
directed edge occurrence of $F$ to a chosen bag containing its endpoints. A
loop at $u$ is assigned to a chosen bag containing $u$. For every original
nice node $x$, insert immediately above $x$, on the edge from $x$ toward the
root, a chain of same-bag nodes indexed by the occurrences assigned to $B_x$.
Each inserted node has a unique child with the same bag and introduces exactly
one assigned occurrence. The empty root receives no occurrence. This
subdivision preserves the maximum bag size and preserves the path property
when the supplied decomposition is a path. Opposite directed edges, parallel
edges, and loops are created separately, even though they do not appear as
separate edges in the underlying simple graph.

Fix a linear order on $V(F)$ and enumerate each bag in this order. By induction
on $T$ we build, for each node $x$ with $b=|B_x|$, an expression $p_x$ of rank
$(0,b)$. Its value contains precisely the vertices occurring in the subtree at
$x$ and the directed edge occurrences assigned within that subtree. Its end
sequence enumerates $B_x$ in the fixed order. Every permutation router
$\gamma:r\to r$ is specified by a word
$G_1,\ldots,G_s$, where each $G_i=\tau_{r,j_i}$ is the fixed
left-associated $\bsum$-block for one adjacent transposition. The inverse word
is $G_s,\ldots,G_1$. In every displayed chain below, $\gamma$ and
$\gamma^{-1}$ are shorthand for their respective ordered lists of elementary
blocks. We first splice all these lists into the chain and then take one left
comb over the entire resulting block list. In particular, a composite router
is never inserted as a separately parenthesized composition subterm. The empty
router contributes no block. The value of each router lies in the appropriate
permutation component $\PERM_r$.

\emph{Leaf} ($B_x=\emptyset$): $p_x=e_0$.

\emph{Introduce} ($B_x=B_y\cup\{v\}$):
$p_x=p_y\circ(e_{b-1}\bsum\io01)\circ\gamma_{\mathrm{ins}}$, where
$\gamma_{\mathrm{ins}}:b\to b$ is the fixed raw term whose value inserts the
fresh wire at the sorted position of $v$.

\emph{Edge introduction} ($B_x=B_y$). Let $y$ be the unique child of $x$ and
define the inherited pre-edge expression by
$p_x^{\mathrm{pre}}:=p_y$. Suppose first that an edge occurrence $u\to v$ with
$u\ne v$ is assigned to $x$. Choose the fixed raw term
$\gamma_{\mathrm{fr}}:b\to b$ whose value brings $u,v$, in that order, to the
first and second positions. Since $b\ge2$, put
\[
p_x=p_x^{\mathrm{pre}}\circ\gamma_{\mathrm{fr}}\circ(\io12\bsum e_{b-1})
\circ(e_1\bsum a\bsum e_{b-1})
\circ(e_1\bsum\io21\bsum e_{b-2})
\circ\gamma_{\mathrm{fr}}^{-1}.
\]
The copy retains one reference to $u$. The letter creates $u\to z$, where $z$
is its fresh output node, and $\io21$ identifies $z$ with $v$.

If the assigned edge is a loop at $u$, choose the fixed raw term
$\gamma_{\mathrm{fr}}:b\to b$ whose value brings $u$ to the first position and put
\[
p_x=p_x^{\mathrm{pre}}\circ\gamma_{\mathrm{fr}}\circ(\io12\bsum e_{b-1})
\circ(e_1\bsum a\bsum e_{b-1})
\circ(\io21\bsum e_{b-1})
\circ\gamma_{\mathrm{fr}}^{-1}.
\]
Here the final merge identifies the fresh head with the retained copy of $u$.
This formula is valid for $b=1$, with $e_0$ as the empty padding. The three
active blocks have rank sums $2b+1$, $2b+2$, and $2b+1$, respectively, so the
stated width audit is unchanged.

\emph{Forget} ($B_x=B_y\setminus\{v\}$):
$p_x=p_y\circ\gamma_{\mathrm{fr}}\circ(\io10\bsum e_{b})$, where
$\gamma_{\mathrm{fr}}:b+1\to b+1$ is the fixed raw term whose value moves $v$ to the first position while
preserving the relative order of all remaining wires. Thus the output order
is again the fixed order on $B_x$.

\emph{Join} ($B_x=B_y=B_z$):
$p_x=(p_y\bsum p_z)\circ\gamma_{\mathrm{sh}}
\circ(\io21\bsum\cdots\bsum\io21)$, where
$\gamma_{\mathrm{sh}}:2b\to2b$ is the fixed raw term whose value is the perfect shuffle placing
the two references to each bag vertex in adjacent positions, and the $b$
merges $\io21$ identify them. (This is also the step that reconciles
multiply-introduced vertices.)

To verify the value invariant at a join, let $V_y$ and $V_z$ be the sets of
vertices occurring in the decomposition subtrees rooted at the two children.
The running-intersection property gives $V_y\cap V_z=B_x$. Indeed, every
vertex in both child subtrees occurs in bags on both sides of $x$, so the
connectedness of its bag set forces it to occur in $B_x$. Conversely,
$B_x=B_y=B_z$, so every bag vertex occurs in both child subtrees. Hence the
displayed merges identify exactly the two copies of each bag vertex. No nonbag
vertex from one child is identified with a vertex from the other child.

The introduce block has rank $(b-1,b)$, and the forget block has rank
$(b+1,b)$. The three active blocks of either edge gadget have ranks
$(b,b+1)$, $(b+1,b+1)$, and $(b+1,b)$. At a join, the shuffle has rank
$(2b,2b)$ and the parallel merge block has rank $(2b,b)$. These types also
show inductively that every $p_x$ has rank $(0,|B_x|)$. The value invariant
follows from the stated actions on the ordered bag wires. At the empty-bag
root, $\val(p_{\mathrm{root}})=F$.

For the width audit, apply the preceding splicing convention at every node.
The ordered elementary blocks at that node are appended one at a time to its
initial child expression or parallel pair of child expressions. The whole
resulting chain is one left comb. Subterms already inside a child expression
satisfy the bound by induction. Every new subterm is one of the following. A
prefix has rank $(0,b')$ with
$b'\le2(k+1)$. An edge-copy prefix may reach $k+2$, and a join prefix may
reach $2(k+1)$. An elementary block has rank sum at most
$2b+2\le2k+4$. Each adjacent-transposition block of a router acts on at most
$2(k+1)$ wires, so it and all of its $\bsum$-subterms have rank sum at most
$4(k+1)$. A join merge block has rank
$(2b,b)$ and rank sum $3b\le3(k+1)$. Hence the width is at most
$4(k+1)$.

If $T$ is a path, there are no joins. Concatenating the elementary block lists
from successive nodes produces one global ordered block list. Taking its one
global left-comb parenthesization gives a layered linear expression in the
sense of the main text's layered-linear definition.
\end{proof}

\begin{corollary}[calibration with treewidth]\label{supp:cor:calib}
For every graph $F$ whose underlying simple graph has at least two edges,
\[
\bw(F)\ \le\ \patw(F)\ \le\ 4\,(\tw(F)+1)
\]
and
\[
\tw(F)+1\ \le\ \max\left\{\tfrac32\,\patw(F),\,2\right\}.
\]
Consequently,
\[
\tfrac23(\tw(F)+1)-\tfrac43\ \le\ \bw(F)\ \le\ \patw(F)\ \le\ 4(\tw(F)+1).
\]
If $\tw(F)\ge2$, then the sharper multiplicative lower bound
\[
\tfrac23(\tw(F)+1)\ \le\ \bw(F)\ \le\ \patw(F)\ \le\ 4(\tw(F)+1)
\]
holds.
\end{corollary}

\begin{proof}
The inequalities $\bw(F)\le\patw(F)$ and
$\patw(F)\le4(\tw(F)+1)$ are Propositions~\ref{supp:prop:lower}
and~\ref{supp:prop:upper}. Since $\bw(F)\le\patw(F)$, the
Robertson and Seymour comparison \eqref{supp:eq:rs} implies
\[
\tw(F)+1\le\max\left\{\tfrac32\patw(F),2\right\}.
\]
Also
\[
\tw(F)+1\le\max\left\{\tfrac32\bw(F),2\right\}
          \le \tfrac32\bw(F)+2,
\]
which gives
$\tfrac23(\tw(F)+1)-\tfrac43\le\bw(F)$. Finally, if
$\tw(F)\ge2$, then $\tw(F)+1>2$, so the exceptional term in
\eqref{supp:eq:rs} cannot be responsible for the upper bound and
$\tw(F)+1\le\tfrac32\bw(F)$.
\end{proof}

\begin{corollary}[normal forms from tree decompositions]
Let $F$ be a finite graph with $\tw(F)\le t$, and let $\Sigma$ be a fixed
finite alphabet of rank-$(1,1)$ edge labels. Fix an orientation and a
$\Sigma$-labeling $\vec F$, and set $k=4(t+1)$.
Then the compiler of Proposition~\ref{supp:prop:upper} produces a closed expression
$p_{\vec F}$ with exact value $\vec F$ and $\patw(p_{\vec F})\le k$. Its
selected normal form satisfies
\[
 \patw(\operatorname{NF}_k(\vec F))\le 4k+4=16(t+1)+4,
\]
and
\[
 p_{\vec F}\equiv_{\mathcal B\cup\E}^{B_\Sigma(k)}
 \operatorname{NF}_k(\vec F).
\]
In particular, compiler outputs obtained from different width-$t$ tree
decompositions, using the same fixed orientation and labeling, are connected
through the same normal form within $B_\Sigma(k)$.
\end{corollary}

\begin{proof}
The proof of Proposition~\ref{supp:prop:upper} applies verbatim with the letter
$a$ in each edge gadget replaced by the assigned rank-$(1,1)$ label. Apply
the bounded-normal-form theorem of the main text over $\Sigma$ to the
resulting compiler output.
\end{proof}

\Needspace{8\baselineskip}
\begin{corollary}[path decompositions: constant $2$]\label{supp:cor:pwupper}
For every finite simple undirected graph $F$, represented by choosing one
orientation and exactly one $a$-occurrence for each edge,
\[
\patw(F)\ \le\ \patwlin(F)\ \le\ 2\,(\pw(F)+1)+2 .
\]
\end{corollary}

\begin{proof}
Take a width-$k$ path decomposition with $k=\pw(F)$ and apply the nice-form
conversion and same-bag edge-node placement from the proof of
Proposition~\ref{supp:prop:upper}. Neither operation increases the maximum bag
size or introduces a join. Expand every router into its ordered list of
adjacent-transposition blocks, splice all such lists together with the
displayed elementary gadget blocks, and take one global left comb. Prefixes
then have rank
$(0,b')$ with $b'\le k+2$. The transient value $k+2$ occurs inside the
edge gadget. The displayed elementary blocks have rank sum at most
$2b+2\le 2(k+1)+2$. The largest possible contribution is the block
$e_1\bsum a\bsum e_{b-1}$ of the edge gadget. Every constituent block of
$\gamma_{\mathrm{ins}}$ and $\gamma_{\mathrm{fr}}$ acts on at most $k+1$
wires, and all its $\bsum$-subterms have rank sum at most $2(k+1)$. The join
shuffle
$\gamma_{\mathrm{sh}}$ on $2b$ wires, the sole source of the constant
$4$ in the general audit, is absent. The stipulated global block list and
global left comb show that the resulting expression is layered linear in the
sense of the main text's layered-linear definition.
\end{proof}

\begin{remark}[sources of slack]\label{rem:slack}
Sources of constant-factor slack, in decreasing order of importance:
(i) the rank sum $a+b$ counts both sides of a subterm boundary, while a
middle set counts each vertex once, (ii) interface wires may carry
\emph{duplicate} references to the same node (created by $\io12$ and not
yet merged), inflating the syntactic count above the semantic boundary,
(iii) the factor $3/2$ in \eqref{supp:eq:rs} between branchwidth and treewidth
is intrinsic to comparing edge-based with bag-based decompositions, (iv)
the constant $4$ in Proposition~\ref{supp:prop:upper} comes solely from the
join shuffle: Corollary~\ref{supp:cor:pwupper} removes it in the absence of
joins. This is the only point in the proof where the factor $4$ enters.
\end{remark}

\begin{remark}[scope of the computational archive]\label{rem:numerics}
The bounds in Proposition~\ref{supp:prop:upper} and
Corollary~\ref{supp:cor:pwupper} are established by the preceding symbolic width
audits. The supplied computational archive checks the clique witnesses and the
lane-spiral family. It does not claim an independent implementation of the
tree-decomposition or path-decomposition compilers.
\end{remark}

\Needspace{6\baselineskip}
\begin{definition}[support width]\label{def:swid}
Let $t$ be a closed expression and let $s$ be a subterm occurrence of $t$.
Write
\[
\theta_s:\val(s)\longrightarrow\val(t)
\]
for the canonical map induced by evaluating the surrounding context. The
\emph{support boundary} of $s$ in $t$ is
\[
\partial_t(s)=
\theta_s\bigl(\mathrm{begin}(\val(s))\cup\mathrm{end}(\val(s))\bigr)
\subseteq V(\val(t)),
\]
where repeated ports mapping to the same final vertex are counted once. Define
\[
\patwstar(t)=\max_s |\partial_t(s)|,
\]
where $s$ ranges over all subterm occurrences of $t$. For a graph $F$, set
\[
\patwstar(F)=\min\{\patwstar(t):t\in\Free_{0,0}(X_\Sigma),\ \val(t)=F\}.
\]
\end{definition}

\begin{lemma}[support width dominates branchwidth]\label{lem:swidlower}
If the underlying simple value of a closed expression $t$ has at least two
edges, then
\[
\max\{\bw(\val(t)),2\}\le\patwstar(t)\le\patw(t).
\]
Consequently, for every graph $F$ whose underlying simple graph has at least
two edges,
\[
\max\{\bw(F),2\}\le\patwstar(F)\le\patw(F).
\]
\end{lemma}

\begin{proof}
Use the branch-decomposition construction from the proof of
Proposition~\ref{supp:prop:lower}. Every edge of the resulting branch decomposition
comes from a parse-tree edge separating a subterm occurrence $s$ from its
context. If a final vertex lies in the corresponding middle set, then it is
incident with an edge generated inside $s$ and with an edge generated outside
$s$. By Lemma~\ref{supp:lem:internal}, such a vertex must be the image of a begin
or end node of $\val(s)$. Hence the middle set is contained in
$\partial_t(s)$. Therefore every middle set has size at most $\patwstar(t)$, and
$\bw(\val(t))\le\patwstar(t)$. The inequality $\patwstar(t)\le\patw(t)$ follows because
$\partial_t(s)$ counts distinct final vertices represented by at most $a+b$
boundary ports when $s$ has rank $(a,b)$.
Finally, choose a leaf that generates a nonloop edge of the underlying simple
value. Its two endpoint ports map to two distinct final vertices, so its
support boundary has size $2$. Hence $\patwstar(t)\ge2$.
\end{proof}

\Needspace{8\baselineskip}
\begin{problem}[sharp support calibration]\label{prob:sharp}
Determine the graph parameter $\patwstar(F)$ exactly. Lemma~\ref{lem:swidlower}
shows that $\max\{\bw(F),2\}$ is a universal lower bound for graphs whose
underlying simple graph has at least two edges, and the constant $2$ cannot be
omitted. Indeed, for the
two-edge path $P_3$ and the disjoint union $2K_2$,
\[
 \patwstar(P_3)=\patwstar(2K_2)=2,
 \qquad \bw(P_3)=1,
 \qquad \bw(2K_2)=0.
\]
For the upper bounds, orient the path and use
$(\io01\circ a)\circ(a\circ\io10)$. For $2K_2$, take the sum of two copies of
$((\io01\circ a)\circ\io10)$. Direct inspection gives support width $2$, while
the lower bounds follow from the lemma. Is
$\patwstar(F)=\max\{\bw(F),2\}$ for every graph whose underlying simple graph
has at least two edges, or what is the optimal gap?
\end{problem}

\section{The linear fragment: exact calibration and the clique lower bound}
\label{sec:linear}

The normal-form theorem applies to arbitrary parse trees. The present section
shows why restricting those trees to one composition chain changes the width
substantially. We prove that
\emph{every} linear expression for $K_n$, with arbitrary
$\bsum$-layers as blocks in the sense of the layered-linear definition in the
main text, has width
at least $2n-3$ (Theorem~\ref{thm:clique}), that for elementary chains
the pattern width of every finite simple undirected graph is twice its
\emph{linear-width} up to
an additive constant (Theorem~\ref{thm:lwcalib}), and that the linear
and general fragments are infinitely separated
(Theorem~\ref{thm:sep}). These results are structural validation of the
interface measure used in the main text's bounded-normal-form theorem.

Throughout this section, $F$ denotes a finite simple undirected graph. It is
represented syntactically by choosing one orientation for each edge and using
exactly one $a$-occurrence for that edge. Comparisons with $F$ forget only
these chosen orientations. Thus the linear-fragment statements below do not
include loops or parallel edge occurrences. In particular, this convention
applies to $K_n$. A linear expression is written
\begin{equation}\label{eq:chainlin}
t \;=\; (((B_1\circ B_2)\circ B_3)\circ\cdots)\circ B_L ,
\end{equation}
with every block $B_j$ a fixed left-associated nonempty $\bsum$-term over
generators and units. Let $c_j$
denote the rank interface between $B_j$ and $B_{j+1}$. Thus $B_j$
has rank $(c_{j-1},c_j)$, $c_0=c_L=0$, and, since each block is a
subterm occurrence,
\begin{equation}\label{eq:blocklin}
\patw(t)\;\ge\;\max_{1\le j\le L}\,(c_{j-1}+c_j).
\end{equation}
Call $t$ \emph{elementary} if every block contains exactly one
generator, i.e.\ is the corresponding fixed-bracketing term of the form
$e_r\bsum x\bsum e_s$. Formally,
\[
 \patw^{\mathrm{el}}(F)=\min\{\patw(t):t\in\Free_{0,0}(X_\Sigma),\
 t\text{ is elementary},\ \val(t)\in\operatorname{Ori}(F)\}.
\]
Then $\patwlin(F)\le\patw^{\mathrm{el}}(F)$.

\subsection{Stages, references, and pieces}
Fix a finite simple undirected target $F$ and a linear expression
\eqref{eq:chainlin} with $\val(t)\cong F$. Define the formal empty prefix
$t_{\le0}=e_0$, put $t_{\le1}=B_1$, and recursively set
$t_{\le j}=t_{\le j-1}\circ B_j$ for $2\le j\le L$. Each $t_{\le j}$ with
$1\le j\le L$ is a literal subterm occurrence of rank $(0,c_j)$. The term
\emph{stage} $j$ refers to its value, with stage $0$ empty.
The term \emph{cut} $j$ refers to the end sequence of that value, a list of $c_j$ node
references (\emph{refs}). For $1\le j\le L$, let
$\theta_j\colon\val(t_{\le j})\to\val(t)=F$ be the canonical node map
of Lemma~\ref{supp:lem:internal}. Every node of every positive stage maps to a vertex
of $F$. A \emph{piece} of a vertex $u$ at a positive stage $j$ is a node with
$\theta_j$-image $u$, and
$r_j(u)$ is the number of cut-$j$ entries referencing pieces of $u$, so
$c_j=\sum_u r_j(u)$. Set $r_0(u)=0$. The $a$-occurrences inside $B_j$ are the edges
\emph{placed at layer} $j$. Since blocks are $\bsum$-terms, every
factor of $B_j$ draws its inputs directly from cut $j-1$ and writes its
outputs directly to cut $j$: distinct factors consume disjoint sets of
cut-$(j-1)$ entries, and every entry is consumed by exactly one factor
(units included). Say $u$ is \emph{started} at stage $j$ if it has a
piece there, and \emph{complete} if it has exactly one piece and every
edge of $F$ at $u$ is the $\theta_j$-image of an edge of
$\val(t_{\le j})$. The property of being started is monotone in $j$, and the
set of incident edges already placed is monotone.

\begin{lemma}[activity]\label{lem:active}
If $u$ is started and not complete at stage $j$, then $r_j(u)\ge 1$.
\end{lemma}

\begin{proof}
Suppose $r_j(u)=0$. Then every piece of $u$ at stage $j$ is internal in
$\val(t_{\le j})$ (the begin sequence is empty). If $u$ has two or more
pieces they are identified by the remaining context, impossible for
internal nodes by Lemma~\ref{supp:lem:internal}. If $u$ has one piece $P$
but a missing edge, that edge of $\val(t)$ is incident with
$\theta_j(P)=u$ yet is not the image of an edge of $\val(t_{\le j})$,
contradicting the same lemma. Hence $u$ is complete.
\end{proof}

\begin{lemma}[last layer]\label{lem:last}
Suppose that the final degree of $u$ is at least $2$, and that $u$ is complete
at stage $j$ but not at stage $j-1$. Then every edge at $u$ placed at layer
$j$ has $u$ as its tail, and these edges consume pairwise distinct
cut-$(j-1)$ entries referencing pieces of $u$. In particular,
$r_{j-1}(u)\ge\beta$, where $\beta$ is their number.
\end{lemma}

\begin{proof}
Since $u$ has final degree at least $2$, it already has a piece at stage
$j-1$: if it did not, every layer-$j$ edge at $u$ would have to create a fresh
head piece of $u$, and those fresh pieces cannot be merged until a later layer. An
$a$-occurrence creates a fresh node as the head of its edge. If an edge at $u$
placed at layer $j$ had head $u$, this fresh node would be a piece of $u$
distinct from the piece already present at stage $j-1$. The earliest layer at
which two pieces can be merged by $\io21$ is the one after both are referenced
on a cut, so $u$ would still have two pieces at stage $j$, contradicting
completeness. Hence $u$ is the tail. Each such $a$ consumes one
cut-$(j-1)$ entry referencing a piece of $u$, and distinct factors consume
distinct entries.
\end{proof}

\Needspace{6\baselineskip}
\begin{lemma}[edge capacity]\label{lem:capacity}
For every layer $j$ let
\[
P_j=\#\{u: r_{j-1}(u)\ge 1 \text{ and } u \text{ is not complete at stage } j\}.
\]
Then the number of edges placed at layer $j$ is at most $c_{j-1}-P_j$.
\end{lemma}

\begin{proof}
Each edge placed at layer $j$ is an $a$-factor with exactly one input
(its tail). Distinct factors have disjoint inputs, so it suffices to
exhibit, for each $u$ counted by $P_j$, a cut-$(j-1)$ entry consumed by
a non-$a$ factor of $B_j$. Fix such a $u$. By Lemma~\ref{lem:active}
applied at stage $j$ ($u$ is started at $j-1$, hence at $j$, and
incomplete at $j$) we have $r_j(u)\ge 1$. A cut-$j$ entry of class $u$
is produced by its factor as follows: a unit or a $\pi$ outputs the
node of its input entry, $\io12$ outputs its input node twice, $\io21$
outputs the merge of its two input nodes, while $a$ and $\io01$ output
\emph{fresh} nodes. If every cut-$j$ entry of class $u$ were fresh,
no piece of $u$ existing at stage $j-1$ would be referenced at cut $j$.
Each such piece (there is at least one, as $r_{j-1}(u)\ge1$) would be
internal at stage $j$. The argument of Lemma~\ref{lem:active}
contradicts the incompleteness of $u$ at stage $j$. Hence some cut-$j$
entry of class $u$ is produced by a factor
$\varphi\in\{e,\pi,\io12,\io21\}$, and any class-$u$ input entry of
$\varphi$ is a cut-$(j-1)$ entry consumed by a non-$a$ factor. For
distinct $u$ these entries are distinct.
\end{proof}

\subsection{The clique lower bound}

The counting method below yields an additive constant of $3$: the configuration
$(k,g,o)=(2,n-1,1)$ in the notation below satisfies every inequality the method
extracts at budget $2n-3$. Closing the remaining gap to the upper bound
requires charging more than the two-layer window. The explicit
construction below has width $2n$, so the remaining gap is at most three.

Throughout Lemmas~\ref{lem:clique-spike} through \ref{lem:clique-final}, fix
$n\ge3$ and a layered linear expression $t$ with $\val(t)\cong K_n$ and
$\patw(t)\le2n-4$.  Let $T$ be the first stage at which some vertex is complete,
and let $w$ be one such vertex.  Since placing an edge requires an input
reference and $c_0=0$, we have $T\ge2$.  Define
\[
K=\{u:u\text{ is complete at stage }T\},\qquad k=|K|,
\]
\[
S=\{u:u\text{ is started at stage }T-2\},\qquad
G=V(K_n)\setminus S,
\]
and write $g=|G|$, $O=K\cap G$, and $o=|O|$.  Finally set
\[
D=(k-o)(g-o)+o(n-1)-\binom{o}{2}.
\]
The quantity $D$ is the number of edges that must be accommodated in the
critical two-layer window after separating the pairs between $K\setminus O$ and
$G\setminus O$ from the edges incident with the overlap $O$.

\begin{lemma}[first-completion spike]\label{lem:clique-spike}
Under the standing assumptions,
\[
  c_{T-1}\ge n-1,
  \qquad
  c_{T-2}\le n-3,
  \qquad
  c_T\le n-3.
\]
\end{lemma}

\begin{proof}
Let $\beta$ be the number of edges at $w$ placed at layer $T$.  By
Lemma~\ref{lem:last}, these edges have $w$ as tail and consume pairwise
distinct cut-$(T-1)$ entries referencing pieces of $w$, so
$r_{T-1}(w)\ge\beta$. The other $n-1-\beta$ edges at $w$ are placed at layers
at most $T-1$. Hence each corresponding neighbor is started at stage $T-1$ and,
since no vertex is complete there, has a reference by Lemma~\ref{lem:active}.
If $\beta<n-1$, then $w$ itself is started and incomplete at stage $T-1$, so
$r_{T-1}(w)\ge\max(\beta,1)$.  Summing the references gives
\[
  c_{T-1}\ge(n-1-\beta)+\max(\beta,1)\ge n-1.
\]
The block bound \eqref{eq:blocklin} applied to $B_{T-1}$ and $B_T$ gives
$c_{T-2}+c_{T-1}\le2n-4$ and $c_{T-1}+c_T\le2n-4$. Hence
$c_{T-2},c_T\le n-3$.
\end{proof}

\begin{lemma}[late introduction and completion]\label{lem:clique-burst-batch}
Under the standing assumptions,
\[
  k\ge3,
  \qquad
  g\ge3.
\]
\end{lemma}

\begin{proof}
Every vertex is started at stage $T$ because its edge to the completing vertex
$w$ is placed by layer $T$. The vertices not in $K$ are still incomplete at
stage $T$ and therefore have references by Lemma~\ref{lem:active}. Hence
$c_T\ge n-k$. Lemma~\ref{lem:clique-spike} gives $c_T\le n-3$, so $k\ge3$.
Similarly, no vertex is complete at stage $T-2$, so every started vertex at
stage $T-2$ has a reference. Therefore $c_{T-2}\ge |S|$. Together with
$c_{T-2}\le n-3$, this gives $|S|\le n-3$, i.e. $g=n-|S|\ge3$.
\end{proof}

\begin{lemma}[two-layer window]\label{lem:clique-window}
Every edge counted by $D$ is placed in layer $T-1$ or layer $T$.  Moreover
\[
  D\le 2n-4,
  \qquad
  D\le 2k+2g-7.
\]
\end{lemma}

\begin{proof}
Let $u\in K$ and $v\in G$, $u\ne v$.  The edge $uv$ is placed no later than
layer $T$, because $u$ is complete at stage $T$.  If it were placed before
layer $T-1$, then $v$ would be started by stage $T-2$, contrary to $v\in G$.
Thus every edge between $K$ and $G$ is placed in the two-layer window
$T-1,T$.  For $u\in O$, all $n-1$ edges at $u$ are likewise in this window,
because $u$ is not started before stage $T-1$ and is complete at stage $T$.
The expression for $D$ counts these two disjoint edge families.

For the first capacity bound, every placed edge is an $a$-factor with exactly
one input, and distinct factors in a layer consume disjoint input entries.  The
number of edges in the two layers is therefore at most
$c_{T-2}+c_{T-1}$. By \eqref{eq:blocklin}, this is at most $\patw(t)\le2n-4$.
Hence $D\le2n-4$.

For the second capacity bound, use Lemma~\ref{lem:capacity}.  At layer $T-1$,
every vertex in $S$ has a reference at cut $T-2$ and is still incomplete at
stage $T-1$, so $P_{T-1}\ge |S|=n-g$.  Since
$c_{T-2}\le n-3$, layer $T-1$ places at most $g-3$ edges.  At layer $T$, at
most $g$ vertices were not started at stage $T-1$, and at most $k$ of the
started vertices are complete at stage $T$. Hence $P_T\ge n-g-k$. Also
$c_T\ge n-k$ by the argument in Lemma~\ref{lem:clique-burst-batch}, so
$c_{T-1}\le\patw(t)-c_T\le(2n-4)-(n-k)=n+k-4$.  Layer $T$ therefore places at
most $(n+k-4)-(n-g-k)=2k+g-4$ edges.  Summing the two layer bounds gives
$D\le(g-3)+(2k+g-4)=2k+2g-7$.
\end{proof}

\Needspace{6\baselineskip}
\begin{lemma}[overlap bound]\label{lem:clique-overlap}
Under the standing assumptions, $o\le3$.
\end{lemma}

\begin{proof}
A vertex $u\in O$ is started no earlier than stage $T-1$ and complete at stage
$T$, so all $n-1$ edges incident with $u$ are placed in the two-layer window.
Each such edge contributes a distinct $u$-side port: either an input consuming a
reference to $u$, or an output writing a fresh piece of $u$.  The total number
of ports in layers $T-1$ and $T$ is
\[
(c_{T-2}+c_{T-1})+(c_{T-1}+c_T)\le2\patw(t)\le4n-8.
\]
Thus $o(n-1)\le4n-8$, and $o\le3$ for $n\ge3$.
\end{proof}

\begin{lemma}[final squeeze]\label{lem:clique-final}
The standing assumptions are inconsistent.
\end{lemma}

\begin{proof}
Lemmas~\ref{lem:clique-burst-batch} and \ref{lem:clique-overlap} give
$k,g\ge3$ and $o\in\{0,1,2,3\}$. If $o=0$, then
Lemma~\ref{lem:clique-window} gives
$kg\le2k+2g-7$, i.e. $(k-2)(g-2)\le-3$, impossible.  If $o=1$, the same
capacity gives
$(k-1)(g-1)+(n-1)\le2k+2g-7$. Writing $x=k-1\ge2$ and $y=g-1\ge2$, this is
$(x-2)(y-2)\le2-n<0$, impossible.  If $o=2$, the input capacity gives
$(k-2)(g-2)+2n-3\le2n-4$, i.e. $(k-2)(g-2)\le-1$, impossible.  If $o=3$, the
input capacity gives
$(k-3)(g-3)+3n-6\le2n-4$, i.e. $(k-3)(g-3)\le2-n<0$, impossible.  All cases
contradict $k,g\ge3$.
\end{proof}

\begin{theorem}[clique lower bound, linear fragment]\label{thm:clique}
For every $n\ge 1$,
\[
2n-3\;\le\;\patwlin(K_n)\;\le\;2n .
\]
\end{theorem}

\begin{proof}
For $n=1$, the linear term $\io01\circ\io10$ represents $K_1$ and has
width $1$, so both displayed bounds hold.
For the upper bound when $n\ge2$, orient the clique as $i\to j$ for
$0\le i<j\le n-1$ and place the edges in lexicographic order by the head $j$,
then by the tail $i$.
Begin with one live wire for vertex $0$.  Just before placing an edge $i\to j$,
keep exactly the already born vertices that are incident with at least one
unplaced edge.  There are never more than $n-1$ such live vertices: when the
last vertex is born, vertex $0$ has just completed and can be retired, and in
the final batch each subsequent edge completes and retires its tail.

To place $i\to j$, use adjacent transpositions to bring the required wires next
to one another.  If the tail $i$ must remain live, first apply a padded copy
block $\io12$ to it, otherwise use its original wire.  Apply a padded $a$-block
to the chosen tail wire. If $j$ was already live, merge the fresh head into
the live wire of $j$ by a padded $\io21$-block. If $j$ is being born, keep the
fresh head as its live wire.  Finally apply padded $\io10$-blocks to any
vertices whose last incident edge has just been placed.  Each block is a pure
$\bsum$-term over generators and units, so the resulting expression is layered
linear.  During the construction the live count is at most $n$, and it is at
most $n-1$ before every copy.  Hence a padded permutation has rank sum at most
$2n$, a copy has rank sum at most $2n-1$, an $a$-block has rank sum at most
$2n$, and a merge or deletion has rank sum at most $2n-1$. Thus
$\patwlin(K_n)\le2n$.

For the lower bound, the case $n=2$ is immediate. If $n\ge3$ and a
layered linear expression $t$ for $K_n$ had $\patw(t)\le2n-4$, then the
standing hypotheses of Lemmas~\ref{lem:clique-spike} through
\ref{lem:clique-final}
would apply, contradicting Lemma~\ref{lem:clique-final}.  Therefore every such
expression has width at least $2n-3$.
\end{proof}

\begin{remark}[mechanism]\label{rem:mechanism}
The proof exhibits the doubling in the layered setting, where the width
profile may jump arbitrarily within a block. The first completion forces a
spike $c_{T-1}\ge n-1$. A budget of $2n-c$ then forces a sharp decrease on
both sides of the spike. This requires at least $c-1$ vertices to be introduced
at the last possible stage and at least $c-1$ completions in the following
stage (the proof uses $c=4$). The connecting edges must all lie in a two-layer
window whose total throughput is linear, in $k+g$ by the capacity lemma and in
$n$ by the input count, while the demand $kg$ is bilinear. Any overlap between
the newly introduced and newly completed sets adds a near-complete star to the
same window.
\end{remark}

\subsection{Elementary chains and linear-width}
Recall the \emph{linear-width} $\lw(F)$ of a finite simple undirected graph
$F$: the minimum,
over orderings $f_1<\cdots<f_m$ of $E(F)$, of the maximum number of
vertices incident both with $\{f_1,\dots,f_i\}$ and with the remaining
edges, where $1\le i<m$. The maximum is zero when there is no proper cut.
The parameter was introduced by Thomas \cite{Tho96}. See
\cite{Thi00,BT04} for its systematic study.

\begin{theorem}[linear fragment calibration]\label{thm:lwcalib}
For every finite simple undirected graph $F$ with at least one edge,
\[
2\,\lw(F)-1\;\le\;\patw^{\mathrm{el}}(F)\;\le\;2\,\lw(F)+4,
\qquad
\patwlin(F)\;\ge\;\tfrac13\,\lw(F).
\]
Lemma~\ref{lem:lwpw} gives
$\pw(F)-1\le\lw(F)\le\pw(F)+1$. Consequently,
$\patw^{\mathrm{el}}(F)=2\,\pw(F)+\Theta(1)$.
\end{theorem}

The proof occupies the next four lemmas.

\begin{lemma}[lower calibration, elementary]\label{lem:el}
For a finite simple undirected graph $F$, every elementary linear expression
$t$ with $\val(t)\cong F$,
$|E(F)|\ge1$, satisfies $\patw(t)\ge 2\lw(F)-1$.
\end{lemma}

\begin{proof}
Each block contains one generator, so the edge placements occur in
distinct blocks and induce an edge ordering of $F$. If $u$ is incident
with edges on both sides of the cut after the block placing $f_i$, then
$u$ is started and incomplete at that stage, so it has a boundary reference by
Lemma~\ref{lem:active}. Hence
$c_j\ge|\partial_i|$ for the corresponding boundaries, and
$\max_j c_j\ge\lw(F)$. Choose $j_0$ with
$c_{j_0}=\max_j c_j$. Since $c_0=c_L=0$ and $F$ has an edge, we may take
$1\le j_0<L$. A block with a single generator changes the interface size by
at most one. Therefore the block $B_{j_0}$ has
$c_{j_0-1}+c_{j_0}\ge 2c_{j_0}-1\ge 2\lw(F)-1$. The claim follows from
\eqref{eq:blocklin}.
\end{proof}

\Needspace{6\baselineskip}
\begin{lemma}[upper calibration, elementary]\label{lem:elupper}
$\patw^{\mathrm{el}}(F)\le 2\,\lw(F)+4$ for every finite simple undirected
graph $F$ with at least one edge.
\end{lemma}

\begin{proof}
Fix an ordering of $E(F)$ with all boundaries $\le\lw=:w$, choose one
orientation for every edge, and compile the resulting edge occurrences by
elementary blocks. Between gadgets maintain one live wire per boundary vertex
of the processed prefix, in a fixed order. For the next oriented representative
$f=u\to v$, create $u$ by $\io01$ if it is new. If $u$ has an edge after $f$,
copy its wire by $\io12$, and apply $a$ to the spare $u$-wire. If $v$ already
exists, use the fixed stable adjacent-transposition word to move the fresh head
immediately next to the live $v$-wire, then merge them by $\io21$. Delete the
wires of vertices leaving the boundary and use fixed adjacent transpositions
to restore the chosen order of the surviving boundary wires.

Every edge-emission or routing block acts on at most $w+2$ wires, so its rank
sum is at most $2w+4$. Creation and copy blocks have rank sum at most $2w+3$.
A merge or deletion may also have $w+2$ input wires, but its rank sum is then
$2(w+2)-1=2w+3$. Post-deletion restoration acts on at most $w$ wires.
Therefore the whole elementary chain has width at most $2w+4$.

It remains to retain isolated vertices. For every isolated vertex prepend the
two consecutive elementary blocks $\io01:0\to1$ and $\io10:1\to0$. Their
composite is the scalar $z=\io01\circ\io10$, and composition of closed
rank-zero terms is disjoint union. These blocks have rank sum one and do not
change any later interface. Prepending one such pair for every isolated
vertex therefore produces exactly $F$ and preserves the bound $2w+4$.
\end{proof}

\begin{lemma}[linear-width of cliques]\label{lem:lwkn}
$\lw(K_2)=0$, while $\lw(K_n)=n-1$ for every $n\ge3$.
\end{lemma}

\begin{proof}
\emph{Lower} ($n\ge3$). Let $i^\ast$ be minimal such that some vertex
$w$ has all its edges among $f_1,\dots,f_{i^\ast}$, and let $v_0$ be
the other endpoint of $f_{i^\ast}$. At cut $i^\ast-1$ every
$v\notin\{w,v_0\}$ has its edge $vw$ on the left and, being
unfinished, an edge on the right. The vertex $w$ has $n-2\ge1$ edges on the left
and $f_{i^\ast}$ on the right. Hence the boundary has size $\ge n-1$.
For $n=2$ there is one edge and every edge-order boundary is empty.
\emph{Upper} ($n\ge3$). Regard the vertices as $\{0,\dots,n-1\}$ and order
all edges $uv$ ($u<v$) lexicographically by $v$, then by $u$. During a
nonfinal batch into $v<n-1$, before its first edge the boundary is contained
in $\{0,\dots,v-1\}$, while after its first edge it is exactly
$\{0,\dots,v\}$, of size $v+1\le n-1$. At the start of the final batch into
$n-1$, the boundary is $\{0,\dots,n-2\}$. After the first $k$ edges of that
batch, where $1\le k\le n-2$, it is $\{k,\dots,n-1\}$, again of size at most
$n-1$. Thus every proper edge-order boundary has size at most $n-1$.
\end{proof}

\begin{lemma}[linear-width versus pathwidth]\label{lem:lwpw}
$\pw(F)-1\le\lw(F)\le\pw(F)+1$ for every finite simple undirected graph $F$
with $E(F)\neq\emptyset$.
\end{lemma}

\begin{proof}
($\pw\le\lw+1$.) From an optimal ordering set
$X_i=\partial_{i-1}\cup f_i$ (endpoints). Each edge $f_i\subseteq
X_i$, and each vertex appears in the consecutive bags between its first
and last edge. Also $|X_i|\le\lw+2$. Append one singleton bag for each
isolated vertex. These bags cover all remaining vertices and do not change the
width or any interval belonging to a nonisolated vertex.
($\lw\le\pw+1$.) From an optimal path decomposition
$(X_1,\dots,X_r)$ order the edges by the first bag containing them.
A vertex $u$ crossing an edge-order cut within the batch assigned to bag $i$ has an edge
first covered by a bag $\le i$ and one first covered by a bag
$\ge i$. Thus $u\in X_j\cap X_{j'}$ with $j\le i\le j'$, whence
$u\in X_i$ by the interval property. The boundaries are $\le|X_i|\le
\pw+1$.
\end{proof}

\begin{lemma}[layered refinement]\label{lem:refine}
For a finite simple undirected graph $F$, every linear expression $t$ with
$\val(t)\cong F$ satisfies
$\lw(F)\le 3\,\patw(t)$.
\end{lemma}

\begin{proof}
Refine the layer order to an edge order, breaking layers arbitrarily.
A refined edge-order cut within layer $j$ has boundary contained in the vertices
crossing the preceding layer boundary $j-1$ (at most $c_{j-1}$, by
Lemma~\ref{lem:active}) together with the endpoints of the edges
placed at layer $j$ (at most $2c_{j-1}$, each $a$ having one input),
so every boundary is at most $3c_{j-1}\le3\patw(t)$.
\end{proof}

\begin{proof}[Proof of Theorem~\ref{thm:lwcalib}]
Combine Lemmas~\ref{lem:el}, \ref{lem:elupper}, \ref{lem:lwpw} and
\ref{lem:refine}.
\end{proof}

\begin{remark}\label{rem:swidlin}
The precise doubling statement proved by Theorem~\ref{thm:lwcalib} is
\[
 \patw^{\mathrm{el}}(F)=2\,\lw(F)+O(1),
\]
with additive error bounded by the displayed constants in that theorem. In
the lower-bound proof, every edge-order boundary is represented in the support
of the corresponding prefix. The upper construction keeps one live reference
per boundary vertex between gadgets. Thus the factor two is the two-sided rank
cost of elementary chains. It does not assert
$\patw(F)=2\patwstar(F)+O(1)$ for unrestricted expressions.
\end{remark}

\subsection{Separation of the fragments}

\begin{theorem}[separation]\label{thm:sep}
For the complete binary tree $T_d$ of height $d\ge1$,
$\patw(T_d)\le 8$, while
$\patwlin(T_d)\ge\frac13(\lfloor d/2\rfloor-1)\to\infty$. Hence
$\patwlin/\patw$ is unbounded on trees.
\end{theorem}

\begin{lemma}\label{lem:pwtree}
$\pw(T_d)\ge\lceil d/2\rceil$.
\end{lemma}

\begin{proof}
We show $\pw(T_d)\ge\pw(T_{d-2})+1$ for $d\ge2$. Together with
$\pw(T_1)=1$ and the trivial bound $\pw(T_0)\ge0$, this proves the claim.
The tree $T_d$
contains three pairwise disjoint copies $H_1,H_2,H_3$ of $T_{d-2}$
(grandchild subtrees), and $T_d\setminus H_c$ is connected for each $c$.
Fix an optimal path decomposition. For a connected subgraph $C$, the bags meeting $C$ form
an interval $I(C)$. Among the three intervals $I(H_c)$ one is
contained in the hull of the union of the other two. If some interval
realizes both the leftmost left endpoint and the rightmost right
endpoint, either other interval is contained in it. Otherwise the
interval realizing neither extreme is contained in the hull of the two
realizing them. Say $I(H_2)\subseteq\mathrm{hull}(I(H_1)\cup I(H_3))$
and let $C=H_1\cup H_3\cup P$, where $P$ is a connecting path in
$T_d\setminus H_2$. Then $C$ is connected and
$I(C)\supseteq I(H_2)$. Hence every bag meeting $H_2$ contains a vertex
of $C$, disjoint from $H_2$. Restricting all bags to $V(H_2)$ yields a
path decomposition of $H_2$ of width $\le\pw(T_d)-1$.
\end{proof}

\begin{proof}[Proof of Theorem~\ref{thm:sep}]
$\patw(T_d)\le4(\tw(T_d)+1)=8$ by Proposition~\ref{supp:prop:upper}. For the
lower bound, Lemmas~\ref{lem:refine}, \ref{lem:lwpw} and
\ref{lem:pwtree} give
$\patwlin(T_d)\ge\frac13\lw(T_d)\ge\frac13(\pw(T_d)-1)\ge
\frac13(\lfloor d/2\rfloor-1)$.
\end{proof}

The proof uses only the qualitative fact that the optimal linear width is
unbounded, via pathwidth and linear-width, while the tree-shaped construction
stays at width $8$.

\section{Unrestricted clique bounds and certified values}\label{sec:cliques}

Proposition~\ref{supp:prop:finiteclosure} gives a finite decision procedure at each
fixed width. The next theorem uses finite inductive-invariant state sets as lower-bound
certificates and explicit raw expressions as upper-bound certificates. These
computations are validation of the syntax and of the closure construction.
\begin{theorem}[machine-verified small clique values]\label{thm:smallcliques}
The unrestricted pattern width of the five cliques $K_3$ through $K_7$ is
\[
\begin{aligned}
 \patw(K_3)&=4, &
 \patw(K_4)&=5, &
 \patw(K_5)&=6,\\
 \patw(K_6)&=6, &
 \patw(K_7)&=6.
\end{aligned}
\]
All five equalities are machine-verified in reproducibility package version
1.1.0. The lower certificates use
\path{pattern-width-independent-certificates-v2} format and
\path{code/check_certificate.py}. The upper witnesses use
\path{pattern-width-upper-witnesses-v1} format and
\path{code/verify_certificates.py}. The version 1.1.0 manifests pin both
checker files by SHA-256.
\end{theorem}

\begin{proof}
The upper bounds are explicit raw expressions. A width-\(4\) expression for
\(K_3\) is
\[
 \io01\circ\Bigl(\bigl((\io12\circ(a\bsum(a\circ a)))\circ\io21\bigr)
 \circ\io10\Bigr).
\]
The width-\(5\) and width-\(6\) expressions for \(K_4\) and \(K_5\) are
listed in \texttt{SMALL\_CLIQUE\_WITNESSES.tex}. The width-\(6\) witnesses for
\(K_6\) and \(K_7\) are stored in
\texttt{code/k6\_exact\_engine\_witness.txt} and
\texttt{code/k7\_exact\_engine\_witness.txt}. The companion evaluators check
their rank, pattern width, vertex count, and edge set. In particular, the
\(K_7\) term has rank \((0,0)\), width \(6\), seven vertices, and all \(21\)
unordered edges exactly once.

For each lower bound, the supplied file contains a finite set $R$ of canonical
states such that its orbit union contains every seed, is closed under every
transition within the rank cap, and omits the accepting orbit. All five files
are in \path{certificates/independent/}.
\begin{table}[ht]
\caption{Independent inductive-invariant certificates for clique lower bounds.}
\label{tab:certificates}
\centering
\small
\begin{tabular}{@{}ccl@{}}
\hline
claim & invariant states & certificate file\\
\hline
\(\patw(K_3)>3\) & 43 & \texttt{patw\_K3\_W3.json}\\
\(\patw(K_4)>4\) & 371 & \texttt{patw\_K4\_W4.json}\\
\(\patw(K_5)>5\) & 6117 & \texttt{patw\_K5\_W5.json}\\
\(\patw(K_6)>5\) & 5457 & \texttt{patw\_K6\_W5.json}\\
\(\patw(K_7)>5\) & 5435 & \texttt{patw\_K7\_W5.json}\\
\hline
\end{tabular}
\end{table}

\noindent The standalone checker \path{code/check_certificate.py} does not
import the search engine.
It independently checks that the exported states are canonical and verifies
well-formedness and the width cap, distinctness of their symmetry orbits,
inclusion of the canonical representative of every generator state, closure
modulo symmetry under sum and composition, and absence of the accepting
orbit. Seed inclusion and closure
imply that the union of the orbits represented by $R$ contains the full
reachable state set. Proposition~\ref{supp:prop:finiteclosure} therefore proves the
stated strict lower bounds.

The quotient used for the clique certificates is sound. Vertex permutations
preserve the target and the transition system. Begin/end duality exchanges
\(\io01\) with \(\io10\) and \(\io12\) with \(\io21\), fixes the remaining
discrete structure, reverses composition order, and reverses only the
orientation forgotten by the main text's underlying-simple convention.
As a positive control for acceptance detection, the supplied $K_4,W=5$
control file already represents the accepting orbit. The independent checker
detects this pre-existing orbit and rejects the file as a lower-bound
certificate with \texttt{FAIL (A)}. As an integrity check, deleting a required
orbit representative breaks closure.
The matching
upper and lower bounds prove all five equalities.
\end{proof}

\begin{theorem}[clique upper bound]\label{thm:cliqueupper}
$\patw(K_n)\le n+1$ for all $n\ge 3$.
\end{theorem}

\begin{proof}
For $n\in\{3,4\}$ use Theorem~\ref{thm:smallcliques}. For $n\ge 5$ we
exhibit a uniform family of expressions called the \emph{lane spiral}. It
packages edge placement into bounded composite blocks and distributes the
remaining live references between descending and ascending lanes.

\emph{Gadgets.} The construction uses three kinds of composite block. Under
the gadget-bundling convention, padding does not enter their interiors. The
first is the edge gadget
$g_2=(\io12\bsum e)\circ(e\bsum a\bsum e)\circ(e\bsum\io21)$ of rank
$(2,2)$. It places an edge between the classes of two adjacent live
wires and preserves both. The second is the closing edge gadget
$g_\times=(a\bsum e)\circ\io21$ of rank $(2,1)$, which places the edge
and consumes its first wire by merging the head into the second. The third
kind consists of the copy gadgets $\io12$ of rank $(1,2)$ and
$g_3=\io12\circ(e\bsum\io12)$ of rank $(1,3)$.
All their proper subterms
have rank sum at most $6\le n+1$.

\emph{Plan.} Put $m=n-1$ and $h=\lfloor(m-3)/2\rfloor\ge 0$. The
vertices $v_0,\dots,v_m$ are created along a spine: $v_0$ by $\io01$,
and $v_j$ ($j\ge1$) by an $a$ applied to the live wire of $v_{j-1}$,
which simultaneously places the path edge $\{v_{j-1},v_j\}$. Call the
remaining pairs $\{i,j\}$, $j\ge i+2$, \emph{chords}, and schedule each
chord exactly once:
\[
\{i,j\}\ \text{is \emph{descending} if } i\le h \text{ and } i+j\le
2h+2,\qquad\text{\emph{ascending} otherwise.}
\]
Thus $v_i$ with $i\le h$ owns the descending targets
$j\in[i+2,\,2h+2-i]$ (a nonempty, nested family of intervals), and the
ascending sources of $v_j$ are
$\sigma(j)=\{i\le j-2: i>h \text{ or } i+j>2h+2\}$. This set is nonempty exactly for
$j\ge h+3$, with minimum $S(j)=\max(0,\,\min(h,2h+2-j)+1)$.

To make the recursion below a literal raw term, every recursive call carries
its current ordered wire layout. A two-input block $g:r\to s$ has an ordered
pair $(x,y)$ of selected wires. Its input router is the fixed word that moves
$x$, one adjacent transposition at a time, until it is immediately left of
$y$, while preserving the relative order of every other wire. The padded copy
of $g$ is applied there, and its declared output word replaces $(x,y)$. The
resulting word is the successor layout. No implicit inverse router or sorting
step is applied. A one-input block similarly replaces its selected wire at its
current position. This padding is a raw term over $\pi,e_1$, and the generators
in $g$.

Let $\mathbf q_l$ be the rider subword entering level $l\ge1$. Its underlying
set is
\[
 \{v_i:0\le i\le h,\ i<l\le2h+2-i\},
\]
and the active wire $v_{l-1}$ follows $\mathbf q_l$. The order of
$\mathbf q_l$ is inherited recursively and need not be increasing. During
$A_j$ the active wire remains rightmost. A rider consumed by $g_\times$ is
deleted, while a rider preserved by $g_2$ remains immediately before the
active wire. Consequently the final copy acts on the rightmost wire. Its
declared outputs already form the successor suffix, so its input and
post-output permutation words are empty.

Put $d_0=0$, $u_{-1}=0$, and, for the indicated ranges, set
\begin{align}
d_l&=1+\max\bigl(0,1+\min(h,l-1,2h+2-l)\bigr)
       &&(1\le l\le m),\label{eq:lane-d}\\
u_l&=m-\max(l,h+2,2h+2-l)
       &&(0\le l\le m-1).\label{eq:lane-u}
\end{align}
Here $d_l$ is the number of wires entering level $l$, and $u_l$ is the
number of returning up-lanes just before the outbound work at level $l$.

\noindent\emph{Typed lane recursion.}
For $0\le j<m$ there are explicitly padded raw terms
\[
 A_j:d_j\longrightarrow d_{j+1}+1,
 \qquad
 B_j:u_j+1\longrightarrow u_{j-1},
\]
and there is a raw term $A_m:d_m\to u_{m-1}$. Defining
\[
 R_m=A_m,
 \qquad
 R_j=A_j\circ(R_{j+1}\bsum e_1)\circ B_j
 \quad(0\le j<m)
\]
gives $R_j:d_j\to u_{j-1}$. In particular, $R_0:0\to0$.

Indeed, the term $A_j$ is the following left-associated list of padded blocks.
First
create $v_0$ by $\io01$ when $j=0$, or create $v_j$ from the active wire of
$v_{j-1}$ by $a$ when $j>0$. Next process the indices
\[
 D_j=\{i:0\le i\le h,\ i+2\le j\le2h+2-i\}.
\]
The index satisfying $j=2h+2-i$, if present, is processed first with
$g_\times$ and its rider is retired. Process every remaining index of $D_j$
in increasing order with $g_2$. If $j\le h$, finish with $g_3$, whose outputs
are the new rider, the active wire, and the spare. If $h<j<m$, finish with
$\io12$, whose outputs are the active wire and the spare. At $j=m$ make no
copy. If $\sigma(m)$ is nonempty, the last active wire is relabeled as the
$v_m$ up-lane. Otherwise it is closed by $\io10$. This completely defines
$A_j$ as a raw term.

For $0\le j<m$, an equivalent explicit description of the rider order is as
follows. Put $C_j=D_j\cap\{2h+2-j\}$. Starting with $\mathbf q_j$, remove the
riders indexed by $D_j$, append the riders indexed by $D_j\setminus C_j$ in
increasing order, and append $v_j$ exactly when $j\le h$. The resulting word
is $\mathbf q_{j+1}$. This recurrence starts with $\mathbf q_0$ empty.

The underlying set of $\mathbf q_l$ is exactly the stated rider set. Adding
the rightmost active wire gives \eqref{eq:lane-d}. The preceding recurrence
deletes precisely the retiring rider, retains every continuing rider, and adds
$v_j$ precisely when it has a descending target. For $0\le j<m$, after the
final copy, removal of the parked spare therefore leaves the required
descending layout at level $j+1$. At the base level the remaining wires are
exactly the $u_{m-1}$
returning lanes. This proves the types of the $A_j$ and the rightmost-active
layout invariant.

For $j<m$, define $B_j$ by processing in increasing order every
$j'\ge j+2$ for which $j\in\sigma(j')$. Apply a padded $g_2$ to the spare of
$v_j$ and the $v_{j'}$ up-lane. If $j=S(j')$, immediately close that up-lane
by a padded $\io10$. After all such chords, relabel the spare as the $v_j$
up-lane when $\sigma(j)$ is nonempty, and close it by $\io10$ otherwise.
The up-lanes present are precisely the $v_{j'}$ with
$S(j')\le j<j'$. Their number is \eqref{eq:lane-u}. Hence
$B_j:u_j+1\to u_{j-1}$. The middle term has type
$R_{j+1}\bsum e_1:d_{j+1}+1\to u_j+1$, so induction gives the displayed
type of $R_j$.

For the illustrative case $K_7$ one has $m=6$ and $h=1$. The path edges are
placed on the spine in the order
$01,12,23,34,45,56$. The descending chords are
$02,03,04,13$, and the ascending chords are
$05,06,14,15,16,24,25,26,35,36,46$.
This uniform lane-spiral realization has width $8$. It is distinct from the
optimized width-$6$ witness certified in Theorem~\ref{thm:smallcliques}.
The central part of the schedule is visible in the following width audit:
\begin{table}[ht]
\caption{Width audit for the central levels of the $K_7$ lane spiral.}
\label{tab:k7-width-audit}
\centering
\small
\begin{tabular}{@{}ccccc@{}}
\hline
 $j$ & $d_{j+1}$ & $u_j$ &
 $2+d_{j+1}+u_j$ & role\\
\hline
0&2&2&6&entry\\
1&3&3&8&central level\\
2&3&3&8&central level\\
3&2&3&7&exit from center\\
4&1&2&5&tail\\
5&1&1&4&tail\\
\hline
\end{tabular}
\end{table}
The general formulas below are the same calculation with $m$ and $h$ left
symbolic. The maximum is attained at the central levels.

\emph{The expression.} Use the raw terms in the typed recursion above. The
single $e_1$ factor in
$R_{j+1}\bsum e_1$ carries the spare wire of $v_j$ while the remaining
wires traverse the deeper region. The preceding type calculation proves
directly that the recursively defined expression $R_0$ is closed.

\emph{Value.} No two vertex classes are ever merged: every $\io21$
occurs inside $g_2$ or $g_\times$, merging a freshly created edge head
into an existing wire. Each path edge is placed once at a birth, each
chord once by its scheduled gadget ($g_2$ or $g_\times$ for descending
chords during $A_j$, $g_2$ for ascending chords during $B_i$), and every
wire is eventually consumed ($g_\times$, the final merges, or $\io10$)
only after its duties are complete. Hence $\val(R_0)$ has vertex set
$\{v_0,\dots,v_m\}$ and contains each of the $\binom{n}{2}$ edges exactly once:
$\val(R_0)\cong K_n$.

\emph{Width.} The counts $d_l$ and $u_l$ in
\eqref{eq:lane-d} and \eqref{eq:lane-u} are unimodal. Their maxima are
$h+2$ and $m-h-2$, attained on $l\in\{h+1,h+2\}$ and
$l\in\{h,h+1\}$ respectively. Every subterm of
$R_0$ now falls into one of five classes. (a) Proper subterms of the
gadgets: rank sum $\le 6\le n+1$. (b) Padded blocks
$e_p\bsum g\bsum e_q$ and the intermediate nodes of their
$e$-wrappers: a block applied at layout width $W$ has rank sum
$2W+\delta$ with $\delta=-1,0,+1,+2$ for
$g\in\{g_\times,\io10\},\{a,\pi,g_2\},\io12,g_3$ respectively. On the
nonexpanding blocks, every stable input router acts on at most $W$ wires, and
no post-output router is inserted. On the expanding final copies, both routing
words are empty because the active wire is rightmost and the declared outputs
already form the successor suffix. Thus all permutation subterms satisfy the
same bounds.
On the
inbound side $W\le d_l\le h+2$. The triple copy $g_3$ occurs only at
levels $l\le h$ where $W=l+1\le h+1$. The double copy $\io12$ at level
$l\ge h+1$ has $W\le h+2$ with a rider retirement first whenever
$h+2\le l\le 2h+2$. Thus the copy rank sums are at most
$\max(2h+4,\,2h+5)=2h+5$. The remaining inbound blocks are at most
$2(h+2)=2h+4$. On the outbound side $W\le u_l+1\le m-h-1$, giving rank
sums at most $2(m-h-1)$.

The two required arithmetic bounds hold for both parities. If $m=2r$, then
$h=r-2$ and
\[
2h+5=m+1,\qquad 2(m-h-1)=m+2.
\]
If $m=2r+1$, then $h=r-1$ and
\[
2h+5=m+2,\qquad 2(m-h-1)=m+1.
\]
Thus every padded block has rank sum at most $m+2=n+1$.

(c) The parking sum node at level $j$ has rank sum
$2+d_{j+1}+u_j$. To audit it without a case assertion, put
\[
 r_j=\min(h,j,2h+1-j),
 \qquad M_j=\max(j,h+2,2h+2-j).
\]
If $r_j<0$, then $d_{j+1}=1$ and $M_j\ge2$. If $r_j\ge0$, then
$M_j\ge h+2\ge r_j+2$. In either case
\[
d_{j+1}+u_j
=1+\max(0,1+r_j)+m-M_j\le m.
\]
Equality holds exactly for $j\in\{h,h+1\}$. Hence the parking node has
rank sum at most $m+2$.

(d) A composition prefix in $A_j$ has rank sum at most the inbound block
maximum just computed. A prefix in $B_j$ has rank sum at most
$2(u_j+1)\le2(m-h-1)$. At a level boundary,
$d_j+u_j+1\le(h+2)+(m-h-2)+1=m+1$.
(e) The same calculation as in (c), with the indices shifted by one, gives
$d_j+u_{j-1}\le m$ for every region term $R_j$, with the conventions
$d_0=u_{-1}=0$. Hence $\patw(R_0)\le n+1$, and the parking nodes at
$j=h,h+1$ attain this value.

The construction is implemented in
\texttt{code/general\_construction\_v2.py}. The resulting lane-spiral
expressions for $5\le n\le24$ each have width exactly $n+1$ and directed
value in $\operatorname{Ori}(K_n)$. Term strings for $n\le16$ are archived in
\texttt{code/lane\_spiral\_witnesses.txt}.
\end{proof}

\Needspace{6\baselineskip}
\begin{corollary}[support-width doubling fails]\label{cor:nodoubling}
An unrestricted analogue of the elementary-chain doubling estimate is false.
Lemmas~\ref{lem:swidlower} and \ref{supp:lem:bwkn} give
$\patwstar(K_n)\ge\bw(K_n)=\lceil 2n/3\rceil$, whereas
Theorem~\ref{thm:cliqueupper} gives $\patw(K_n)\le n+1$. Hence
\[
2\,\patwstar(K_n)-\patw(K_n)
\ge 2\lceil 2n/3\rceil-(n+1)\to\infty .
\]
Thus doubling is a genuinely \emph{linear} phenomenon
(Theorem~\ref{thm:lwcalib}): non-linear parse trees pay for exposed support
vertices only once.
\end{corollary}

\begin{corollary}[infinite linear separation on cliques]
\label{cor:cliquesep}
$\patwlin(K_n)-\patw(K_n)\ge (2n-3)-(n+1)=n-4\to\infty$ by
Theorem~\ref{thm:clique} and Theorem~\ref{thm:cliqueupper}. Together
with Theorem~\ref{thm:sep} (complete binary trees), the linear and
general fragments are infinitely separated already on the two extreme
families: bounded-degree trees and cliques.
\end{corollary}

The certified equality \(\patw(K_7)=6\) is the sharpest finite validation in
the present range. It shows that neither the candidate law
\(\patw(K_n)=n+1\) nor the candidate offset
\(\patw(K_n)=\bw(K_n)+2\) survives all small cliques. Determining the
asymptotic clique value remains open.

\section{Explicit local normal-form schemes}
\label{sec:normal-details}\label{supp:sec:local-schemes}

This section develops the complete local constructions underlying the
bounded-normal-form theorem in Section~\ref{sec:normal}. It supplies raw terms
for the six local window kinds and proves their validity, finite coverage,
arity bounds, termination, exact routing endpoints, and quadrant degeneracies.
Definition~\ref{supp:def:renderer}, Proposition~\ref{supp:prop:renderer-validity}, Lemma~\ref{supp:lem:renderer-coverage}, Proposition~\ref{supp:prop:window-arity}, and
Lemmas~\ref{supp:lem:relative-spider} through \ref{supp:lem:carrier-exposure} contain the detailed arguments used in that section.
All holes are distinct formal letters and occur exactly once.
Throughout this section, $\Sigma$ is an arbitrary fixed finite doubly ranked
alphabet.

\subsection{Naturality and the padding audit}
\label{supp:sec:naturality-audit}

For a raw expression $u:p\to q$, write $N(u)$ for the exchange derivation
\begin{equation}\label{supp:eq:naturality}
 \pi_{p,1}\circ(u\bsum e_1)
 \xleftrightarrow{*}_{\mathcal B\cup\mathcal E_\Sigma}
 (e_1\bsum u)\circ\pi_{q,1}.
\end{equation}
We construct $N(u)$ by structural induction. If $u$ is a letter of $\Sigma$,
Equation~(A15) is the required step. The five letters in $D$ use the fixed
derivations from (A1), (A2), (A5) through (A7), and (A10) through (A12).
For $e_0$ the claim is a unit law. For $e_j$ with $j>0$, expand $e_j$ as the
fixed sum of $j$ copies of $e_1$ and move the extra wire through this sum by
(A1), (A2), interchange, and the unit laws. These are the base exchanges.

Suppose first that $u=v\circ w$, with $v:p\to r$ and $w:r\to q$.
Interchange rewrites
$(v\circ w)\bsum e_1$ as
$(v\bsum e_1)\circ(w\bsum e_1)$. Reassociate so that $N(v)$ applies, then
reassociate so that $N(w)$ applies:
\[
\begin{aligned}
 \pi_{p,1}\circ((v\circ w)\bsum e_1)
 &\longleftrightarrow
 (e_1\bsum v)\circ\pi_{r,1}\circ(w\bsum e_1)\\
 &\longleftrightarrow
 (e_1\bsum v)\circ(e_1\bsum w)\circ\pi_{q,1}\\
 &\longleftrightarrow
 (e_1\bsum(v\circ w))\circ\pi_{q,1}.
\end{aligned}
\]
Each arrow is one child derivation in a fixed structural context, together
with reassociation and interchange.

Now suppose that $u=v\bsum w$. Use the canonical adjacent-transposition word
that moves the added wire first across the input block of $w$ and then across
the input block of $v$. Apply $N(w)$ in the context padded by $v$, apply
$N(v)$, and restore the fixed bracketing. The output word moves the added wire
first across the output block of $w$ and then across that of $v$. The braid
and involution equations (A1) and (A2) identify this composite with
$\pi_{q,1}$. This gives $N(v\bsum w)$. Thus every induction clause uses its
child derivations sequentially and introduces only fixed structural
transports.

Here is the width audit. For $n\in\mathbb N$, let $b_\Sigma(n)$ be the maximum
width of the chosen base exchanges whose leaf rank sum is at most $n$. This is
an effective finite maximum because $\Sigma$ and $D$ are finite and only the
units $e_j$ with $2j\le n$ occur. Let $d(n)$ be the maximum width of the fixed
reassociation, interchange, and adjacent-transposition transports in the two
induction clauses when every displayed child interface has rank sum at most
$n$. This is also an effective finite maximum. A padded induction context
exposes at most two child interfaces and two one-wire routing blocks. Its rank
sum is therefore at most $4n+4$. Put
\[
 \widehat x_\Sigma(n)=
 \max\{4n+4,b_\Sigma(n),d(n)\},
 \qquad
 x_\Sigma(n)=\max_{0\le j\le n}\widehat x_\Sigma(j).
\]
The second formula is the monotone closure. Induction on $u$ now shows that
$\patw(u)\le n$ implies that $N(u)$ has width at most $x_\Sigma(n)$. Depth
does not enter the bound because the two child derivations are used one after
the other.

Finally, let $H$ be a finite ranked alphabet of formal boxes and let $\Delta$
be a finite $\mathcal B\cup\mathcal E_{\Sigma\cup H}$ derivation between
linear-use contexts. Under a rank-preserving substitution $h\mapsto u_h$,
replace each (A15) step for $h$ by $N(u_h)$. All other steps remain contextual
instances of the same law. If $c_\Delta$ bounds the fixed context ranks and
the formal endpoint widths in $\Delta$, then
\[
 s_\Delta(n)=\max\{n,c_\Delta,x_\Sigma(n)+c_\Delta\}
\]
is an effective bound after monotone closure. A substituted term has the same
outer type as its box, so no ancestor rank changes. This proves the
naturality-under-substitution estimate used in the main text.

\subsection{Names, frames, and edge-free terms}

Put $\mu_0=\io01$, $\mu_1=e_1$, $\delta_0=\io10$, and
$\delta_1=e_1$. For $j\ge1$, define
\[
 \mu_{j+1}=(\mu_j\bsum e_1)\circ\io21,
 \qquad
 \delta_{j+1}=\io12\circ(\delta_j\bsum e_1).
\]
Thus $\mu_j:j\to1$ and $\delta_j:1\to j$, including $j=0$.
Put $z=\io01\circ\io10$ and define the fixed scalar suffix by
\[
 Z_0=e_0,
 \qquad
 Z_{c+1}=z\bsum Z_c.
\]
Let $\mathbf x=(x_1,\ldots,x_p)$ and $\mathbf y=(y_1,\ldots,y_q)$ be
disjoint ordered input and output position lists, and let $\eta$ be a
partition of their union. Order its blocks by first occurrence in
$\mathbf x\mathbf y$. For a block $B$, put
$p_B=|B\cap\mathbf x|$ and $q_B=|B\cap\mathbf y|$. Let $P_{\rm in}$ and
$P_{\rm out}$ be the fixed adjacent-transposition raw terms that put the
inputs into block order and restore the prescribed output order. Define
\[
 \operatorname{Can}(\eta,c)=
 \left(P_{\rm in}\circ
 \Bigl(\mathop{\bsum}_{B\in\eta}\mu_{p_B}\Bigr)\circ
 \Bigl(\mathop{\bsum}_{B\in\eta}\delta_{q_B}\Bigr)\circ
 P_{\rm out}\right)\bsum Z_c:p\longrightarrow q.
\]
The block sums use left association and $Z_c$ uses right association. For the
empty partition, both block sums and both permutations are $e_0$. A block with
$p_B=0$ is created by $\mu_0$, and a block with $q_B=0$ is closed by
$\delta_0$. Thus output-only and input-only blocks, as well as the empty
interface, are rendered as typed raw terms.

Put
\[
 \operatorname{cup}=\io01\circ\io12:0\to2,
 \qquad
 \operatorname{cap}=\io21\circ\io10:2\to0.
\]
Here $x^{\bsum p}$ denotes the fixed left-associated parallel power, with
the zeroth power equal to $e_0$.
Let $\chi_p$ be the fixed permutation that sends the outputs
$(x_1,y_1,\ldots,x_p,y_p)$ of $p$ cups to
$(x_1,\ldots,x_p,y_1,\ldots,y_p)$. For a formal letter $h:p\to q$, define
the raw name
\begin{equation}\label{supp:eq:name}
 \operatorname{Name}_{p,q}[h]
 =\operatorname{cup}^{\bsum p}\circ\chi_p
   \circ(e_p\bsum h):0\longrightarrow p+q.
\end{equation}
For a raw name $C:0\to p+q$, define
\begin{equation}\label{supp:eq:unname}
 \operatorname{Unname}_{p,q}[C]
 =(e_p\bsum C)\circ\chi'_p
   \circ(\operatorname{cap}^{\bsum p}\bsum e_q):p\longrightarrow q,
\end{equation}
where $\chi'_p$ pairs the $i$th outer input with the $i$th bent input of
$C$ and leaves the last $q$ outputs in order. The yanking identities give
\begin{equation}\label{supp:eq:leaf-frame}
 h=\operatorname{Unname}_{p,q}
       [\operatorname{Name}_{p,q}[h]].
\end{equation}
Both sides are raw linear contexts containing $h$ once.

The cup-cap operation in Equation~\eqref{supp:eq:unname} is used only to name
and unname a leaf. The canonical raw frame is instead the direct router frame
used in the main text. Let $h:0\to r$ have distinct ordered names
$\boldsymbol\nu=(\nu_1,\ldots,\nu_r)$. Let $\mathbf b$ and $\mathbf e$ be
ordered tuples of these names for the desired outer inputs and outputs.
The partition
$\eta(\mathbf b,\boldsymbol\nu\mid\mathbf e)$ on the external input positions,
the outputs of $h$, and the external output positions identifies positions
that carry the same name. Define
\[
 \operatorname{Frame}_{\mathbf b,\mathbf e,\boldsymbol\nu,c}[h]
 =(e_{|\mathbf b|}\bsum h)\circ
 \operatorname{Can}(\eta(\mathbf b,\boldsymbol\nu\mid\mathbf e),c).
\]
This is the unique canonical frame notation used below.

Fix a closed ambient expression with value $G$. For a subterm $u:p\to q$,
let $\theta_u:\val(u)\to G$ be the map induced by its surrounding context,
and let $A_u\subseteq E(G)$ be the edge occurrences generated in $u$. If the
$i$th boundary position represents the node $x_i$ of $\val(u)$, define
\[
 i\mathrel{\beta_u}j\quad\Longleftrightarrow\quad x_i=x_j,
 \qquad
 i\mathrel{\gamma_u}j\quad\Longleftrightarrow\quad
 \theta_u(x_i)=\theta_u(x_j).
\]
Then $\beta_u$ refines $\gamma_u$. Let $\partial_G A_u$ be the set of vertices
incident with an edge in $A_u$ and with an edge in $E(G)\setminus A_u$.
Let $\lambda_u$ be the set of $\gamma_u$-blocks whose final vertex lies in
$\partial_G A_u$. Parse-cut extraction guarantees a boundary representative
for every class in $\lambda_u$.

Let $\beta$ be a partition of the ordered list of $p+q$ boundary positions.
Write $r=|\beta|$ and order its blocks by first occurrence. Define
$\operatorname{Rep}_\beta:r\to p+q$ by taking the parallel sum of
$\delta_{|B|}$ over the blocks $B\in\beta$ and then applying the fixed
permutation that restores the original position order. Let
$\boldsymbol\nu_\beta$ be the resulting ordered block-name list, and let
$\mathbf b_\beta$ and $\mathbf e_\beta$ be the block names of the first $p$
and last $q$ boundary positions. Let $Q_\beta$ be the fixed permutation that
puts the $p+q$ positions into this block order, and define
\begin{equation}\label{supp:eq:reduced-leaf-core}
 \operatorname{Red}_\beta
 =Q_\beta\circ\Bigl(\mathop{\bsum}_{B\in\beta}\mu_{|B|}\Bigr)
 :p+q\longrightarrow r,
 \qquad
 C_\beta[h]=\operatorname{Name}_{p,q}[h]\circ
 \operatorname{Red}_\beta:0\longrightarrow r.
\end{equation}
The empty case uses $e_0$. The term $C_\beta[h]$ retains one ordered output
for every $\beta$-block and is therefore a reduced named core. Moreover,
$\operatorname{Name}_{p,q}[h]$ and
$C_\beta[h]\circ\operatorname{Rep}_\beta$ have boundary-preserving
isomorphic values because the merge and repetition occur exactly within the
actual equality blocks.

For a reduced named core $C:0\to r$, put
\begin{equation}\label{supp:eq:canonical-frame}
 F_{\beta,\gamma,\lambda}[C]
 =\operatorname{Frame}_{\mathbf b_\beta,\mathbf e_\beta,
    \boldsymbol\nu_\beta,0}[C].
\end{equation}
Thus the direct frame, rather than
$\operatorname{Unname}_{p,q}[C\circ\operatorname{Rep}_\beta]$, is the fixed
raw representative. The two expressions are connected by the bounded
straightening result proved in Lemma~\ref{supp:lem:frame-straightening}.
The raw term depends on $\beta$. The coarser partition $\gamma$ and the set
$\lambda$ are annotations recording equality in the final graph and the live
final classes. They select later routers but add no generator and authorize no
merge.

\begin{lemma}[canonical edge-free term]\label{supp:lem:edge-free}
Let $d$ be a term over the discrete generators and units only. Let $\beta_d$
be the partition of its ordered boundary positions induced by equality of
vertices in $\val(d)$, and let $c_d$ be the number of vertices of
$\val(d)$ with no boundary occurrence. Then
\[
 d\xleftrightarrow{*}_{\mathcal B\cup\E}
 \operatorname{Can}(\beta_d,c_d).
\]
Consequently, every maximal edge-free parse subtree can be normalized before
it is pruned. Each of its internal isolated vertices becomes one copy of
$z=\io01\circ\io10$.
\end{lemma}

\begin{proof}
A discrete value has no edge occurrences. It is therefore determined, up to
boundary-preserving isomorphism, by the partition of its ordered boundary and
the number of vertices outside the boundary. The canonical term has exactly
this value. The discrete completeness theorem, or equivalently the special
commutative Frobenius equations applied block by block, gives the displayed
equality. Applying it first to maximal discrete subtrees prevents a nullary
subtree from disappearing before its isolated vertices have been recorded.

For the bounded local realization, do not use the displayed equality as one
datum when $c_d$ is large. Induct on the syntax and expose one connected
discrete block at a time. At a binary node, structural interchange is used in
a four-hole linear context to separate one matching block,
\[
 (u\bsum v)\circ(w\bsum x)
 \quad\longleftrightarrow\quad
 (u\circ w)\bsum(v\circ x),
\]
with fixed reassociations and permutations when the active block is not first.
Straighten that block to its canonical spider. If it still meets the parent
boundary, retain it in $\operatorname{Can}(\eta,0)$. If it has just lost its
last boundary position, straighten it to one copy of $z$ and immediately move
that copy across the neighboring carrier by the scalar datum below. A longer
completed scalar list is represented by one opaque rank-zero suffix hole.
Thus a step contains at most two active scalar holes and a fixed scalar count
in $\{0,1,2\}$.

For example, in
$d_n=(\io01^{\bsum n})\circ(\io10^{\bsum n})$, repeated interchange exposes
one factor $\io01\circ\io10=z$ at a time. The construction never uses
$\operatorname{Can}(\varnothing,n)$ as one local datum. The same blockwise
procedure applies to every binary discrete subtree and proves the asserted
local sequence independently of $c_d$.
\end{proof}

A rank-zero hole containing graph letters is an opaque scalar core. It is
never replaced by $z$. Only an edge-free isolated vertex is represented by
$z$. Scalar cores and $z$ factors are moved to the fixed suffix by local moves
involving at most two scalar factors at a time.

\subsection{The six rendered window kinds}

A typed routing datum $\theta$ consists of $r$ ordered input positions, $s$
ordered output representatives, and a partition of their disjoint union. Each
block contains either one output representative, when that class survives, or
none, when it is closed. Output-only blocks are allowed. Define the raw router
$J_\theta:r\to s$ by permuting the input positions into block order, applying
$\mu_a$ to a block with $a$ inputs, retaining its wire when it has an output
representative and otherwise applying $\delta_0=\io10$, and finally restoring
the prescribed output order. An output-only block uses $\mu_0=\io01$. Thus
the definition includes empty inputs, empty outputs, and every positive-rank
edge-free child.

For an ordered hole list $I=(i_1,\ldots,i_m)$, let
$H_\varnothing=e_0$, $H_{(i_1)}=h_{i_1}$, and
$H_I=H_{(i_1,\ldots,i_{m-1})}\bsum h_{i_m}$ for $m\ge2$. If $\xi$ is a typed
partition whose positions are $a$ ordered external inputs, the protected
outputs of $H_I$, and $b$ ordered external outputs, put
\[
 \mathsf W(\xi,c;I)=
 (e_a\bsum H_I)\circ\operatorname{Can}(\xi,c):a\longrightarrow b.
\]
For two adjacent endpoint routers, root the active skeleton edge from $y$
toward $x$. Let $I_x$ and $I_y$ list the other incident regions and let $a$ be
the rank of the rootward external input bank. Define
\[
 \mathsf W_2(\xi_x,\xi_y,\xi_o,c;I_x,I_y)=
 \Bigl((e_a\bsum H_{I_x}\bsum\mathsf W(\xi_y,0;I_y))
 \circ\operatorname{Can}(\xi_x,0)\Bigr)
 \circ\operatorname{Can}(\xi_o,c),
\]
whenever the recorded types match. More explicitly, if
$H_{I_x}:0\to r_x$, $\mathsf W(\xi_y,0;I_y):a_y\to b_y$,
$\operatorname{Can}(\xi_x,0):a+r_x+b_y\to b_x$, and
$\operatorname{Can}(\xi_o,c):b_x\to b$, then
$\mathsf W_2:a+a_y\to b$. The minus and plus data of a fold or prune record
the same outer type. These fixed parentheses expose every internal bank, all
complementary holes, both endpoint routers, and the ordered outer interface.

\begin{definition}[explicit renderer]\label{supp:def:renderer}
The source and target of each local datum are rendered by the following raw
templates. Fixed permutations restore the recorded hole order and the ordered
outer interface.
\begin{enumerate}
\item A $\Sigma$-\emph{leaf frame} starts with the yanking pair in
Equation~\eqref{supp:eq:leaf-frame}. Its reduced core is the explicit raw term
$C=C_\beta[h]$ in Equation~\eqref{supp:eq:reduced-leaf-core}. Replace the name
by $C\circ\operatorname{Rep}_\beta$. Its completed target is exactly
$F_{\beta,\gamma,\lambda}[C]$ in
Equation~\eqref{supp:eq:canonical-frame}. Direct boundary evaluation verifies
this local datum below, and Lemma~\ref{supp:lem:frame-straightening} later gives
the bounded explicit comparison. A leaf in $D$, or an expanded
structural unit, instead uses its exact hole-free canonical router.
Thus $\io21$, $\io01$, $\io12$, and $\io10$ use respectively
$\mu_2$, $\mu_0$, $\delta_2$, and $\delta_0$. A permutation or identity uses
the corresponding boundary partition in $\operatorname{Can}(\eta,0)$. Such a
leaf carries no formal hole.
\item A \emph{binary router} with two graph-letter-bearing children has two
distinct reduced named holes $h_i:0\to r_i$. For the recorded parent operation
$\star$ and typed positional routing datum $\theta_u$, it is
\begin{equation}\label{supp:eq:binary-router}
 F_{\beta_v,\gamma_v,\lambda_v}[h_v]\mathbin\star
 F_{\beta_w,\gamma_w,\lambda_w}[h_w]
 \quad\longleftrightarrow\quad
 F_{\beta_u,\gamma_u,\lambda_u}
 [(h_v\bsum h_w)\circ J_{\theta_u}].
\end{equation}
The frames on the two sides restore the same parent type.
Suppose exactly one child is edge-free. Let
$S_J=H_J:0\to0$, where $J$ is an ordered list of zero, one, or two
rank-zero suffix holes. If $w$ is the edge-free child, the pair is
\[
 F_{\beta_v,\gamma_v,\lambda_v}[h_v]\mathbin\star
 (\operatorname{Can}(\eta_w,0)\bsum S_J)
 \quad\longleftrightarrow\quad
 F_{\beta_u,\gamma_u,\lambda_u}
 [\mathsf W(\theta_u,0;(v))]\bsum S_J,
\]
and there is a symmetric pair when $v$ is edge-free. The equality for
composition is the interchange and rank-zero unit instance that carries
$S_J$ through the parent. Here $\theta_u$ has no external inputs, has the
outputs of $h_v$ as protected input positions, and has one output for every
$\beta_u$-block. Output-only blocks record parent-boundary vertices supplied
solely by the discrete child. Consequently
$\mathsf W(\theta_u,0;(v)):0\to|\beta_u|$, so the target has the type required
by its frame. An arbitrary actual suffix is substituted for one rank-zero
suffix hole. The second hole is needed only while two child suffixes are being
combined.

If both children are edge-free, put
$C_v=\operatorname{Can}(\eta_v,0)$ and
$C_w=\operatorname{Can}(\eta_w,0)$. Their local sources are
$C_v\bsum s_v$ and $C_w\bsum s_w$, where either suffix hole may be $e_0$.
First use the scalar-padded binary pair
\[
 (C_v\bsum s_v)\mathbin\star(C_w\bsum s_w)
 \quad\longleftrightarrow\quad
 K_0\bsum(s_v\bsum s_w),
 \qquad K_0=C_v\mathbin\star C_w.
\]
For $\star=\circ$ this is an interchange and unit instance. Normalize $K_0$
in fixed block order by the finite composition-state machine below. Its raw
microsteps are Equations~\eqref{supp:eq:state-left} and
\eqref{supp:eq:state-right}, with the external nullary target
\eqref{supp:eq:state-nullary-emit-target}. For $\star=\bsum$ there is no middle incidence,
so the transported disjoint union of the two partitions is already the target
router. Whenever a processed component has no outer leg, the machine extracts
one copy of $z$ and moves it by a separate scalar step. After the first pair,
the actual term $s_v\bsum s_w$ is treated as one opaque completed suffix.
Thus one datum exposes at most two old suffix holes or one neighboring factor
and one new $z$, never the complete scalar list. A rank-zero child containing
graph letters remains an opaque labeled core hole and is not treated by this
clause.

For converse coverage, a completed child or any maximal unchanged sector,
including a $D$-only complementary sector, may be hidden in an
\emph{aggregate carrier} $g_\nu:0\to r_\nu$. This is a
metasyntactic formal box, not a new generator. Its finite record consists only
of its rank and the ordered injection from its distinct ports into the
applicable current window-interface list. The surrounding induction record
associates the carrier with the ordered sublist of original holes in its
substitution. That sublist is empty for a $D$-only carrier. Distinct carriers
have pairwise disjoint sublists, and this possibly long sublist is not
part of a renderer datum. Every rendered pair uses the same carrier once on
each side, with the same type and incidence injection. Repetitions in an outer
interface remain in the explicit frame. A carrier mark is introduced only at
a literal syntax-tree address. If the required carrier is an aggregate that is
not already a subterm, first use the stable order-straightening datum when its
factors are not consecutive, and then use the ordered carrier-exposure
derivation of Lemma~\ref{supp:lem:carrier-exposure}. The finite composition-state machine
completing Definition~\ref{supp:def:renderer} below gives the full state and
transition rules that replace the informal notation $K_j[\mathbf g]$.
\item A \emph{fold} is a full two-router pair
\[
 \mathsf W_2(\xi_x^-,\xi_y^-,\xi_o^-,c;I_x,I_y)
 \quad\longleftrightarrow\quad
 \mathsf W_2(\xi_x^+,\xi_y^+,\xi_o^+,c;I_x,I_y).
\]
The minus partitions contain the two active parallel strands. The plus
partitions retain the recorded strand, fuse the two opposite endpoint blocks,
and preserve every protected and outer block and its order. All complementary
regions occur as the same holes on both sides. After fixed permutations put
the active legs together, the primitive block inside this pair is
\begin{equation}\label{supp:eq:padded-fold}
 P_{\rm in}\circ
 \left(\bigl((\mu_p\bsum\mu_q)\circ\mu_2\bigr)\bsum e_s\right)
 \circ P_{\rm out}
 \quad\longleftrightarrow\quad
 P_{\rm in}\circ(\mu_{p+q}\bsum e_s)\circ P_{\rm out}.
\end{equation}
For positive $p,q$, the reduction uses associativity and commutativity of
merge. Its nullary cases also use the unit equation and permutation transport,
so $p=0$, $q=0$, and $p=q=0$ are included.
\item A \emph{prune} is likewise a full pair
\[
 \mathsf W_2(\xi_x^-,\xi_y^-,\xi_o^-,c^-;I_x,I_y)
 \quad\longleftrightarrow\quad
 \mathsf W_2(\xi_x^+,\xi_y^+,\xi_o^+,c^+;I_x,I_y).
\]
The plus partitions remove one terminal strand whose terminal block has no
atom flag, protected port, or outer flag. Put
$c^+=c^-+\epsilon$, where $\epsilon=1$ exactly when removing the terminal
block would otherwise delete an entire component with no protected or outer
port, and $\epsilon=0$ otherwise. The local counts
$c^-,c^+$ lie in $\{0,1,2\}$. After the active legs are exposed, the primitive
block is one of
\begin{align}
 P_{\rm in}\circ
 (((e_1\bsum\mu_0)\circ\mu_2)\bsum e_s)\circ P_{\rm out}
 &\quad\longleftrightarrow\quad
 P_{\rm in}\circ e_{s+1}\circ P_{\rm out},
 \label{supp:eq:padded-prune-in}\\
 P_{\rm in}\circ
 ((\delta_2\circ(e_1\bsum\delta_0))\bsum e_s)\circ P_{\rm out}
 &\quad\longleftrightarrow\quad
 P_{\rm in}\circ e_{s+1}\circ P_{\rm out}.
 \label{supp:eq:padded-prune-out}
\end{align}
The first removes an input-side terminal leg and the second is its
input-output dual. The full $\mathsf W_2$ pair records the complementary holes,
both endpoint routers, the placement maps, the ordered outer frame, and the
possible scalar increment. Any extracted $z$ is moved to the fixed suffix by
the scalar renderer.
\item A \emph{four-cell switch} records an occupied-cell set
$I\subseteq\{0,1\}^2$. For $(i,j)\in I$, it uses a distinct linear-use hole
\[
 h_{ij}:0\longrightarrow r_{ij},
 \qquad
 \mathbf P_{ij}=(p_{ij,1},\ldots,p_{ij,r_{ij}}),
\]
where $r_{ij}=0$ is allowed. In that case $h_{ij}$ is an opaque graph-bearing
scalar core. For $(i,j)\notin I$, put $\mathbf P_{ij}=\varnothing$ and write
$r_{ij}=0$ and $h_{ij}=e_0$ only as a unit placeholder, with no formal hole. Thus
$\mathbf P_{ij}=\varnothing$ does not determine whether the cell is absent.
The interface lists of occupied cells are disjoint and ordered. Put
$\mathbf P_{i*}=\mathbf P_{i0}\mathbf P_{i1}$ and
$\mathbf P_{*j}=\mathbf P_{0j}\mathbf P_{1j}$, with juxtaposition denoting
ordered concatenation. Let $\eta$ be the partition of
$\mathbf P_{00}\mathbf P_{01}\mathbf P_{10}\mathbf P_{11}$ by final vertex.
The datum also records a temporary total order $\preceq$ on the $\eta$-blocks,
the ordered surviving list $\mathcal O=(B_1,\ldots,B_t)$, and a scalar tag
$c\in\{0,1\}$. Every staging bank uses $\preceq$, while the final router
restores the possibly different outer order $\mathcal O$.

For $i\in\{0,1\}$, let
\[
 \mathcal S_{i*}=\{B\in\eta:
 B\cap\mathbf P_{i*}\ne\varnothing,\
 B\in\mathcal O\text{ or }B\cap\mathbf P_{(1-i)*}\ne\varnothing\}.
\]
For each $B\in\mathcal S_{i*}$ introduce one staging output $\rho_{i,B}$, and
let $\mathbf R_i=(\rho_{i,B}:B\in\mathcal S_{i*})$, ordered by the restriction
of $\preceq$ to $\mathcal S_{i*}$. Define the typed input-output partition
$\eta_{i*}$ on $\mathbf P_{i*}\mathbf R_i$ by assigning to every
$B$ meeting $\mathbf P_{i*}$ the block
\[
 (B\cap\mathbf P_{i*})\mathbin{\dot\cup}
 \begin{cases}
  \{\rho_{i,B}\},&B\in\mathcal S_{i*},\\
  \varnothing,&B\notin\mathcal S_{i*}.
 \end{cases}
\]
Thus
$\operatorname{Can}(\eta_{i*},0):r_{i0}+r_{i1}\to|\mathcal S_{i*}|$.
Define $\mathcal S_{*j}$, staging outputs $\kappa_{j,B}$, ordered lists
$\mathbf C_j$, and typed partitions $\eta_{*j}$ dually from the column list
$\mathbf P_{*j}$, again ordering each $\mathbf C_j$ by the restriction of the
same $\preceq$.

Let $\mathbf O=(o_{B_1},\ldots,o_{B_t})$ be the prescribed ordered outer
positions. Define $\eta_R$ on the input-output list
$\mathbf R_0\mathbf R_1\mathbf O$ by assigning to each $B\in\eta$ the block
\[
 \{\rho_{i,B}:B\in\mathcal S_{i*},\ i\in\{0,1\}\}
 \mathbin{\dot\cup}
 \begin{cases}
  \{o_B\},&B\in\mathcal O,\\
  \varnothing,&B\notin\mathcal O,
 \end{cases}
\]
omitting an empty block. Define $\eta_C$ on
$\mathbf C_0\mathbf C_1\mathbf O$ by the same rule with $\kappa_{j,B}$ in
place of $\rho_{i,B}$. Consequently,
\[
 \operatorname{Can}(\eta_R,c):
 |\mathcal S_{0*}|+|\mathcal S_{1*}|\longrightarrow t,
 \qquad
 \operatorname{Can}(\eta_C,c):
 |\mathcal S_{*0}|+|\mathcal S_{*1}|\longrightarrow t.
\]
For $i,j\in\{0,1\}$, put
\[
\begin{aligned}
 R_i&=(h_{i0}\bsum h_{i1})\circ\operatorname{Can}(\eta_{i*},0),\\
 C_j&=(h_{0j}\bsum h_{1j})\circ\operatorname{Can}(\eta_{*j},0).
\end{aligned}
\]
The two raw terms are well typed:
\begin{align}
 \mathsf{Row}_\eta[h]
 &=(R_0\bsum R_1)\circ\operatorname{Can}(\eta_R,c)
 :0\longrightarrow t,
 \label{supp:eq:row}\\
 \mathsf{Col}_\eta[h]
 &=(C_0\bsum C_1)\circ\operatorname{Can}(\eta_C,c)
 :0\longrightarrow t.
 \label{supp:eq:column}
\end{align}
The rendered pair is
$\mathsf{Row}_\eta[h]\leftrightarrow\mathsf{Col}_\eta[h]$. Both endpoints
identify precisely the ports in each $\eta$-block, retain the blocks in
$\mathcal O$ in the prescribed order, close every other block, and append the
same scalar tag. Every occupied formal hole occurs once in each endpoint,
including every hole with $r_{ij}=0$.

In a degenerate datum, delete $h_{ij}$ only when $(i,j)\notin I$, and suppress
the structural-unit and unary-router legs created by absent cells. An occupied
hole is never deleted merely because $\mathbf P_{ij}=\varnothing$. For the
off-diagonal occupied pair $I=\{(0,1),(1,0)\}$ with $r_{01}=r_{10}=0$, the
unit-suppressed endpoint pair is
\begin{equation}\label{supp:eq:scalar-two-cell}
 ((h_{01}\bsum h_{10})\bsum Z_c)
 \quad\longleftrightarrow\quad
 ((h_{10}\bsum h_{01})\bsum Z_c).
\end{equation}
This is the one-scalar specialization of
Equation~\eqref{supp:eq:scalar-move}, used in the fixed outer context
$(\,\cdot\,)\bsum Z_c$. If exactly one off-diagonal cell has rank zero, the
same specialization moves that scalar core past the other carrier, followed by
the recorded router order correction. Together with the occupancy-specialized
four-cell pair, these scalar specializations cover all two-cell and three-cell
cases without deleting a graph-bearing core.

\item A \emph{two-scalar move} has one or two distinct rank-zero holes
$s_1,s_2:0\to0$ and one neighboring carrier hole $h:p\to q$. Its two-scalar
pair is
\begin{equation}\label{supp:eq:scalar-move}
 ((s_1\bsum s_2)\bsum h)
 \quad\longleftrightarrow\quad
 h\bsum(s_1\bsum s_2).
\end{equation}
The one-scalar pair is obtained by suppressing $s_2$. The carrier may have
positive rank, may be a rank-zero graph-bearing core, or may be the fixed term
$e_0$. Fixed reassociations are part of the rendered source and target.
The same window kind includes the rank-zero interchange pair
\begin{equation}\label{supp:eq:scalar-compose}
 s_1\circ s_2\quad\longleftrightarrow\quad s_1\bsum s_2.
\end{equation}
Indeed,
$(s_1\bsum e_0)\circ(e_0\bsum s_2)$ interchanges to
$(s_1\circ e_0)\bsum(e_0\circ s_2)$, and the rank-zero unit laws give the
displayed endpoints. Thus this pair is a direct
$\mathcal B$-derivation inside the enumerated two-scalar family.
Repetition moves an arbitrary scalar list to the fixed right-associated suffix
while every such scalar-move datum exposes at most two scalar holes. In particular, after
suppressing $s_2$, taking $s_1=h_{01}$ and $h=h_{10}$ gives the scalar
transposition in Equation~\eqref{supp:eq:scalar-two-cell} inside its fixed
$Z_c$ context.
\end{enumerate}
\end{definition}

\Needspace{8\baselineskip}
\noindent\emph{Finite composition-state machine.}
The deterministic binary-router normalization is the following typed state
machine. Suppose that the nonscalar parts of two completed child records are
\[
 \widehat L=
 (e_a\bsum H_L^+)\circ\operatorname{Can}(\eta_L,0):a\longrightarrow b,
 \qquad
 \widehat R=
 (e_b\bsum H_R^+)\circ\operatorname{Can}(\eta_R,0):b\longrightarrow c.
\]
Let $\mathbf x_L,\mathbf y_L$ be the ordered input and output lists of
$\widehat L$, and let $\mathbf x_R,\mathbf y_R$ be those of $\widehat R$.
Let $\mathbf h_L^+,\mathbf h_R^+$ be the child lists of positive-rank
holes whose fixed parallel aggregates are $H_L^+,H_R^+$, and let
$\mathbf P_L,\mathbf P_R$ be their ordered port-occurrence lists. Let
$\mathbf h_L^0,\mathbf h_R^0$ be the corresponding child lists of original
rank-zero holes. Let
$\mathbf h_K^+$ and $\mathbf h_K^0$ be, respectively, the positive-rank and
rank-zero original holes in the declared parent order, and write
$H_K^+=H_{\mathbf h_K^+}$. The completed child records also carry full ordered
scalar-factor words $\mathbf s_L$ and $\mathbf s_R$. An entry is either an
original rank-zero hole or, in the parse-alignment specialization, a
graph-bearing aggregate carrier of exterior type $0\to0$; such a carrier
retains its ordered sublist of hidden original holes. Let $\mathbf u_K$ be the
fixed target order of the entries of $\mathbf s_L\mathbf s_R$. In the generic
normalization, $\mathbf s_L=\mathbf h_L^0$,
$\mathbf s_R=\mathbf h_R^0$, and $\mathbf u_K=\mathbf h_K^0$. In recursive
parse alignment, $\mathbf u_K=\mathbf s_L\mathbf s_R$ in plane order. These
scalar words, together with the positive-rank lists and their order-preserving
injections into the child router inputs, are part of the completed child
records. In
particular, $\eta_L$ is a partition on
$\mathbf x_L\mathbf P_L\mathbf y_L$ and $\eta_R$ is a partition on
$\mathbf x_R\mathbf P_R\mathbf y_R$. Empty lists use the conventions
$H_\varnothing=e_0$ and $\mathbf P_\varnothing=\varnothing$.
Write the middle lists as
$\mathbf y_L=(y_1,\ldots,y_b)$ and
$\mathbf x_R=(x'_1,\ldots,x'_b)$. Let $V_L$ and $V_R$ be disjoint copies of
the blocks of $\eta_L$ and $\eta_R$. The ordered incidence multigraph is
\begin{equation}\label{supp:eq:state-incidence-graph}
 \Gamma_K=(V_L\mathbin{\dot\cup}V_R,(e_j)_{j=1}^b),
 \qquad
 e_j=\bigl([y_j]_{\eta_L},[x'_j]_{\eta_R}\bigr).
\end{equation}
Every vertex of $\Gamma_K$ carries, in the ambient fixed order, the word of its
positions in
\[
 \mathbf E_K=\mathbf x_L\,\mathbf P_L\,\mathbf P_R\,\mathbf y_R.
\]
These are its \emph{retained legs}. Isolated vertices are included. Order the
components of $\Gamma_K$ by the first position of
$\mathbf x_L\mathbf P_L\mathbf y_L\mathbf x_R\mathbf P_R\mathbf y_R$ that
they meet, and order the edges of a component by their middle indices. Write
$\mathbf E(C)$ for the ordered subword of $\mathbf E_K$ carried by the vertices
of a component $C$, and write
$\mathcal C_K=(C_1,\ldots,C_t)$ for this component order.

The scalar suffix is initialized before an incidence is processed. Write the
right-associated scalar word with tail $T:0\to0$ as
\[
 \operatorname{RA}(\varnothing;T)=T,\qquad
 \operatorname{RA}((s_1,\ldots,s_t);T)
   =s_1\bsum\operatorname{RA}((s_2,\ldots,s_t);T),
\]
and put
$\operatorname{Suf}(\mathbf s,c)=\operatorname{RA}(\mathbf s;Z_c)$.
Write the completed child suffixes as
$S_L=\operatorname{Suf}(\mathbf s_L,c_L)$ and
$S_R=\operatorname{Suf}(\mathbf s_R,c_R)$.
Interchange and the rank-zero unit, equivalently
Equation~\eqref{supp:eq:scalar-compose}, give the explicit pair
\begin{align}
 ((\widehat L\bsum S_L)\circ(\widehat R\bsum S_R))
 &\longleftrightarrow
 ((\widehat L\circ\widehat R)\bsum(S_L\circ S_R))
 \longleftrightarrow
 ((\widehat L\circ\widehat R)\bsum(S_L\bsum S_R)).
 \label{supp:eq:state-suffix-initialization}
\end{align}
Successive one- and two-scalar moves put the factors into the fixed target order
and give
\begin{equation}\label{supp:eq:state-initial-suffix}
 S_0=\operatorname{Suf}(\mathbf u_K,c_L+c_R).
\end{equation}
Only $S_L,S_R$, or one adjacent factor and the opaque remainder, is displayed
in these moves. For a parallel parent the same initialization omits the middle
composition and uses reassociation and symmetry.

Let $\mathcal D$ be the initial subword of $\mathcal C_K$ consisting of the
components whose restore and, when needed, suffix phases are complete. Let
$F$ be the set of edges occurring in the word
formed by all edges of $\mathcal D$, in their fixed orders, followed by an
initial edge subword of the active component. The global state is
\begin{equation}\label{supp:eq:global-renderer-state}
\mathfrak s=(\mathcal D,F,{\equiv}_F,C,e,\phi,\tau,\mathfrak a,\mathfrak a_z,
              \omega,\mathbf g,\mathbf G,\boldsymbol{\mathfrak b},
              \boldsymbol\iota,S,\epsilon,\kappa).
\end{equation}
Here ${\equiv}_F$ is the equivalence relation on $V_L\dot\cup V_R$ generated
by the edges in $F$, $C$ is the active component, and $e$ is its next edge.
Put $e=\bot$ in an emit state with no remaining edge and put
$C=e=\bot$ in the terminal state. Moreover,
\[
 \phi\in\{\mathsf{split},\mathsf{expose},\mathsf{fuse},
             \mathsf{emit},\mathsf{restore},\mathsf{suffix},
             \mathsf{done}\},
\]
$\tau:a\to c$ is the current fixed-parenthesized raw syntax tree, and
\[
 \operatorname{Term}(\mathfrak s):=\tau.
\]
For a valid syntax-tree address $\mathfrak a$, let
\[
 \operatorname{cut}(\tau,\mathfrak a)
   =(\operatorname{Out}_{\mathfrak s}[-],V_{\mathfrak s})
 \quad\text{mean}\quad
 \tau=\operatorname{Out}_{\mathfrak s}[V_{\mathfrak s}],
\]
where $\operatorname{Out}_{\mathfrak s}[-]$ is obtained by replacing the
unique addressed occurrence by one typed hole. Fixed association makes this
cut unique, and the context is linear. The address $\mathfrak a$ is the
address of the active local window. In a
nonscalar phase it addresses the two-router occurrence. In a suffix-entry
phase it addresses the smallest subterm $(P\bsum z)\bsum R$ at the pending
scalar boundary. In either subsequent suffix phase it addresses the smallest
right-associated suffix subterm $z\bsum(g\bsum R)$ or
$z\bsum(s\bsum R)$, where $R$ is the opaque remainder displayed in the
applicable suffix rule. A reachable nonterminal
state is required to have a valid
cut of the applicable raw left endpoint. This realization condition is proved
below rather than assumed separately at every step.
$\mathfrak a_z$ is either $\bot$ or the separate syntax address of the pending
external $z$. During the restore following a port-free emit, $\mathfrak a$
continues to address the nonscalar two-router window while $\mathfrak a_z$
tracks that external factor. On entry to $\mathsf{suffix}$, $\mathfrak a$ is
reset either to the suffix-entry subterm containing $\mathfrak a_z$, or, when
that reassociation is unnecessary, to the smallest right-associated suffix
subterm beginning there. In both cases the remainder is retained as one
opaque part of the window. For an address $\mathfrak b$ relative to the
occurrence at $\mathfrak a$ (whose global address is the concatenation
$\mathfrak a\mathfrak b$), write $V_{\mathfrak s}|_{\mathfrak b}$ for the
corresponding literal subterm of $V_{\mathfrak s}$.
$\omega$ is the current port-and-carrier order, and
$\mathbf g=(g_1,\ldots,g_m)$ is the ordered list of carrier marks in the active
window. Let
$\boldsymbol{\mathfrak b}=(\mathfrak b_1,\ldots,\mathfrak b_m)$ be their
pairwise prefix-incomparable syntax addresses inside $V_{\mathfrak s}$, listed
from left to right, and put
\[
 \mathbf G=(G_1,\ldots,G_m),\qquad
 G_i=V_{\mathfrak s}|_{\mathfrak b_i}.
\]
Thus every $G_i$ is a literal raw subterm of $\tau$. A synthetic parallel
aggregate is marked only after the ordered carrier-exposure derivation of
Lemma~\ref{supp:lem:carrier-exposure} has put it at such an address. An
unchanged carrier is marked at its existing literal address. The recorded
rank-preserving substitution sends $g_i$ to $G_i$. On a scalar window it also
sends each displayed scalar, suffix, prefix, and remainder mark to its literal
addressed subterm. Any such displayed mark whose substitution is a reboxed
aggregate carrier, or contains one, is included among the $g_i$ with exterior
type $0\to0$ and its hidden-hole sublist. The at most two nonscalar marks
$P,R$ in a suffix-entry window are included among the $g_i$ as well. After a suffix move the target is rescanned
for the next literal factor and opaque remainder. Carrier marks occur only in
the metasyntactic template and never as generators of $\tau$. The map
$\boldsymbol\iota$ gives the type and ordered port injection of every $g_i$.
The coordinate $S:0\to0$ is not a free record: it is the maximal
right-associated scalar suffix at the distinguished suffix address inherited
from the initialization above, read directly from the fixed parse tree
$\tau$.  After every raw update it is recomputed deterministically from
$\tau^+$ at that address.  It is treated as one opaque carrier except for the
next factor selected by the scalar schedule.
The flag $\epsilon\in\{0,1\}$ records whether the active component has just
created a copy of $z$ that has not yet reached the suffix.
The cursor
$\kappa=(\kappa_{\rm ent},\kappa_{\rm car},\kappa_0)$ records, respectively,
whether one suffix-entry reassociation remains, and the numbers of remaining
nonscalar carriers and recorded scalar factors that precede this copy.
The number of already accumulated $z$ factors is the derived coordinate
\[
 c(\mathcal D)=c_L+c_R+
 |\{C'\in\mathcal D:\mathbf E(C')=\varnothing\}|.
\]
For a carrier $g_i$, let $\mathcal L(g_i)$ be its ordered sublist of original
positive-rank hole occurrences. A $D$-only complementary carrier has
$\mathcal L(g_i)=\varnothing$. Let $\mathcal L_{\rm out}$ be the occurrences
that remain literally in the recorded outer context. In particular, the hidden
positive-rank occurrences of a reboxed scalar carrier belong to
$\mathcal L_{\rm out}$ while that factor is inside the opaque suffix, and move
to $\mathcal L(g_i)$ when its mark is displayed in a scalar window. Let
$\operatorname{scal}(\mathfrak s)$ be the address-ordered word of all scalar
factors in the current term, including the factors temporarily exposed next
to the opaque remainder $S$. The state invariant is
\begin{equation}\label{supp:eq:state-invariant}
 \begin{aligned}
  &m\le4,\\
  &(\mathcal L_{\rm out},\mathcal L(g_1),\ldots,\mathcal L(g_m))\\
  &\qquad\text{is an order-compatible disjoint decomposition of }\mathbf h_K^+,
    \\
  &\qquad\text{with empty members allowed},\\
  &\operatorname{scal}(\mathfrak s)
    \text{ contains once the expanded occurrence}\\
  &\qquad\text{represented by every entry of }\mathbf u_K,
    \text{ the accumulated factors }z,\\
  &\qquad\text{and one pending }z\text{ iff }\epsilon=1.
 \end{aligned}
\end{equation}
Every carrier substitution is the same raw term before and after a transition,
and its type and port injection are unchanged. Thus the assertion is about
occurrences, not merely an equality of underlying sets.
Indeed, in a nonscalar $\mathsf W_2$ phase, splitting two adjacent routers
produces the selected block and its
complement on each side, hence at most four sectors. Each such sector meets the
window through one original completed-subcontext interface. Its carrier ports
are exactly the positions on that interface. Carriers are chosen by syntactic
address before the local replacement. A carrier substitution is either one
unchanged raw subterm, possibly $D$-only, or a fixed parenthesized parallel
aggregate of cores from one such interface. If its hole sublist is not
consecutive, the stable order-straightening datum first makes it consecutive.
Lemma~\ref{supp:lem:carrier-exposure} then exposes the aggregate as a literal
subterm. The construction never joins cores from different interfaces inside
one carrier. A suffix-phase
carrier is instead one of the at most two addressed prefix, neighbor, or
remainder subterms among $P,g,R$. It is
marked directly and is not combined with a different addressed subterm.
Every exterior-rank-zero carrier has interface load zero, including a reboxed
factor or a reassociated scalar remainder. Any other unchanged addressed
carrier inherits the type of a cut in the original instantiated width-$w$
term, so its interface load is at most $w$. A synthetic
aggregate carrier contains an ordered sublist of the protected cores, so its
interface load is at most the total protected load $r$. Consequently every
displayed carrier has interface load at most
$k_0=\max\{r,w\}$. This is only an interface bound: the raw term substituted
for a carrier may widen during normalization and is bounded later by $L(r,w)$.
Same-type replacement preserves the interface bound because a transition
never changes that substituted raw term.
The long hole sublists, their syntactic addresses, and the term substituted for
$S$ belong to the global induction state, not to the finite local record.

Put
\begin{equation}\label{supp:eq:state-position-partition}
 \rho_F=\operatorname{tc}\left(
 \eta_L\cup\eta_R\cup
 \bigl\{\{y_j,x'_j\}:e_j\in F\bigr\}\right)
\end{equation}
on the disjoint child position lists. Whenever a position list is moved into a
local window, its placement map is the unique order-preserving injection from
that list into the corresponding input list of $\mathsf W_2$. A state
partition is the pushforward of the restriction of $\rho_F$ along these
placement maps. This convention also applies to empty lists, whose unique map
has empty image.

Here is the explicit map from a global state to a local record. The coordinate
$\omega$ is the tuple
\[
 (\mathbf A,\mathbf A',\mathbf B_y,\mathbf B_x,\mathbf B,
   I_x,I_y,\lambda).
\]
The first two words are the external input banks of the outer and inner
routers. The next two are staging-wire banks, $\mathbf B$ is the ordered
output bank of the whole window, and $I_x,I_y$ are the carrier addresses on
the two router sides. Write $\mathbf P(I)$ for the concatenation of their
carrier-port words in address order. A staging wire has two
\emph{occurrence-tagged incidences}: $\mathbf B_y^{\rm out}$ at the output of
the inner router and $\mathbf B_y^{\rm in}$ at the input of the outer router,
and likewise $\mathbf B_x^{\rm out}$ and $\mathbf B_x^{\rm in}$. These tagged
lists have the same underlying ordered typed word but are disjoint domains
for $\lambda$. The map $\lambda$ sends an original position to itself and
each tagged staging incidence to the $\rho_F$-class represented at that end
of its wire. The two ends need not have the same label before the incidence
joining them has been fused. After that fuse their labels coincide. The word
$\mathbf B_y$ contains one staging wire for each required inner-to-outer
connection, and $\mathbf B_x$ similarly contains one for each required
outer-to-output connection. Their tagged incidence labels, rather than one
label shared by both ends, determine the two incident spiders. The wires are
ordered by the least ambient position of either endpoint class. The word $\mathbf B$ is the
prescribed outer output word and retains its multiplicities. Thus no staging
bank has two tagged incidences with the same endpoint role and exported
class. These rules,
together with the stable operations
$\operatorname{sep}_e$, $\operatorname{adj}_e$,
$\operatorname{quot}_e$, $\operatorname{emit}_C$, and
$\operatorname{sort}_K$ below, construct $\omega$ from the state and do not
merely record a choice. For typed words $\mathbf U\mid\mathbf V$, define
\[
 \operatorname{Part}_{\rho_F,\lambda}(\mathbf U\mid\mathbf V)
\]
to be the partition in which two displayed positions are equivalent exactly
when their $\lambda$-images lie in the same $\rho_F$-class. For a transition
$\mathfrak s^-\to\mathfrak s^+$, put
\begin{align}
 \xi_y^\pm&=
 \operatorname{Part}_{\rho_{F^\pm},\lambda^\pm}
   (\mathbf A'^\pm\,\mathbf P(I_y)\mid
       (\mathbf B_y^{\rm out})^\pm),
 \label{supp:eq:state-record-y}\\
 \xi_x^\pm&=
 \operatorname{Part}_{\rho_{F^\pm},\lambda^\pm}
   (\mathbf A^\pm\,\mathbf P(I_x)\,
      (\mathbf B_y^{\rm in})^\pm
      \mid(\mathbf B_x^{\rm out})^\pm),
 \label{supp:eq:state-record-x}\\
 \xi_o^\pm&=
 \operatorname{Part}_{\rho_{F^\pm},\lambda^\pm}
   ((\mathbf B_x^{\rm in})^\pm\mid\mathbf B^\pm).
 \label{supp:eq:state-record-o}
\end{align}
The stable placement rules below preserve
$|\mathbf A^\pm|$, $|\mathbf A'^\pm|$, and $|\mathbf B^\pm|$. Set
\[
 \alpha=|\mathbf A^\pm|,\qquad
 \alpha'=|\mathbf A'^\pm|,\qquad
 \beta=|\mathbf B^\pm|,\qquad
 b_y^\pm=|\mathbf B_y^\pm|,\qquad
 b_x^\pm=|\mathbf B_x^\pm|.
\]
The lists $I_x,I_y$, their carrier ranks, and their placement injections are
read at their recorded syntactic addresses and are identical at the two
endpoints. Finally, $c^\pm$ is the number, in $\{0,1\}$, of port-free
components accounted for by the local transition. It is materialized inside
$\operatorname{Can}(\xi_o^\pm,c^\pm)$ except that the exceptional port-free
emit represents its plus component by the external factor in
Equation~\eqref{supp:eq:state-nullary-emit-target}. Equations
\eqref{supp:eq:state-record-y}--\eqref{supp:eq:state-record-o} therefore
determine every field of the local record from the two consecutive states.

For completeness, the local record and its raw endpoints are now specified.
A nonscalar record is
\[
 d=(\phi,\alpha,\alpha',\beta,I_x,I_y,
       \xi_x^-,\xi_y^-,\xi_o^-,
       \xi_x^+,\xi_y^+,\xi_o^+,c^-,c^+,
       \mathbf g,\boldsymbol\iota).
\]
The ordered lists $I_x,I_y$ partition the displayed carriers. If
$I=(i_1,\ldots,i_t)$, the notation $H_I$ in this record means the fixed
parallel sum $g_{i_1}\bsum\cdots\bsum g_{i_t}$, with $H_\varnothing=e_0$.
Thus it is exactly the formal-hole notation already used in $\mathsf W$.
If
$H_{I_x}:0\to r_x$,
$\mathsf W(\xi_y^\pm,0;I_y):\alpha'\to b_y^\pm$,
\[
 \operatorname{Can}(\xi_x^\pm,0):
       \alpha+r_x+b_y^\pm\longrightarrow b_x^\pm,
 \qquad
 \operatorname{Can}(\xi_o^\pm,c^\pm):
       b_x^\pm\longrightarrow\beta,
\]
then its two raw terms are
\begin{align}
 \mathsf L(d)&=
 \Bigl((e_\alpha\bsum H_{I_x}\bsum
         \mathsf W(\xi_y^-,0;I_y))
       \circ\operatorname{Can}(\xi_x^-,0)\Bigr)
       \circ\operatorname{Can}(\xi_o^-,c^-),
 \label{supp:eq:state-left}\\
 \mathsf R(d)&=
 \Bigl((e_\alpha\bsum H_{I_x}\bsum
         \mathsf W(\xi_y^+,0;I_y))
       \circ\operatorname{Can}(\xi_x^+,0)\Bigr)
       \circ\operatorname{Can}(\xi_o^+,c^+).
 \label{supp:eq:state-right}
\end{align}
Both have type $\alpha+\alpha'\to\beta$. Equations
\eqref{supp:eq:state-left} and \eqref{supp:eq:state-right} are the expanded
$\mathsf W_2$ renderer. All identity padding and all permutations are fixed by
the placement maps in $\boldsymbol\iota$ and the order $\omega$. The three
minus partitions are the state partitions obtained from
Equation~\eqref{supp:eq:state-position-partition}. The plus partitions are the
same pushforwards after the following indicated update. This makes all six
partitions functions of the state, rather than additional choices.

There is one explicit target override. In a port-free emit, put
\[
 X^+(d)=
 \Bigl((e_\alpha\bsum H_{I_x}\bsum
        \mathsf W(\xi_y^+,0;I_y))
       \circ\operatorname{Can}(\xi_x^+,0)\Bigr)
\]
and render the plus endpoint as
\begin{equation}\label{supp:eq:state-nullary-emit-target}
 \mathsf R_\varnothing(d)=
 \bigl(X^+(d)\circ\operatorname{Can}(\xi_o^+,0)\bigr)\bsum z,
\end{equation}
not as a scalar hidden inside the last composition. This is the same local
value as $\mathsf R(d)$ with $c^+=1$. Indeed, fixed $e_0$ suppression
identifies $\operatorname{Can}(\xi_o^+,0)$ with its scalar-free router $Y$ and
$\operatorname{Can}(\xi_o^+,1)$ with $Y\bsum z$. For arbitrary matching types
$X:u\to v$ and $Y:v\to w$, use the fully raw extraction
\begin{equation}\label{supp:eq:state-nullary-extraction}
 X\circ(Y\bsum z)
 \ \longleftrightarrow\
 (X\bsum e_0)\circ(Y\bsum z)
 \ \longleftrightarrow\
 (X\circ Y)\bsum(e_0\circ z)
 \ \longleftrightarrow\
 (X\circ Y)\bsum z.
\end{equation}
The middle step is interchange and the others are structural-unit equations.
Consequently the emitted $z$ is already an external scalar factor before
$\mathsf{restore}$ or $\mathsf{suffix}$. Define
$\mathsf R_\phi(d)=\mathsf R_\varnothing(d)$ for this subcase and
$\mathsf R_\phi(d)=\mathsf R(d)$ for every other nonscalar transition.

The transition function takes the first applicable clause in this fixed list.
For a stable word bijection $\pi$, to \emph{transport $\lambda$ along $\pi$}
means
\[
 \lambda^+(\pi p)=p
 \quad\text{for an original position }p,
 \qquad
 \lambda^+(\pi p,\nu)=
   [\lambda^-(p,\nu)]_{\mathord{\equiv}_{F^+}}
 \quad(\nu\in\{\mathsf{in},\mathsf{out}\})
\]
for a tagged staging incidence $(p,\nu)$.  The tag is its incidence role at
the relevant router, not a claim that the two ends of the staging wire have
the same class. For a deletion or quotient this
formula is restricted to the surviving incidences. Every end of a new staging
wire is sent to the $\mathord{\equiv}_{F^+}$-class represented at that router.
Thus the following clauses
specify both the word coordinates of $\omega$ and its label coordinate
$\lambda$.
\begin{enumerate}
\item In phase $\mathsf{split}$, the recursive words for $\mu$ and $\delta$
separate the active block from its complement. The plus partitions are the
transport of the minus partitions by the recorded block permutation.
More precisely, $\omega^+=\operatorname{sep}_e(\omega^-)$ is the unique stable
permutation that lists, on each router side, the positions of the block
in the current $\equiv_F$-class of an endpoint of $e$ before its complement
and preserves the relative order
inside both subwords.
The set $F$ and $\equiv_F$ do not change, $c^+=c^-=0$, and $\lambda^+$ is
the transport of $\lambda^-$ along $\operatorname{sep}_e$.
\item In phase $\mathsf{expose}$, adjacent transpositions bring the two
occurrences belonging to
$e=e_j=([y_j]_{\eta_L},[x'_j]_{\eta_R})$ to the active
composition position. Let $p$ and $q$ be the current positions of the selected
$y_j$- and $x'_j$-occurrences in a middle bank of length $N$. The fixed word
for $\operatorname{adj}_e$ moves the selected $y_j$-occurrence until it is
immediately left of the selected $x'_j$-occurrence. In application order its
blocks are $\tau_{N,p},\ldots,\tau_{N,q-2}$ if $p<q$, with the empty word when
$q=p+1$, and $\tau_{N,p-1},\ldots,\tau_{N,q}$ if $p>q$. Every other occurrence
retains its relative order. The plus partitions are the transported minus
partitions, and the middle identity has type $1\to1$. Set
$\omega^+=\operatorname{adj}_e(\omega^-)$, $F^+=F$, $c^+=c^-=0$, and obtain
$\lambda^+$ by transporting $\lambda^-$ along $\operatorname{adj}_e$.
\item In phase $\mathsf{fuse}$, if the two ends of $e$ are in distinct
${\equiv}_F$-classes, the plus router replaces their two exposed spiders by
their single union spider and sets
\[
 F^+=F\cup\{e\},
 \qquad
 {\equiv}_{F^+}=\operatorname{tc}({\equiv}_F\cup\{e\}).
\]
If the ends are already equivalent, the plus router removes the redundant
incidence by the special Frobenius equation and sets
$F^+=F\cup\{e\}$ and ${\equiv}_{F^+}={\equiv}_F$.
In both subcases $c^+=c^-=0$, and
$\omega^+=\operatorname{quot}_e(\omega^-)$ deletes the two consumed middle
occurrences, replaces their two class names by the name of their new
equivalence class, and lists the remaining legs of that class in ambient
first-occurrence order. Explicitly, $\lambda^+$ fixes every surviving
original position and sends each surviving or new tagged staging incidence to
its $\equiv_{F^+}$-class at that router. In particular, the two labels on the
fused staging wire now agree. This is the quotient transport just defined.
\item In phase $\mathsf{emit}$, assume that every edge of $C$ has been
processed. If $C$ has retained-leg word
$\mathbf E(C)\ne\varnothing$, replace its connected spider network by the
one canonical block on that word, ordered by first occurrence, and take
$c^+=c^-=0$. The update $\omega^+=\operatorname{emit}_C(\omega^-)$ replaces
all active class names by this one ordered block. Its new staging
representative is labelled by the unique $\equiv_F$-class of $C$, and
each of its tagged incidences is labelled at its incident router.
$\lambda^+$ agrees with $\lambda^-$ on every unaffected tagged incidence and
original position. If
$\mathbf E(C)=\varnothing$, delete the active class names and replace the
component by
$z=\io01\circ\io10$, take $(c^-,c^+)=(0,1)$ in the displayed window, and set
$\epsilon^+=1$. In this subcase
$\mathsf R_\phi(d)=\mathsf R_\varnothing(d)$ is the actual raw target, so the
$c^+=1$ tag records the emitted component but the displayed $z$ is external.
$\omega^+=\operatorname{emit}^{\varnothing}_C(\omega^-)$ restricts every
nonscalar word coordinate to the positions outside $C$, in stable order.
$\lambda^+$ is the corresponding restriction of $\lambda^-$, and the new
$z$ is placed at the separate address $\mathfrak a_z^+$. The address $\mathfrak a^+$
addresses the ensuing restore window when it is nonvacuous. If restore is
vacuous, it is set to the first required suffix window when
$\kappa\ne(0,0,0)$. For a zero cursor, the component completes without using
$\mathfrak a^+$. In the retained-leg
subcase set
$\epsilon^+=0$, $\mathfrak a_z^+=\bot$, and $\kappa^+=(0,0,0)$. Every
graph-bearing positive-rank carrier contributes a position to
$\mathbf P_L\mathbf P_R$, while every declared graph-bearing rank-zero hole
or already reboxed scalar carrier is in $S$. Thus this active subcase contains
no graph-bearing carrier,
although inactive $D$-only complementary carriers may remain. It is the only
clause that creates $z$.
\item In phase $\mathsf{restore}$, use the inverse recorded permutation to put
the emitted block and all carriers into their prescribed order. The partition,
$F$, and the scalar count do not change. Thus
$\omega^+=\operatorname{sort}_K(\omega^-)$ is the unique stable sort by the
ambient parent order, and $\lambda^+$ is the transport of $\lambda^-$ along
that sort. Here $c^+=c^-=0$ in both subcases: after a port-free emit the
external $z$ belongs to $\operatorname{Out}_{\mathfrak s}[-]$, rather than to
the nonscalar restore window. The restore phase preserves $\epsilon$ and
transports $\mathfrak a_z$ with the unchanged outer context.
\item In phase $\mathsf{suffix}$, if a copy of $z$ was created, keep every
nonscalar carrier other than the next neighbor in the outer context and use
the one-scalar specialization of Equation~\eqref{supp:eq:scalar-move}.  The
external emit target initially has the fixed parse
$(P\bsum z)\bsum R$ whenever a nonunit prefix $P$ and trailing remainder
$R$ meet at the insertion boundary.  In that case first use the explicit
suffix-entry pair
\begin{equation}\label{supp:eq:state-suffix-entry}
 (P\bsum z)\bsum R\quad\longleftrightarrow\quad
 P\bsum(z\bsum R),
 \qquad P:p\longrightarrow q,\quad R:r\longrightarrow s.
\end{equation}
This is a fixed reassociation. It decreases $\kappa_{\rm ent}$ from one to zero
and does not decrease either other cursor coordinate.  If the source already
has the right-associated form, $\kappa_{\rm ent}=0$ and no administrative
state is created.  Once $\kappa_{\rm ent}=0<\kappa_{\rm car}$, move $z$
across exactly one neighboring nonscalar carrier
by
\begin{equation}\label{supp:eq:state-scalar-step}
 z\bsum(g\bsum R)\quad\longleftrightarrow\quad
 g\bsum(z\bsum R),
 \qquad g:p\longrightarrow q,\quad R:r\longrightarrow s.
\end{equation}
Then move it through the remaining recorded scalar factors one at a time by
\begin{equation}\label{supp:eq:state-scalar-order}
 z\bsum(s\bsum R)
 \quad\longleftrightarrow\quad
 s\bsum(z\bsum R),
 \qquad s,R:0\longrightarrow0.
\end{equation}
The already traversed prefix belongs to the linear outer context. Here $s$ is the
next recorded scalar factor and $R$ is the opaque remainder. If either mark
hides positive-rank original holes, it is recorded as a rank-zero carrier in
$\mathbf g$ as specified above.
This step is used only when
$\kappa_{\rm ent}=\kappa_{\rm car}=0$ and $\kappa_0>0$. Write
$\mathbf u_K=(s_1,\ldots,s_t)$ and put $c=c(\mathcal D)$. After
$j=t-\kappa_0$ scalar-order steps, the scalar word
is, with fixed right association,
\begin{equation}\label{supp:eq:state-pending-suffix}
 \operatorname{RA}\bigl((s_1,\ldots,s_j);
   z\bsum\operatorname{RA}((s_{j+1},\ldots,s_t);Z_c)\bigr).
\end{equation}
One carrier move decreases $\kappa_{\rm car}$ by one. One scalar-order move
decreases $\kappa_0$ by one.  After each of the three suffix pairs, including
entry, set $S^+$ to the maximal right-associated scalar suffix obtained by
scanning $\tau^+$ at the distinguished suffix address.  Thus the $S$
coordinate is defined on every successor and carries no independent choice.
The right endpoint of the last required move is,
with the fixed association displayed above, already the exact suffix
\begin{equation}\label{supp:eq:state-finished-suffix}
 S^+=\operatorname{Suf}(\mathbf u_K,c+1),\qquad \epsilon^+=0.
\end{equation}
If all three cursor coordinates are already zero on entry, this is an
administrative completion with no raw rewrite: the current external factor is
literally $z\bsum Z_c=Z_{c+1}$ after the fixed unit suppression. No
$\mathsf{suffix}$ state with $\kappa=(0,0,0)$ is created.
Thus $z$ finishes after all recorded scalar factors and before the earlier
$z$ factors. Only the new $z$, one neighboring factor, and one opaque
remainder are displayed in any step. Throughout a suffix step $F$,
$\equiv_F$, $C$, $e$, $\omega$, and $\lambda$ are unchanged. No nonscalar
record is formed, so $c^\pm$ is inapplicable. The address $\mathfrak a$ is
advanced to the next neighboring carrier or scalar factor after each move,
and $\mathfrak a_z$ is the address of the moved copy. When the last required
move finishes, set $\mathfrak a_z^+=\bot$.
\end{enumerate}
For every nonterminal state, define its displayed raw pair
$(U^-_{\mathfrak s},U^+_{\mathfrak s})$ as follows. In a nonscalar phase it is
\[
 (U^-_{\mathfrak s},U^+_{\mathfrak s})
   =(\mathsf L(d_{\mathfrak s}),\mathsf R_\phi(d_{\mathfrak s})).
\]
In a suffix phase it is the left and right side of
Equation~\eqref{supp:eq:state-suffix-entry} while
$\kappa_{\rm ent}=1$, the two sides of
Equation~\eqref{supp:eq:state-scalar-step} while
$\kappa_{\rm ent}=0<\kappa_{\rm car}$, and the two sides of
Equation~\eqref{supp:eq:state-scalar-order} when
$\kappa_{\rm ent}=\kappa_{\rm car}=0<\kappa_0$. For every reachable
nonterminal state, the literal carrier-address invariant is
\begin{equation}\label{supp:eq:carrier-address-invariant}
 \operatorname{cut}(\tau,\mathfrak a)
   =(\operatorname{Out}_{\mathfrak s}[-],V_{\mathfrak s}),
 \qquad
 V_{\mathfrak s}=\sigma_{\mathfrak s}(U^-_{\mathfrak s}),
 \qquad
 G_i=V_{\mathfrak s}|_{\mathfrak b_i}\quad(1\le i\le m).
\end{equation}
In particular, the substituted renderer source is a literal subterm
occurrence, not merely a raw expression with the same graph value. The
renderer target before successor exposure is
\begin{equation}\label{supp:eq:state-term-update}
 \widetilde\tau^+:=\operatorname{Out}_{\mathfrak s}
      [\sigma_{\mathfrak s}(U^+_{\mathfrak s})].
\end{equation}
Let $\mathcal I^+$ be the ordered interval family determined by the successor
carrier grouping. Lemma~\ref{supp:lem:carrier-exposure} gives the explicit
structural derivation
\begin{equation}\label{supp:eq:state-successor-exposure}
 \widetilde\tau^+
 \xleftrightarrow{*}_{\mathcal B}
 \tau^+:=\operatorname{Exp}_{\mathcal I^+}(\widetilde\tau^+).
\end{equation}
Every elementary reassociation in this derivation is inserted into the
normalization sequence. If the successor is a scalar phase,
$\mathcal I^+=\varnothing$ and the derivation is reflexive. For the terminal
successor, its target is the prescribed fixed parenthesization of
$\mathsf N(K)$. After successor exposure, set $S^+$ to the maximal
right-associated scalar suffix read from $\tau^+$ at the distinguished suffix
address. Thus every transition acts on an actual syntax tree, and the scalar
coordinate is total and deterministic.

The successor state is fixed as follows. For a nonisolated component, the
phase word for its first edge is
$\mathsf{split},\mathsf{expose},\mathsf{fuse}$. After a fuse, the next edge of
the component starts again in phase $\mathsf{split}$. After its last edge, the
next phase is $\mathsf{emit}$. An isolated component also starts in
$\mathsf{emit}$. After emit, use $\mathsf{restore}$ exactly when
$\operatorname{sort}_K$ has a nonempty nonscalar window. If all surviving
banks are empty, or the fixed sort is vacuous after unit suppression, omit
that phase. In particular, when all outer factors also unit-suppress, the
composite $\io01\circ\io10:0\to0$ goes directly from its port-free emit
target $z$ to component completion. Otherwise the first successor is the
suffix-entry or scalar window below, never a vacuous restore.

At the last nonvacuous one of emit and restore, if $\epsilon=1$, set
\[
 \kappa=\bigl(
 \mathbf 1_{\rm ent},
 |\{\text{nonscalar carriers strictly between the new }z\text{ and }S\}|,\,
 |\mathbf u_K|\bigr).
\]
$\mathbf 1_{\rm ent}=1$ precisely when the fixed current parse at the pending
scalar boundary is $(P\bsum z)\bsum R$ with nonunit $R$. Otherwise it is
zero. If $\kappa\ne(0,0,0)$, reset $\mathfrak a$ to that suffix-entry subterm
when $\kappa_{\rm ent}=1$, and otherwise to the smallest right-associated
subterm $z\bsum(g\bsum R)$ or $z\bsum(s\bsum R)$ beginning at the $z$ at
$\mathfrak a_z$. Here $R$ is the opaque remainder. Enter $\mathsf{suffix}$.
A suffix phase first applies Equation~\eqref{supp:eq:state-suffix-entry} when
$\kappa_{\rm ent}=1$, then repeats
Equation~\eqref{supp:eq:state-scalar-step} until $\kappa_{\rm car}=0$, and
then repeats
Equation~\eqref{supp:eq:state-scalar-order} until $\kappa_0=0$.
The component is complete after the last nonvacuous emit or restore when
$\epsilon=0$, immediately when $\epsilon=1$ and $\kappa=(0,0,0)$, or after the
final suffix move. The zero-cursor case uses the administrative
identification in Equation~\eqref{supp:eq:state-finished-suffix} and introduces
no state without an active raw window. At component completion set
\[
\mathcal D^+=\mathcal D\,C
\]
and activate the first component of $\mathcal C_K$ outside $\mathcal D^+$.
This also marks an isolated component as processed. If the next component has
an edge, set $e$ to its first edge and $\phi=\mathsf{split}$. If it is
isolated, set $e=\bot$ and $\phi=\mathsf{emit}$. If there is no next
component, set $C=e=\bot$ and $\phi=\mathsf{done}$. In all three cases set
$\epsilon=0$, $\mathfrak a_z=\bot$, and $\kappa=(0,0,0)$. This specifies one successor at every
nonterminal state.
Put
\[
 \widetilde\tau_0=((\widehat L\circ\widehat R)\bsum S_0).
\]
Let $\mathcal I_0$ be the ordered interval family of the first active window.
When there is an active window, the initial state has
$\mathcal D=\varnothing$, $F=\varnothing$, $S=S_0$, and
\[
 \tau=\operatorname{Exp}_{\mathcal I_0}(\widetilde\tau_0).
\]
It also has $\epsilon=0$, $\mathfrak a_z=\bot$, $\kappa=(0,0,0)$, the child
block order $\omega$, the first ordered component active, and phase
$\mathsf{split}$, or
$\mathsf{emit}$ when that component is isolated. When $b=0$ and blocks are
present, the first state is the emit state of the first isolated block. When
there is no block, the initial state is already terminal. The initial exposure
derivation from $\widetilde\tau_0$ is included in the normalization sequence
whenever it is present. The terminal state
has $\mathcal D=\mathcal C_K$, $F=E(\Gamma_K)$, no active component,
$\epsilon=0$, $\mathfrak a_z=\bot$, $\kappa=(0,0,0)$, and phase $\mathsf{done}$. Put
\[
 \rho_*=\operatorname{tc}\bigl(\eta_L\cup\eta_R\cup
 \{\{y_j,x'_j\}:1\le j\le b\}\bigr).
\]
Its position partition is
$\rho_*|_{\mathbf E_K}$. Its suffix is
\[
 \operatorname{Suf}\left(
 \mathbf u_K,\,
 c_L+c_R+
 |\{A\in\rho_*:A\cap\mathbf E_K=\varnothing\}|\right).
\]
Its raw-term coordinate is
\begin{equation}\label{supp:eq:state-terminal-term}
 \tau=\mathsf N(K):=
 \Bigl((e_a\bsum H_K^+)\circ
       \operatorname{Can}(\rho_*|_{\mathbf E_K},0)\Bigr)
 \bsum
 \operatorname{Suf}\left(
 \mathbf u_K,\,
 c_L+c_R+
 |\{A\in\rho_*:A\cap\mathbf E_K=\varnothing\}|\right).
\end{equation}
In the no-block case $\mathcal C_K=\varnothing$ and the displayed set of lost
classes is empty, so the initial terminal state keeps $S=S_0$ literally and
performs no transition.
Thus $\mathsf N(K)$ is exactly the parent normal form for the selected
scalar-factor record, with every internal middle position removed and every
lost class retained as one $z$. In the generic specialization
$\mathbf u_K=\mathbf h_K^0$, so this is the normal form used in the typed
coverage statement below.

\medskip
\noindent\emph{State-realization and successor lemma.}
Let
$\mathfrak s_0,\ldots,\mathfrak s_N$ be the deterministic run specified
above. For every $j$, $\operatorname{Term}(\mathfrak s_j)$ is a fully expanded
raw term over $D$ and the original holes, of type $a\to c$. Its marked carrier
occurrences have precisely the substitutions $\mathbf G_j$. On positions
still present in the active and unprocessed router banks, its recorded
partition is the corresponding restriction of $\rho_{F_j}$. For every
$C'\in\mathcal D_j$, the emitted block realizes
$\rho_{F_j}|_{\mathbf E(C')}$, or is its external $z$ when
$\mathbf E(C')=\varnothing$. The active $C_j$ may already have its canonical
block or pending external $z$ after $\mathsf{emit}$, but it enters
$\mathcal D_j$ only after restore and suffix completion. Its port-and-carrier
order is $\omega_j$, and its scalar word is the one in
Equation~\eqref{supp:eq:state-invariant}. Every nonterminal state satisfies
the literal carrier-address invariant
\eqref{supp:eq:carrier-address-invariant}. When $N>0$, the initial exposure satisfies
\begin{equation}\label{supp:eq:state-realization-initial}
 \widetilde\tau_0
 \xleftrightarrow{*}_{\mathcal B}
 \operatorname{Term}(\mathfrak s_0).
\end{equation}
For $0\le j<N$, define the unexposed renderer target
\[
 \widetilde K_{j+1}:=
 \operatorname{Out}_{\mathfrak s_j}
       [\sigma_{\mathfrak s_j}(U^+_{\mathfrak s_j})].
\]
Then
\begin{align}
 \operatorname{Term}(\mathfrak s_j)
   &=\operatorname{Out}_{\mathfrak s_j}
       [\sigma_{\mathfrak s_j}(U^-_{\mathfrak s_j})],
   \label{supp:eq:state-realization-left}\\
 \widetilde K_{j+1}
   &=\operatorname{Out}_{\mathfrak s_j}
       [\sigma_{\mathfrak s_j}(U^+_{\mathfrak s_j})]
   \label{supp:eq:state-realization-right}
\end{align}
and the deterministic successor exposure gives
\begin{equation}\label{supp:eq:state-successor-glue}
 \widetilde K_{j+1}
 \xleftrightarrow{*}_{\mathcal B}
 \operatorname{Term}(\mathfrak s_{j+1}).
\end{equation}
If $j+1<N$, the right endpoint has the next literal carrier addresses and is
\[
 \operatorname{Term}(\mathfrak s_{j+1})
 =\operatorname{Out}_{\mathfrak s_{j+1}}
   [\sigma_{\mathfrak s_{j+1}}(U^-_{\mathfrak s_{j+1}})].
\]
Finally,
$\operatorname{Term}(\mathfrak s_N)=\mathsf N(K)$ from
Equation~\eqref{supp:eq:state-terminal-term}.

\emph{Proof.}
If $N=0$, unit suppression identifies the initial term with
\eqref{supp:eq:state-terminal-term}, and there is no active cut or successor
to prove. Assume $N>0$.
The initialization
\eqref{supp:eq:state-suffix-initialization}--\eqref{supp:eq:state-initial-suffix}
gives $\widetilde\tau_0$. Apply
Lemma~\ref{supp:lem:carrier-exposure} to the ordered interval family of the
first active window. This proves
\eqref{supp:eq:state-realization-initial}. The resulting fixed parse tree has
the first active two-router or isolated-component occurrence at the declared
address. Its carrier occurrences are literal subterms at the pairwise
incomparable addresses $\boldsymbol{\mathfrak b}_0$. Reading the six banks and
replacing precisely those occurrences by the formal marks $g_i$ gives
$d_{\mathfrak s_0}$. Substituting $G_i$ back reconstructs the addressed
occurrence literally. This proves
\eqref{supp:eq:carrier-address-invariant} and
\eqref{supp:eq:state-realization-left} at $j=0$.

Assume the invariants hold at $j$. Equation~\eqref{supp:eq:state-term-update}
replaces that one literal occurrence and leaves
$\operatorname{Out}_{\mathfrak s_j}[-]$ unchanged. Its result is
$\widetilde K_{j+1}$, which proves
\eqref{supp:eq:state-realization-right}. In the split and expose cases, the
recorded stable permutation gives exactly the bank order read by the next
phase. In a fuse, the quotient deletes precisely the two exposed middle
occurrences and the new labels are exactly the classes of $\rho_{F^+}$. A
retained emit replaces the completed component by its ordered canonical block.
A port-free emit instead uses
Equation~\eqref{supp:eq:state-nullary-emit-target}. By
\eqref{supp:eq:state-nullary-extraction} its new $z$ is a literal external
factor. Hence any following restore window excludes that $z$. If the restore
window is vacuous it is skipped, and the emit target itself is scanned for the
first suffix or next-component address. In either case, any later suffix address
first selects the literal $(P\bsum z)\bsum R$ of
\eqref{supp:eq:state-suffix-entry} when reassociation is required. Its target
contains the right-associated $z\bsum R$ at the next address.  It then selects
the exact three-factor window assumed in
\eqref{supp:eq:state-scalar-step} or
\eqref{supp:eq:state-scalar-order}. Restore performs the recorded stable
inverse permutation. Each suffix clause is itself a literal raw pair. Its
successor address is the next window in the displayed target, and its last required
move gives the suffix in
\eqref{supp:eq:state-finished-suffix}. A zero-cursor administrative completion
does not introduce an additional state. In these scalar cases the successor
exposure is reflexive.

None of these replacements changes an expanded carrier subterm $G_i$.
The stable bank update gives $\omega_{j+1}$ and transports the two tagged
incidences of every staging wire separately under $\lambda$. If the next phase
changes the maximal-sector grouping, this update first puts each required
carrier sublist into consecutive order. The order-compatible occurrence
partition in \eqref{supp:eq:state-invariant} then determines pairwise disjoint
intervals in the displayed target. Apply
Lemma~\ref{supp:lem:carrier-exposure} to those intervals. This gives the
derivation in \eqref{supp:eq:state-successor-exposure}, whose target is
$\operatorname{Term}(\mathfrak s_{j+1})$. Its exposed literal occurrences are
the new addresses $\boldsymbol{\mathfrak b}_{j+1}$ and their addressed
subterms are $\mathbf G_{j+1}$. Replacing them by the new marks and substituting
them back reconstructs the target literally. Therefore the deterministic next
address cuts out exactly
$\sigma_{\mathfrak s_{j+1}}(U^-_{\mathfrak s_{j+1}})$. This proves both
\eqref{supp:eq:carrier-address-invariant} at the successor and
\eqref{supp:eq:state-successor-glue} and closes the induction, including the
case in which component completion merely selects the first address of the
next component.

At termination $F=E(\Gamma_K)$. Hence the surviving partition is
$\rho_*|_{\mathbf E_K}$, while the emit/suffix clauses have produced exactly
one external $z$ for every $\rho_*$-class missing $\mathbf E_K$. Stable restore
has put $H_K^+$ and the outer banks into declared order. The terminal exposure
restores the prescribed fixed parenthesization, so the fully expanded term is
exactly $\mathsf N(K)$. \hfill$\square$

For the run, write
$\sigma_j=\sigma_{\mathfrak s_j}$ and
$U_j^\pm=U_{\mathfrak s_j}^\pm$, and define only now
\[
 K_j:=\operatorname{Term}(\mathfrak s_j),
 \qquad
 \mathcal C_j[-]:=\operatorname{Out}_{\mathfrak s_j}[-].
\]
Equations~\eqref{supp:eq:state-realization-left}--
\eqref{supp:eq:state-successor-glue} give the noncircular concatenation
\begin{align}
 K_j&=\mathcal C_j[\sigma_j(U_j^-)],
 \label{supp:eq:defined-Kj-left}\\
 \widetilde K_{j+1}&=\mathcal C_j[\sigma_j(U_j^+)],
 \label{supp:eq:defined-Kj-right}\\
 \widetilde K_{j+1}&
 \xleftrightarrow{*}_{\mathcal B}K_{j+1}.
 \label{supp:eq:defined-Kj}
\end{align}
These relations hold for every $0\le j<N$. If $j+1<N$, then additionally
\[
 K_{j+1}
 =\mathcal C_{j+1}[\sigma_{j+1}(U_{j+1}^-)]
\]
by Equation~\eqref{supp:eq:state-successor-glue}. The final transition has
$K_N=\mathsf N(K)$. For a nonscalar phase this says
$U_j^-=\mathsf L(d_j)$ and
$U_j^+=\mathsf R_{\phi_j}(d_j)$. For a suffix phase the $U_j^\pm$ are the
explicit scalar endpoints. Thus every renderer source is a literal subterm.
Its target reaches the next literal renderer source through the explicit
$\mathcal B$-reassociation in Equation~\eqref{supp:eq:defined-Kj}. Concatenating
the retained renderer derivations with these carrier-exposure derivations
gives one raw derivation. No equality of graph values is used to glue
successive steps.

The syntax recursion is exhaustive. At $L\bsum R$, first apply stable
order-straightening when a required carrier sublist is not consecutive, then
use Lemma~\ref{supp:lem:carrier-exposure} to expose each nonempty child
aggregate as a literal carrier. Concatenate their recorded internal hole
sublists to form the parent aggregate, and take the transported disjoint union
$\eta_L\dot\cup\eta_R$. Then put the two suffixes into declared order by
scalar moves. At $L\circ R$, run the state machine above. If both children
have nonempty positive-rank nonscalar aggregates, both are initially present.
If exactly one has such an aggregate, the other is empty and the output-only
or input-only blocks of its discrete router remain explicit. If neither does,
$\mathbf g$ is empty and the same machine normalizes one partition block at a
time. Rank-zero graph-bearing holes, reboxed graph-bearing scalar carriers,
and children whose graph-bearing content consists only of such factors occur
in $S_L$ or $S_R$ and are handled by the suffix initialization, not by the
incidence machine. An actual
unit $e_0$ is suppressed. A positive-rank discrete child is not suppressed,
and an occupied hole of type $0\to0$ is a labeled scalar carrier, never a
unit. The base states are a positive-rank hole, a rank-zero hole, a letter of $D$, and a
structural unit. Hence the two graph-bearing, both mixed, two discrete,
unit-child, empty-bank, input-only, output-only, and scalar-nullary cases are
all determined without an implicit choice.

Finally, the local record is bounded. First consider a nonscalar datum. If at
most $r$ protected and outer positions are present, every displayed carrier
has interface load at most $k_0=\max\{r,w\}$, and
\begin{equation}\label{supp:eq:state-load-audit}
 r_i\le k_0,
 \qquad
 \sum_{i=1}^m r_i\le4k_0,
 \qquad
 T:=r+\sum_{i=1}^m r_i\le r+4k_0,
 \qquad
 2T+4\le2(r+4k_0)+4.
\end{equation}
Name the six source slots
\[
 \mathcal S_1=\mathbf A,\quad
 \mathcal S_2=\mathbf P(I_x),\quad
 \mathcal S_3=\mathbf A',\quad
 \mathcal S_4=\mathbf P(I_y),\quad
 \mathcal S_5=\mathbf B_y,\quad
 \mathcal S_6=\mathbf B_x\mathbf B.
\]
The base word
$\mathbf A\mathbf A'\mathbf P(I_x)\mathbf P(I_y)\mathbf B$ has length at
most $T$. Here the carrier words include every maximal unchanged
complementary sector meeting the window, whether or not that sector contains
an original graph-bearing hole. A $D$-only carrier has an empty original-hole
sublist. By construction of $\operatorname{sep}_e$, every nonactive class
occurring at an unprocessed middle incidence belongs to one such
complementary sector and therefore has a representative in
$\mathbf P(I_x)\mathbf P(I_y)$. More explicitly, for a nonactive class that
was middle-only, choose its least unprocessed occurrence-tagged middle
endpoint $(j,L)$ or $(j,R)$. The stable separation puts the sector containing
that endpoint behind a unique carrier-port representative $p_A$. Charge the
class to $p_A$. Distinct classes have disjoint tagged endpoints and distinct
carrier representatives. Charge every other exported staging class that
meets the base word to its least base position. The only exported class with
no such position is the single active output-only or port-free class, which is
charged to one dummy position. Distinct exported classes receive distinct
charges. Hence
\[
 |\mathbf B_y|\le T+1,\qquad
 |\mathbf B_x|\le T+1,\qquad
 |\mathbf B_x\mathbf B|\le2T+1.
\]
Thus, if $b_i:=|\mathcal S_i|$ (and $b_i=0$ for an absent slot), then
\[
 b_i\le2T+1\le2T+4
 \le2\bigl(r+4\max\{r,w\}\bigr)+4=B(r,w).
\]

Here is the charging map $\chi$ for the rendered count. Fix either endpoint.
An occurrence whose source is $p\in\mathcal S_i$ is sent to
$\chi(o)=(i,p,0)$ when it is native
and to $\chi(o)=(i,p,1)$ when it is repeated in a naming, staging, or
restoration list. The renderer uses each source position at most once at each
level, so $\chi$ is injective on these occurrences. Send the at most four
remaining auxiliary occurrences (the active identity, one nullary leg, and
two scalar-padding positions) injectively to distinct dummy targets
$(i,\ast,\nu)$ in
$\bigcup_{i=1}^6\{i\}\times(\mathcal S_i\cup\{\ast\})\times\{0,1\}$.
Since the map is injective,
\begin{equation}\label{supp:eq:state-six-slot-audit}
 \ell(d)\le2\sum_{i=1}^{6}b_i+4
 \le2\sum_{i=1}^{6}(b_i+1)
 \le12\bigl(B(r,w)+1\bigr)
 =12\left(2\bigl(r+4\max\{r,w\}\bigr)+5\right)=Q(r,w).
\end{equation}
The externalized nullary target and its extraction
\eqref{supp:eq:state-nullary-extraction} use the same six source slots. Their
$e_0$ and $z$ positions are among the four auxiliary occurrences already
charged above.
For a suffix-entry or carrier-crossing scalar datum, the window exposes at
most two addressed carrier marks. Each has total input-output interface load
at most $\max\{r,w\}$. There is no naming or staging copy in these
associativity and scalar-commutation pairs, so
\begin{equation}\label{supp:eq:state-suffix-load}
 \ell(d)\le2\max\{r,w\}\le Q(r,w).
\end{equation}
The scalar-order pair has port load zero.  Thus the six-slot argument is the
nonscalar audit, while Equation~\eqref{supp:eq:state-suffix-load} covers all
three suffix pairs.
Parse-specific reboxing does not alter these inequalities: every added scalar
factor has exterior type $0\to0$, and the width of its expanded substitution
is charged to the substituted-hole bound rather than to the rendered port
load.
Consequently the phase, types, placement injections, partitions, orders,
scalar roles, and raw endpoints of every local transition range over a finite
set for fixed $(r,w)$. The number of global tasks may depend on the syntax,
but no local record does. This completes the explicit renderer of
Definition~\ref{supp:def:renderer}.

The order-straightening family is a decidable $\mathsf W_2$ subfamily. Its
datum records a shared bank $\mathbf P$, a permutation $\sigma$, both incident
router partitions, the complementary hole placements, and the outer order.
The target postcomposes the child router with $\sigma$, precomposes the parent
with $\sigma^{-1}$, and uses bank $\sigma\mathbf P$. The raw words cancel, so
the protected and outer partition is unchanged. This is included in the
validity test and two-router arity audit below. It is separately enumerated
and is not a seventh window kind.

\begin{proposition}[validity of the six templates]
\label{supp:prop:renderer-validity}
Each pair in Definition~\ref{supp:def:renderer}, including every state-machine
transition specified there, has boundary-preserving
isomorphic graph values over the distinct-hole alphabet. Every displayed
formal hole occurs once on each side with the same ordered input-output
incidence list.
\end{proposition}

\begin{proof}
The $\Sigma$-leaf frame starts with yanking. The merge and repetition in
Equation~\eqref{supp:eq:reduced-leaf-core} act exactly within the
$\beta$-blocks. Hence $C_\beta[h]\circ\operatorname{Rep}_\beta$, the raw
name, and the direct frame have the same boundary-preserving value: each has
the same single occurrence of $h$, the same ordered incidence list, and the
same $\beta$-partition of the outer positions. This is a direct evaluation and
does not use bounded frame straightening. A hole-free discrete leaf has exactly
the partition and component count rendered by its canonical router. For a
binary pair, the typed partition $\theta_u$ performs precisely
the disjoint sum or positional quotient at the parent, closes exactly the
classes absent from the parent boundary, and restores the ordered parent
representatives. In the mixed case its output-only blocks create exactly the
boundary classes supplied solely by the discrete child.

In the state machine, split, expose, and restore only transport a typed
partition by fixed permutations and interchange. A fuse transition either
joins two distinct spider blocks along the displayed identity wire or removes
a redundant incidence within one block. These are respectively the fusion and
special Frobenius cases, and both preserve the partition after transitive
closure. At emit, a component with retained legs becomes its canonical spider.
A component without a retained leg is connected and port free, so its value is
one isolated vertex and the target $z$ is exact.
Equation~\eqref{supp:eq:state-nullary-extraction} is interchange plus unit, so
the actual external endpoint $\mathsf R_\varnothing(d)$ has the same holes,
ordered incidences, and component count as the internally tagged
$\mathsf R(d)$. The suffix-entry transition changes only association, and the
remaining suffix transitions change only the order and association of the
pending scalar component. Equations
\eqref{supp:eq:state-left} and \eqref{supp:eq:state-right} use the same carrier
list. Hence every aggregate carrier is untouched, occurs once on each side,
and retains its ordered incidence map.

A full fold pair has the same protected and outer partition on both sides and
the same scalar count. Its active primitive is associativity of merge when
both active blocks are nonnullary. The cases with a nullary block also use the
unit equation, commutativity, and permutation transport. A full prune pair
removes only an unflagged terminal routing piece. If the removed terminal
piece constitutes an entire component with no protected or outer port, the
increase from $c^-$ to $c^+$ retains its unique vertex as one copy of $z$.
Otherwise the scalar count is unchanged. The two
primitive prune blocks are the unit and counit equations.

The row and column terms use the same occupied-cell set, the same partition
$\eta$, the same ordered survivor list $\mathcal O$, and the same scalar tag
$c$. Every occupied positive-rank or rank-zero hole occurs once on each side.
An absent cell contributes only the structural unit. A rank-zero occupied hole
contributes the same labeled rank-zero formal box, hence the same closed
substituted value, on both sides.
Equation~\eqref{supp:eq:scalar-two-cell} is a contextual use of the one-scalar
instance of Equation~\eqref{supp:eq:scalar-move}. The two-scalar move changes only the order
and association of disjoint rank-zero components. The rank-zero composition
pair in Equation~\eqref{supp:eq:scalar-compose} is interchange padded by
$e_0$, followed by the unit laws. In an auxiliary
order-straightening pair, the raw words for $\sigma$ and
$\sigma^{-1}$ are adjacent after the shared bank is exposed, so they cancel
and leave both incident router partitions unchanged. All equalities hold over
the auxiliary alphabet before any core is substituted, and every formula
displays each formal hole once on each side.
\end{proof}

\subsection{Finite enumeration, coverage, and the arity audit}
\label{supp:sec:coverage-audit}

For a rendered endpoint $E$, let $\ell(E)$ be its complete
\emph{rendered port load}. It counts every named position in a formal hole
interface, in the ordered outer interface, and in every typed internal bank
recorded by $\operatorname{Name}$, $\operatorname{Unname}$,
$\operatorname{Frame}$,
$\operatorname{Can}$, $\mathsf W$, or $\mathsf W_2$. A position used again in
a naming, staging, or restoration list is counted again. For a datum $d$, put
\[
 \ell(d)=\max\{\ell(L(d)),\ell(R(d))\}.
\]
For a composition-state datum, $L(d)$ means $\mathsf L(d)$ and $R(d)$ means
the actual endpoint $\mathsf R_\phi(d)$, including the external target
\eqref{supp:eq:state-nullary-emit-target}. For a scalar-state datum they mean
the corresponding two sides of
\eqref{supp:eq:state-suffix-entry},
\eqref{supp:eq:state-scalar-step}, or
\eqref{supp:eq:state-scalar-order}.
Only occurrences in the recorded position lists are counted. The unnamed
intermediate wires inside the fixed recursive words for $\mu$, $\delta$, and
a permutation are controlled later by the endpoint-size bound.

Fix $q$. Enumerate the six window kinds, at most four ordered hole types, the
four-cell occupancy sets independently of the interface ranks, the ordered
outer input and output type, the parent operation when one is present, every
port-placement map, every typed input-output partition at each $\mathsf W$ or
$\mathsf W_2$ stage, every survivor and temporary block order, every
$\mathsf W_2$ bank permutation, every active-strand choice, scalar tags in
$\{0,1,2\}$, and the role of every rank-zero formal hole as a cell core,
active scalar factor, suffix factor, or carrier.
For the two-scalar kind also enumerate whether the displayed operation is
composition or parallel sum, as in
Equation~\eqref{supp:eq:scalar-compose}.
For each bounded well-typed pair of child routers, including a pair carrying
ordered aggregate boxes with their types and port incidence injections,
enumerate the local records $d$ of the finite composition-state machine in
Definition~\ref{supp:def:renderer}. The phase, minus and plus partitions,
active middle incidence, survivor order, scalar increment, and fixed placement
maps determine the raw endpoints by
Equations~\eqref{supp:eq:state-left}--\eqref{supp:eq:state-right}, with
\eqref{supp:eq:state-nullary-emit-target} as the port-free emit target.
Also enumerate the scalar pairs
\eqref{supp:eq:state-suffix-entry},
\eqref{supp:eq:state-scalar-step}, and
\eqref{supp:eq:state-scalar-order}. The extraction
\eqref{supp:eq:state-nullary-extraction} is part of the binary-router
state-machine family and is not a seventh window kind. The four
aggregate boxes of a nonscalar datum, or the at most two carriers and one or
two scalar holes of a scalar datum, are enumerated in their recorded
left-to-right order. In particular, the suffix-entry pair is a structural
associativity datum with two carrier slots of types $P:p\to q$ and
$R:r\to s$. It is one of the six existing window kinds.
A completed suffix is otherwise kept in the outer context. For a leaf, enumerate the
letters of the fixed finite alphabet $\Sigma$ whose
rank sum is at most $q$. Rename formal holes $h_1,\ldots,h_m$ in recorded
left-to-right order and take data modulo this type- and order-preserving
renaming. Retain only data with $\ell(d)\le q$.

Render both endpoints by Definition~\ref{supp:def:renderer}. Evaluate them over
the finite auxiliary alphabet $H$ consisting of the distinct formal holes with
their recorded ranks. Retain the pair exactly when a finite vertex bijection
preserves edge labels, directions, every hole's ordered input-output incidence
list, the ordered outer begin and end lists, and the port-free component count.
This test is finite. It gives $\mathcal M_q$. Since every recorded position
list has length at most $q$, all enumerated choices are finite.

For each retained pair, the completeness theorem
\cite[Cor.~3]{BK04}, applied to the edge alphabet $\Sigma\cup H$, supplies a
$\mathcal B\cup\E_{\Sigma\cup H}$-derivation before substitution. Dovetailed
proof search finds one effectively. After actual cores are substituted,
structural induction expands naturality for every composite box and converts
the derivation to one over $\mathcal B\cup\E_\Sigma$, as in the main text's
naturality-under-substitution lemma.

Let $P(q)$ be the largest raw size of a fixed permutation word on at most $q$
recorded positions, let $C(q)$ be the largest raw size of a canonical router
on at most $q$ recorded positions with scalar tag at most two, and let $K(q)$
be the largest raw size of $\mathsf L(d)$ or the actual
$\mathsf R_\phi(d)$ for a local state
record of load at most $q$. These functions are computable by finite
enumeration. The renderer recursions and the bounded local records therefore
give the computable endpoint bound
\[
 N(q)=\max\bigl(\{0\}\cup
 \{|L|,|R|:(L,R)\in\mathcal M_q\}\bigr).
\]

\Needspace{6\baselineskip}
\begin{lemma}[typed coverage of the renderer]
\label{supp:lem:renderer-coverage}
Let $C:p\to q$ be a raw linear-use context over $D$ and the structural units,
with a declared ordered list
$\mathcal H=(h_1,\ldots,h_m)$ of formal holes $h_i:0\to r_i$. Every hole with
$r_i>0$ has distinct protected outputs, and every hole with $r_i=0$ is retained
as a labeled graph-bearing scalar hole. Suppose that at most $r$ protected and
ordered outer positions occur and that, after the prescribed cores are
inserted, the resulting expression and every inserted core have pattern width
at most $w$. Put
\begin{equation}\label{supp:eq:coverage-load-bounds}
 B(r,w)=2\bigl(r+4\max\{r,w\}\bigr)+4,
 \qquad
 Q(r,w)=12\bigl(B(r,w)+1\bigr).
\end{equation}
Then $C$ has an effective bottom-up normalization in which every rendered
transition is a rank-preserving substitution instance, in a fixed outer
context, of a retained datum in $\mathcal M_{Q(r,w)}$. Consecutive rendered
transitions are interleaved only with the explicit ordered carrier-exposure
reassociations of Lemma~\ref{supp:lem:carrier-exposure}. A nonscalar transition
exposes at most four aggregate carriers. A scalar transition exposes at most
two carrier marks and at most two scalar holes; when a displayed scalar factor
is itself reboxed, the same linear-use formal hole is counted in both
descriptions. Suffix entry and a nonscalar-carrier crossing expose at most two
carrier slots and the single scalar $z$. Every rendered source occurs at a literal syntax-tree address, and every
carrier occurs once on each side
with the same type and ordered incidence map, and the normalization terminates
under a strictly decreasing lexicographic measure.

The same retained families give the three coverage properties used in the
main text: parse alignment, routing normalization, and four-cell switching.
Positive-rank edge-free children, nullary children, empty banks, isolated
vertices, and rank-zero graph-bearing cores are all preserved.
\end{lemma}

\begin{proof}
Fix the declared order of $\mathcal H$. For a rooted subcontext
$K:p_K\to q_K$, let $\mathbf x_K$ and $\mathbf y_K$ be its ordered input and
output position lists. Let $\mathbf h_K^+$ be the sublist of holes in $K$ with
positive rank and let $H_K^+$ be their fixed parenthesized parallel sum in the
declared order. Its ordered output list $\mathbf P_K$ consists of the pairs
$(h_i,j)$ with $h_i\in\mathbf h_K^+$ and $1\le j\le r_i$. Put
$H_K^+=e_0$ and $\mathbf P_K=\varnothing$ when the list is empty.

Regard every position in
\[
 \Omega_K=\mathbf x_K\,\mathbf P_K\,\mathbf y_K
\]
as labeled. Let $\eta_K$ be the partition of $\Omega_K$ induced by the ordered
incidence of the hole ports together with the $D$-wiring of $K$. The ordered
surviving interface $\mathbf o_K$ is the word
of $\eta_K$-blocks met by $\mathbf x_K$ and $\mathbf y_K$, with repetitions
and the input-output division retained. Let $c_K$ be the number of vertex
classes created by the $D$-wiring that meet no position of $\Omega_K$. Let
$\mathbf h_K^0$ be the declared-order sublist of rank-zero graph-bearing holes
in $K$, and let $S_K:0\to0$ be the fixed right-associated sum of
$\mathbf h_K^0$ followed by $c_K$ copies of $z$. The completed target at $K$
is
\begin{equation}\label{supp:eq:typed-discrete-normal-form}
 \mathsf N(K)=
 \bigl((e_{p_K}\bsum H_K^+)\circ
       \operatorname{Can}(\eta_K,0)\bigr)\bsum S_K
 :p_K\longrightarrow q_K,
\end{equation}
with the fixed structural-unit suppression when a list is empty. Thus the
complete induction record is
\[
 \mathfrak R_K=
 \bigl((p_K,q_K),\mathbf x_K,\mathbf y_K,\eta_K,
       \mathbf o_K,H_K^+,\iota_K,S_K\bigr),
\]
where $\iota_K$ sends each output of $H_K^+$ to its labeled position in
$\mathbf P_K$. This records the surviving ordered interface and its partition,
the aggregate-hole incidence, and the complete scalar suffix.

We prove simultaneously, by structural induction on $K$, that $K$ reaches
$\mathsf N(K)$ through retained windows and that $\mathfrak R_K$ is exact. A
positive-rank hole is already the required aggregate, up to the fixed unit and
identity suppression. A rank-zero hole is the one-element scalar suffix. A
generator of $D$ or a structural unit is rendered by its exact hole-free
canonical router. The cases $e_0$, $\mu_0$, and $\delta_0$ give the empty,
output-only, and input-only bases.

Suppose first that $K=L\bsum R$. After using the two induction sequences in
their fixed outer contexts, concatenate the protected-position lists and
transport them to the declared hole order. The parent partition is the
disjoint union of $\eta_L$ and $\eta_R$, transported by the same fixed
permutation, and $c_K=c_L+c_R$. The binary-router datum combines the two
canonical routers. One-scalar and two-scalar data move the child suffixes to
the right and put their factors in the declared order, exactly as in the
parallel-parent clause following
Equation~\eqref{supp:eq:state-suffix-initialization}. An already completed
part of a suffix is one opaque rank-zero carrier. Empty children contribute
only $e_0$. These operations give exactly $\mathsf N(K)$.

Now let $K=L\circ R$, where $L:a\to b$ and $R:b\to c$. Write
$\mathbf y_L=(y_1,\ldots,y_b)$ and
$\mathbf x_R=(x'_1,\ldots,x'_b)$. On the disjoint union of the child position
sets, form
\[
 \zeta_K=\operatorname{tc}\left(
   \eta_L\cup\eta_R\cup
   \bigl\{\{y_j,x'_j\}:1\le j\le b\bigr\}
 \right).
\]
Then
\begin{equation}\label{supp:eq:composition-record}
 \eta_K=
 \left.\zeta_K\right|_{\mathbf x_L\,\mathbf P_L\,
                              \mathbf P_R\,\mathbf y_R},
 \qquad
 c_K=c_L+c_R+
 \bigl|\{A\in\zeta_K:A\cap
 (\mathbf x_L\mathbf P_L\mathbf P_R\mathbf y_R)=\varnothing\}\bigr|.
\end{equation}
A class counted in the last term contains no protected or outer position, so
it is a $D$-only port-free vertex. Conversely, every new port-free vertex is
one such class. Equation~\eqref{supp:eq:composition-record} therefore records
both the exact parent partition and the exact scalar increment.

First apply Equation~\eqref{supp:eq:state-suffix-initialization} and the
scalar-order moves to obtain $S_0$ in
Equation~\eqref{supp:eq:state-initial-suffix}. Then run the finite composition-state machine of
Definition~\ref{supp:def:renderer} on the incidence graph $\Gamma_K$.
After its first $j$ processed middle edges, its equivalence
relation is exactly the transitive closure of the first $j$ gluing pairs.
Therefore its terminal relation, restricted to
$\mathbf x_L\mathbf P_L\mathbf P_R\mathbf y_R$, is $\eta_K$ in
Equation~\eqref{supp:eq:composition-record}. The emit rule produces one $z$
for exactly those terminal components lost under this restriction. Hence the
terminal scalar suffix contains precisely the additional term in the formula
for $c_K$.

For every nonscalar microstep, Equation~\eqref{supp:eq:state-left} and the
actual target $\mathsf R_\phi(d)$ give the raw source and target, including
their types and identity padding. The state-realization identities
\eqref{supp:eq:state-realization-left}--\eqref{supp:eq:defined-Kj} insert that
pair into its uniquely cut outer context. Its unexposed target reaches the
next literal source through
\eqref{supp:eq:state-successor-exposure}. The state invariant
\eqref{supp:eq:state-invariant} preserves the disjoint ordered sublists of
original holes behind at most four carriers. If all four carriers are needed,
the suffix remains in the outer context and the separate step
\eqref{supp:eq:state-suffix-entry}, followed when necessary by
\eqref{supp:eq:state-scalar-step} and then
\eqref{supp:eq:state-scalar-order}, moves a newly produced $z$. A rank-zero
graph-bearing carrier is retained as a labeled hole and is never identified
with $e_0$ or $z$.

We now verify membership in the retained family without using completeness or
the relative-spider lemma. For each generated transition, the local datum
records its window kind, outer type, carrier types, placement and incidence
maps, source and target partitions, survivor order, active phase, and scalar
roles. The induction record separately tracks the ordered original-hole
sublist hidden in each carrier. These are among the candidates in the
preceding finite enumeration because the hidden sublists are properties of
the later substitutions, not choices in a renderer datum. The deterministic
renderer gives the two raw endpoints. They are
well typed because each exposed composition bank is split into the same
ordered list on both sides and padded by identities of exactly the
complementary ranks. Every carrier occurs once in each endpoint with the same
type and ordered incidence map. Direct evaluation gives the same partition of
protected and outer positions, the same ordered outer interface, and the same
port-free component count. The increment in
Equation~\eqref{supp:eq:composition-record} is represented by the exposed
$z$. In the port-free case
\eqref{supp:eq:state-nullary-extraction} makes it external without changing
any carrier occurrence. Thus the endpoints have boundary-preserving isomorphic values over the
distinct-carrier alphabet and pass the validity test defining
$\mathcal M_q$.

It remains to check the complete rendered load. The explicit inequalities
\eqref{supp:eq:state-load-audit} and
\eqref{supp:eq:state-six-slot-audit} give
$\ell(d)\le Q(r,w)$ for every nonscalar state-machine transition. The same
estimate for a sum is smaller because it has no middle incidence, and
Equation~\eqref{supp:eq:state-suffix-load} covers suffix entry and the two
scalar moves, including their general-type carrier interfaces.
The datum consequently belongs to $\mathcal M_{Q(r,w)}$. The following
proposition gives the renderer-by-renderer version of this count.

For termination, use a fixed postorder traversal with one active node. For a
normalization state set
\begin{equation}\label{supp:eq:coverage-termination-measure}
 \Theta=(U,N,\phi,I,S_{\rm out},S_{\rm inv},S_{\rm ass})\in\mathbb N^7,
\end{equation}
ordered lexicographically. Here $U$ is the number of incomplete syntax nodes,
$N$ is the number of unprocessed incidence, fusion, redundancy, survival, and
closure tasks at the active node, and $\phi$ is the number of remaining
phases of the current task in the fixed order split, expose, fuse,
close-or-retain, and restore. The coordinate $I$ is the inversion number of
the current port and carrier order relative to the order required in that
phase. The last three coordinates count scalar factors outside the terminal
suffix, scalar-order inversions relative to the selected order $\mathbf u_K$
followed by the $z$ factors, and scalar right-association defects. An adjacent
order move decreases $I$. A completed microphase decreases $\phi$. Completion
of a block task decreases $N$, although lower coordinates may then reset.
Extracting a new $z$ occurs only in a move that decreases $N$ or $\phi$. The
subsequent scalar schedule decreases, in order, $S_{\rm out}$,
$S_{\rm inv}$, and $S_{\rm ass}$. Completion of a syntax node decreases $U$.
Thus every transition of the outer normalization decreases $\Theta$. The
formal derivation later selected for one retained datum is the finite
implementation of one such transition and need not decrease $\Theta$ at each
of its internal equation steps. A carrier-exposure derivation is also finite,
is executed only on entry to its uniquely determined successor, and need not
decrease $\Theta$ at each associativity rotation.

The exhaustive syntax clause of the same state machine handles all binary
cases. When both children have nonempty positive-rank nonscalar parts, both
child aggregates are present. When exactly one does, the other child's
output-only or input-only blocks remain in the state partitions. When neither
does, the carrier list is empty and the machine processes the same ordered
component list. A declared graph-bearing scalar hole is already in the
initialized suffix and remains an opaque labeled carrier. A positive-rank discrete router is
retained, while only an actual $e_0$ is suppressed. Thus empty banks and both
nullary router orientations are included, and only a class without a protected
or outer position becomes $z$.

For parse alignment, apply the preceding discrete normalization to every
maximal $D$-only subtree before pruning it. At an edge leaf use the yanking and
repetition pair to obtain its reduced named core. At a graph-bearing binary
node use the ordinary or mixed binary-router case just proved. Direct
boundary evaluation above establishes the $\Sigma$-leaf datum without using
relative-spider normalization or bounded frame straightening. Hence the
coverage proof is acyclic.

There is one additional bookkeeping operation for recursive parse alignment.
After a graph-bearing parse node $u$ reaches completed-step form, the next
parent window regards its completed nonsuffix part as one aggregate carrier
\[
 g_u:0\longrightarrow b_u,
 \qquad b_u=|\partial_G A_u|.
\]
The induction record retains the ordered sublist of original holes hidden in
$g_u$. In the scalar-word notation above, the state machine at $u$ enters with
the completed child words $\mathbf s_y,\mathbf s_z$ and has
$\mathbf u_u=\mathbf s_y\mathbf s_z$ in plane order. The completed record
passed to the parent is
\[
 (\widehat U_u,\mathbf s_u)=
 \begin{cases}
  (e_0,(g_u)\mathbf u_u),
    & \begin{aligned}[t]
      &b_u=0\text{ and the substitution for }g_u\\
      &\text{contains a graph-letter occurrence},
      \end{aligned}\\
  (g_u,\mathbf u_u),&\text{otherwise}.
 \end{cases}
\]
Thus a closed graph-bearing $g_u$ enters the child suffix even when its hidden
sublist contains positive-rank holes, while a pure structural $0\to0$
completion remains in the nonsuffix record. The new factor precedes the list
already inherited from that child. Existing suffix factors are not collapsed
into the new carrier.
This reboxing is metasyntactic and performs no raw rewrite. It is specific to
recursive parse alignment and does not change the original-hole
classification in Equation~\eqref{supp:eq:typed-discrete-normal-form}.

For routing normalization, expose the two endpoint routers and the at most
four complementary regions. A full $\mathsf W_2$ fold reduces one duplicate
strand, and a full $\mathsf W_2$ prune removes one unflagged terminal strand.
If a prune removes a whole port-free component, its scalar tag increases and
a separate scalar datum moves the resulting $z$ to the suffix. The sum of the
duplicate-strand multiplicities and terminal excess strands strictly
decreases. The process ends exactly when every skeleton cut carries its live
atom boundary. These fold, prune, and scalar records occur in the finite
enumeration, Proposition~\ref{supp:prop:renderer-validity} proves their
validity, and the six-slot count above gives $\ell(d)\le Q(r,w)$. Hence each
is a member of $\mathcal M_{Q(r,w)}$.

For four-cell switching, the simultaneous intersection recursion exposes
every occupied cell core before changing row grouping to column grouping.
The partitions $\eta_{i*}$, $\eta_{*j}$, $\eta_R$, and $\eta_C$ have the same
final partition, survivor order, and scalar tag. Thus
the corresponding record occurs in the finite enumeration and
Proposition~\ref{supp:prop:renderer-validity} proves its validity. When every
cell and staging load is at most $r$, the six-slot count gives
$\ell(d)\le12(r+1)\le Q(r,w)$ because $B(r,w)\ge r$. Therefore
Equations~\eqref{supp:eq:row} and \eqref{supp:eq:column} form a datum in
$\mathcal M_{Q(r,w)}$.
An absent cell suppresses its structural-unit leg. An occupied cell with an
empty position list remains a rank-zero graph-bearing hole. Equation
\eqref{supp:eq:scalar-two-cell} and the one-scalar renderer cover both
off-diagonal rank-zero degeneracies. This proves all three coverage claims.
\end{proof}

\Needspace{6\baselineskip}
\begin{proposition}[per-window arity audit]
\label{supp:prop:window-arity}
Suppose every incident skeleton bank and atom, and each carrier of a scalar
move when present, has rank sum at most $k$. For a four-cell datum, suppose the
four cell interfaces and every row, column, or outer staging bank have load at
most $r$. If every expression substituted for a formal hole has width at most
$w$, put $q=\ell(d)$ for the actual rendered port load. Then
\[
\begin{array}{@{}c|c|c@{}}
\text{window}&\text{rendered port-load upper bound}&\text{endpoint width}\\ \hline
\text{leaf frame}&4(k+1)&\max\{w,4q+4\}\\
\text{binary router}&6(k+1)&\max\{w,4q+4\}\\
\text{fold}&12(k+1)&\max\{w,4q+4\}\\
\text{prune}&12(k+1)&\max\{w,4q+4\}\\
\text{four-cell switch}&12(r+1)&\max\{w,4q+4\}\\
\text{two-scalar move}&k&\max\{w,2\}\\
\text{suffix entry/carrier crossing}&2k&\max\{w,2k\}
\end{array}
\]
In particular, every required parse, fold, prune, or scalar datum is
$12(k+1)$-local, and
every required four-cell datum is $12(r+1)$-local. The auxiliary
order-straightening family is a full two-router window and has the same
$12(k+1)$ bound. In the application to
Lemma~\ref{supp:lem:renderer-coverage}, every carrier slot has load at most
$k_0=\max\{r,w\}$.
\end{proposition}

\begin{proof}
For a fold or prune, use its full $\mathsf W_2$ renderer. Removing the two
endpoint routers of the active skeleton edge exposes at most six incident bank
slots. Let $b_i$ be the full protected-port load of slot $i$, including
multiplicity in its recorded ordered tuple, and pad absent slots by $b_i=0$.
Then $b_i\le k$. In either rendered endpoint, a slot occurs at most once in the
protected core or outer list and at most once more in a naming, staging, or
restoration list. The harmless $+1$ records a nullary spider, empty bank,
scalar flag, or atom endpoint. Hence
\begin{equation}\label{supp:eq:arity-audit}
 \ell(d)\le2\sum_{i=1}^{m}(b_i+1)
 \le2\cdot6(k+1)=12(k+1),
 \qquad m\le6.
\end{equation}
This is the complete rendered load of either endpoint. It is not a one-sided
rank bound, so no further factor two is introduced. Fixed input and output
permutations preserve the count. A fold or prune removes a protected position,
and a restoration copy has already been charged to the second occurrence of
its slot. The estimate therefore holds at every completed step, including the
nullary cases rendered by $\mu_0$ and $\delta_0$.

A leaf frame uses at most its atom slot and its outer slot, so the same
calculation gives $4(k+1)$. A binary router uses at most the two child slots
and the parent outer slot, giving $6(k+1)$. Scalar holes have rank zero. The
carrier of a two-scalar move is not copied, so its rendered load is at most its
rank sum $p+q\le k$. The suffix-entry associativity pair and the generalized
carrier-crossing pair have two linear-use carrier slots, so their rendered
load is at most $2k\le12(k+1)$.

For a four-cell endpoint, the four cell slots and the two row or column staging
banks give at most six slots. A present rank-zero cell hole has slot load zero
and remains a formal hole whose substituted width is already included in $w$.
Occupancy flags add no port positions. Every outer survivor is a restoration
occurrence of a staging position and is charged to the second copy. With
$b_i\le r$, Equation~\eqref{supp:eq:arity-audit} gives
$\ell(d)\le12(r+1)$. The typed augmented partitions in
Equations~\eqref{supp:eq:row} and \eqref{supp:eq:column} show that this count
includes all occupied hole interfaces, both first-stage banks, and the common
ordered outer interface. Absent cells and occupied rank-zero cells contribute
zero interface load.

Finally, direct inspection of $\operatorname{Name}$,
$\operatorname{Unname}$, $\operatorname{Frame}$, $\operatorname{Can}$,
$\mathsf W$, and
$\mathsf W_2$ shows that every fixed permutation or spider stage has input and
output rank at most $2q+2$. Its rank sum is at most $4q+4$. Substituted holes
contribute width at most $w$, giving $\max\{w,4q+4\}$. The scalar renderer has
no positive-arity staging and gives $\max\{w,2\}$.
\end{proof}

\subsection{Recursive normalization and exact routing endpoints}

\begin{lemma}[effective relative-spider normalization]
\label{supp:lem:relative-spider}
Let $C$ be a raw context over the discrete signature $D$ with opaque formal
holes. Suppose that every hole occurs once, every positive-rank hole is a
reduced named core, and at most $r$ protected core ports and ordered outer
ports occur. Assume that the expressions substituted for the holes and the
resulting instantiated context all have width at most $w$. When two contexts
are compared, assume that both instantiated endpoints satisfy this bound.
There is an effective bottom-up normalization of $C$ with the following
properties.
\begin{enumerate}
\item Every completed subtree is represented by its exact ordered port
partition, one canonical router, and a fixed right-associated scalar suffix.
\item Every local step is an instance of a datum in
$\mathcal M_{Q(r,w)}$ and has at most four linear-use holes. An already
accumulated scalar suffix is substituted for one opaque rank-zero carrier
hole. A suffix-entry or nonscalar-carrier-crossing step may instead expose two
carrier holes. A reboxed scalar factor is simultaneously a scalar hole and a
carrier for hidden-hole bookkeeping, but remains one linear-use hole. Every step
exposes at most one new $z$ and at most two scalar holes.
\item If two contexts use the same ordered formal holes, with matching holes
instantiated by the same labeled cores (including the same labeled scalar cores
at rank zero), and have the same protected-and-outer port partition, ordered
outer interface, and number of port-free discrete components, then their
substitutions are connected within width $a_\Sigma(r,w)$ for computable
$a_\Sigma$.
\end{enumerate}
No step duplicates, deletes, or identifies two protected ports inside a hole.
\end{lemma}

\begin{proof}
Put $q_*=Q(r,w)$, with $Q$ as in
Equation~\eqref{supp:eq:coverage-load-bounds}. Thus
$B(r,w)=B_D(r,w)$ and $q_*=Q(r,w)=Q_D(r,w)$ in the notation of the main
article. The discrete-normalization part
of Lemma~\ref{supp:lem:renderer-coverage} supplies a terminating sequence in
which every rendered step is an instance of a member of $\mathcal M_{q_*}$
and consecutive rendered steps are joined by the explicit structural
reassociations of Lemma~\ref{supp:lem:carrier-exposure}. Let
$\mathcal P_{r,w}\subseteq\mathcal M_{q_*}$ be the subset consisting of the
$D$-leaf, generalized binary-router, order-straightening, and scalar data that
can be produced by that induction. It is finite and effective. This selection
does not use the $\Sigma$-leaf frame, the four-cell switch, or bounded frame
straightening, so the dependency is acyclic.

For every $M=(\ell_M,r_M)\in\mathcal M_{q_*}$, direct pointed-value equality
is part of membership in $\mathcal M_{q_*}$. The completeness theorem over the
finite alphabet obtained by treating the carriers as letters supplies a
finite $\mathcal B\cup\E$ derivation $\Delta_M$, and dovetailed proof search
finds one. This simultaneous finite selection uses only the validity test and
completeness, not this relative-spider lemma, so it is noncircular. Define
\begin{equation}\label{supp:eq:relative-spider-return-width}
 L(r,w)=\max\{r,w,4Q(r,w)+4\}.
\end{equation}
The exact normal form in
Equation~\eqref{supp:eq:typed-discrete-normal-form} has width at most
$L(r,w)$. Indeed, the flat protected aggregate has width at most
$\max\{r,w\}$. Applying the per-window arity audit with substituted-hole
width $\max\{r,w\}$ bounds every completed rendered endpoint by
$\max\{r,w,4q_*+4\}=L(r,w)$.

For a selected $\Delta_M$, let $s_{\Delta_M}$ be the computable
substitution-width function supplied by the main text's
naturality-under-substitution lemma. Put
\[
 a_\Sigma(r,w)=
 \max\bigl(\{L(r,w)\}\cup
 \{s_{\Delta_M}(L(r,w)):M\in\mathcal M_{q_*}\}\bigr)
\]
and take monotone closure in $(r,w)$. An aggregate carrier and an opaque
scalar suffix substituted at any step have width at most $L(r,w)$. The local
target and its carrier-exposure derivation have width at most $L(r,w)$ by
Lemma~\ref{supp:lem:carrier-exposure}. Their outer type is unchanged, so no
ancestor rank changes. Thus implementation of every local transition has
width at most $a_\Sigma(r,w)$ independently of depth and derivation length.

Finally, Equation~\eqref{supp:eq:typed-discrete-normal-form} is determined by
the declared ordered holes, their protected-and-outer partition, the ordered
outer interface, and the number of $D$-only port-free vertices. The two
contexts in the statement therefore normalize to the same raw endpoint.
Reverse the second normalization to obtain the required derivation.
\end{proof}

\begin{lemma}[bounded frame straightening]
\label{supp:lem:frame-straightening}
Let $\beta$ be a partition of $p+q$ ordered boundary positions, let
$r=|\beta|$, and let $C:0\to r$ be a reduced named core of width at most $w$.
Then
\[
 \operatorname{Unname}_{p,q}
 [C\circ\operatorname{Rep}_\beta]
 \xleftrightarrow{*}_{\mathcal B\cup\E}
 \operatorname{Frame}_{\mathbf b_\beta,\mathbf e_\beta,
   \boldsymbol\nu_\beta,0}[C].
\]
The derivation has width bounded by a computable function of $p+q$ and $w$.
\end{lemma}

\begin{proof}
Regard $C$ as one protected hole. On the left, each external input is paired
with its bent copy and the repeated outputs are supplied by
$\operatorname{Rep}_\beta$. On the right, the direct canonical router joins
the external inputs and the outputs of $C$ by the same block-name partition.
Both contexts therefore induce the same equivalence relation on the protected
and outer positions, have the same ordered output list, and have no additional
component without a port. Lemma~\ref{supp:lem:relative-spider} applies.
A coarse explicit choice is
\[
 a_\Sigma\bigl(4(p+q+1),\max\{w,4(p+q)+4\}\bigr),
\]
which charges every bend, repetition, naming position, and restoration
position. This proves bounded straightening and confirms that
Equation~\eqref{supp:eq:canonical-frame} selects one raw representative.
\end{proof}

We now record the literal scalar-suffix compiler used by recursive parse
alignment. Fix a closed graph $G$, a rooted plane binary compiler tree $T$ on
$E(G)$, and boundary orders $\rho=(\rho_x)_x$ as in the main text. For every
compiler node $x$, define a body $B_x$ and a scalar word
$\mathbf s_x$ bottom-up. If $x$ is a leaf, put
\[
 P_x=R_x,
 \qquad \mathbf u_x=\varnothing.
\]
If $x$ has plane-ordered children $y,z$, put
\[
 P_x=(B_y\bsum B_z)\circ J_x,
 \qquad \mathbf u_x=\mathbf s_y\mathbf s_z.
\]
After the fixed unit suppressions, define
\begin{equation}\label{supp:eq:scalar-compiler-recursion}
 (B_x,\mathbf s_x)=
 \begin{cases}
  (e_0,(P_x)\mathbf u_x),
    & b_x=0\text{ and }P_x\text{ contains a graph-letter occurrence},\\
  (P_x,\mathbf u_x),&\text{otherwise}.
 \end{cases}
\end{equation}
Thus a newly completed graph-bearing scalar body precedes the scalar words
already inherited from its children. At the closed root define
\begin{equation}\label{supp:eq:scalar-compiler}
 R_{\mathrm{sc}}(G,T,\rho)=
 B_{\operatorname{root}(T)}\bsum
 \operatorname{RA}\bigl(
   \mathbf s_{\operatorname{root}(T)};Z_{i(G)}\bigr),
\end{equation}
with the fixed unit suppression. For
$A=((\io01\circ a)\circ\io10):0\to0$, the recursion sends
$((A_1\bsum A_2)\bsum A_3)$ to
$A_1\bsum(A_2\bsum A_3)$.

\begin{lemma}[routing invariant and scalar-suffix compiler endpoint]
\label{supp:lem:routing-endpoint}
Start with the parse-aligned completed-step realization produced in
Lemma~\ref{supp:lem:renderer-coverage}. Let $S$ be its rooted plane
skeleton after unary suppression and the prescribed order correction on each
collapsed chain. For every skeleton node $x$, let $\rho_x$ be its prescribed
surviving-bank order and put $\rho=(\rho_x)_x$. Let $V_{\mathrm{rt}}(G)$ be the set of final vertices represented by
a block of a node router, including vertices with no atom incidence. For every
$v\in V_{\mathrm{rt}}(G)$, form the routing multigraph $\mathcal R_v$ whose
vertices are the $v$-blocks of the node routers and whose edges are the
$v$-pieces carried by skeleton banks. Mark the atom incidences of $v$ as
flags. Let $\varpi_v:\mathcal R_v\to S$ map every router block to its skeleton
node and every routing edge to its skeleton edge. Then every $\mathcal R_v$ is
connected. The fold and prune sequence, with the completed-child reboxing
convention, terminates at exactly $R_{\mathrm{sc}}(G,S,\rho)$, including its
prescribed bank orders, graph-bearing scalar factors, and isolated-vertex
suffix. For every closed compiler datum $(G,T,\rho)$ there is, moreover, an
effective raw derivation
\[
 R_{\mathrm{sc}}(G,T,\rho)
 \xleftrightarrow{*}_{\mathcal B\cup\E}
 R(G,T,\rho).
\]
It uses scalar reassociation, scalar-carrier moves, and padded rank-zero
unit/interchange steps only. These add no pattern-width load: scalar exterior
types are $0\to0$, every positive-rank bank keeps its rank, and same-type
replacement preserves all ancestor ranks.
\end{lemma}

\begin{proof}
Consider the incidence graph obtained from all canonical router blocks and
all bank pieces before any fold. Its connected components are precisely the
vertex classes created by structural evaluation. Indeed, passage through a
bank joins the two copies of one piece, and membership in one router block is
exactly one local positional identification. These are all vertex
identifications performed by a sum or a composition node. Conversely, every
identification used by structural evaluation is recorded by one such bank
edge or router block. Hence two local pieces have the same image in $G$ if
and only if they lie in the same component of this incidence graph. Its fibre
over $v$ is $\mathcal R_v$, which proves connectedness.

Let $H_v$ be the minimal subtree of $S$ spanning the atom flags of $v$, with
$H_v=\varnothing$ when there is no flag. For a skeleton edge $e$, put
\[
 m_v(e)=|\varpi_v^{-1}(e)|.
\]
Connectedness gives $m_v(e)\ge1$ whenever $e\in H_v$. If two routing edges
over $e$ meet the same router block, a fold retains one edge, fuses the two
opposite blocks, and leaves every other incidence unchanged. The padded
identity in Equation~\eqref{supp:eq:padded-fold} realizes this operation and
reduces $m_v(e)$ by one.

After no fold remains, the map from $\mathcal R_v$ to the tree $S$ is locally
injective. A cycle or two distinct router vertices with the same image would
map a shortest path to a nontrivial reduced closed walk in $S$. Neither is
possible. Thus the image of $\mathcal R_v$ is a subtree. If $H_v$ is nonempty
and the image properly contains it, choose a leaf edge outside $H_v$. Its
terminal block has no atom flag, protected port, or outer port, so one of the
two padded prune identities removes it. Continue until the image is $H_v$.
If $H_v=\varnothing$, use the same prune identities until one zero-leg block
remains. Render that block as
$\operatorname{Can}(\varnothing,1)\leftrightarrow z$ and move the resulting
$z$ to the fixed scalar suffix.

The nonnegative integer
\[
 \Phi=\sum_{v\in V_{\mathrm{rt}}(G)}\sum_{e\in E(S)}
 \bigl(m_v(e)-\mathbf{1}_{e\in H_v}\bigr)
\]
decreases by one at each fold or prune. A move for one $v$ does not increase
any multiplicity for another vertex, so the process terminates. At its
endpoint a bank over $e$ contains one $v$-piece exactly when $e\in H_v$.
For an atom-incident vertex this says that one representative remains
precisely when its atom edges occur on both sides of the cut defined by $e$,
which is the compiler boundary rule. A flag-free vertex leaves no bank piece
and contributes exactly one scalar factor.

It remains to identify the raw endpoint. At every node, the left child bank
precedes the right child bank in the plane order. Blocks are processed in the
fixed final-vertex order, and the surviving bank is placed in the prescribed
order $\rho$ by the fixed adjacent-transposition word. After every fold or
prune, the renderer completes both endpoint routers as the fixed
$\operatorname{Can}$ terms before another window is opened. These data
determine the body $P_x$ in
Equation~\eqref{supp:eq:scalar-compiler-recursion}. When $b_x=0$ and $P_x$
contains a graph-letter occurrence, parse-specific reboxing records this body
as the first factor of the current scalar word. Otherwise it remains the body
$B_x$. Induction from the leaves
therefore produces exactly the records $(B_x,\mathbf s_x)$. The scalar schedule
concatenates the child words in their declared order, right associates them,
and puts the isolated-vertex factors after them. Hence the literal endpoint is
$R_{\mathrm{sc}}(G,S,\rho)$, and every isolated vertex occurs exactly once.

For completeness, reverse this schedule in reverse postorder. Suffix
reassociation and the one-scalar carrier move return each factor to its
recorded child position. At a composition node, for matching $X,Y$ and
$s:0\to0$, the required reinsertion is the inverse of one of the two raw
chains
\begin{align*}
 ((X\bsum s)\circ Y)
 &\longleftrightarrow
 ((X\bsum s)\circ(Y\bsum e_0))
 \longleftrightarrow
 (X\circ Y)\bsum(s\circ e_0)
 \longleftrightarrow
 (X\circ Y)\bsum s,\\
 X\circ(Y\bsum s)
 &\longleftrightarrow
 (X\bsum e_0)\circ(Y\bsum s)
 \longleftrightarrow
 (X\circ Y)\bsum(e_0\circ s)
 \longleftrightarrow
 (X\circ Y)\bsum s.
\end{align*}
The middle arrows are interchange, and the other arrows are rank-zero unit
steps. Together with the scalar-carrier crossing pair, these moves restore
every instance of the plane recursion and terminate at $R(G,S,\rho)$. Every
scalar term remains opaque. Its exterior type is $0\to0$, the padded
unit/interchange steps leave every positive-rank bank rank unchanged, and
same-type contextual replacement preserves the ranks of all ancestors. Thus
the bridge adds no pattern-width load beyond the endpoints and is covered by
the existing finite-template and width audits. The same reverse-postorder
construction applies to every closed compiler datum $(G,T,\rho)$, proving the
displayed bridge in the statement.
\end{proof}

\begin{lemma}[plane order and degenerate quadrant switches]
\label{supp:lem:quadrant-degeneracies}
The recursive quadrant interpolation preserves the plane orders inherited
from its two original partition trees. Every local change is represented by
the four-cell renderer, including all two-cell and three-cell cases.
\end{lemma}

\begin{proof}
At a recursive cell $C=X\cap Y$, index the first restricted split from the
first tree by $A_0,A_1$ and the first restricted split from the second tree by
$B_0,B_1$, always in inherited plane order. The four possible children are
$A_i\cap B_j$. The row endpoint lists them in row order and the column
endpoint lists them in column order. Each nonempty cell keeps its own inherited
subtree and boundary order. The staging routers alone change the grouping.
Their common outer router restores the prescribed boundary order of $C$, so
the switch does not change the order visible to an ancestor.

With four occupied cells, Equations~\eqref{supp:eq:row} and
\eqref{supp:eq:column} give the displayed switch. With three occupied cells,
use the corresponding presence specialization of
Definition~\ref{supp:def:renderer}. Its endpoints are the unary-suppressed
compiler terms. Structural-unit and identity-router legs may be adjoined to
recover the displayed four-position formula. Every occupied cell remains
present even when its interface is empty.

Because both restricted binary splits are nontrivial, each row and each column
contains an occupied cell. Hence the interpolation never invokes a datum with
fewer than two occupied cells or with both occupied cells in one row or one
column. With exactly two occupied cells, they form a diagonal pair. For the
main diagonal cells $00$ and $11$, the two groupings agree after unary
suppression. For the other diagonal pair $01$ and $10$, row order and column
order are opposite. If both interfaces have positive rank, the fixed symmetry
word and outer-router compensation give the switch. If exactly one has rank
zero, the one-scalar renderer moves that opaque scalar core past the other
carrier. If both have rank zero,
Equation~\eqref{supp:eq:scalar-two-cell} gives the switch. The unused occupancy
specializations remain harmless members of the renderer family.

Every suppression is performed only after the child term has been restored to
completed-step form and is triggered only by an absent child cell. Therefore
no occupied hole, suppressed unary node, or local permutation can change an
ancestor type. Recursive calls use strict nonempty
cells and hence terminate. Each such cell remains the intersection of one
cluster from each original tree. Its boundary is therefore bounded by the sum
of the two original boundary widths. The row, column, and common outer banks
are recorded separately in the renderer, as required by the definition of an
$r$-bounded assembly.
\end{proof}

\begin{lemma}[ordered carrier exposure]
\label{supp:lem:carrier-exposure}
Let $U$ contain a literal occurrence $V$ of a parenthesized parallel word
with ordered leaves $v_1,\ldots,v_n$, where $v_i:p_i\to q_i$. Let
$\mathcal I=(I_1,\ldots,I_s)$ be pairwise disjoint nonempty intervals of
$\{1,\ldots,n\}$, listed from left to right, and fix a parenthesization
$V_{I_j}$ of the interval word indexed by $I_j$. There is an effective
$\mathcal B$-derivation
\begin{equation}\label{supp:eq:carrier-exposure}
 U\xleftrightarrow{*}_{\mathcal B}
 \operatorname{Exp}_{\mathcal I}(U)
\end{equation}
such that every $V_{I_j}$ is a literal subterm occurrence of
$\operatorname{Exp}_{\mathcal I}(U)$. Their syntax-tree addresses are
pairwise prefix-incomparable and occur in the prescribed left-to-right order.
Replacing these occurrences simultaneously by distinct formal holes gives a
linear multi-hole context, and substituting the $V_{I_j}$ back recovers
$\operatorname{Exp}_{\mathcal I}(U)$ literally.

The derivation uses only associativity of $\bsum$, its reverse restores the
original parenthesization, and no expression in the derivation has pattern
width greater than $\patw(U)$. For finitely many pairwise disjoint literal
parallel-word occurrences, $\operatorname{Exp}_{\mathcal I}$ denotes the
successive application of this construction in syntax-address order, with the
same conclusions.
\end{lemma}

\begin{proof}
Replace every interval $I_j$ temporarily by one atom, use the fixed
left-associated parenthesization of the resulting word, and then expand that
atom using the prescribed parenthesization $V_{I_j}$. This defines
$\operatorname{Exp}_{\mathcal I}(U)$. Every parenthesization of a fixed ordered
word is connected to every other one by associativity rotations. The interval
atoms are disjoint, so the resulting occurrences are pairwise
prefix-incomparable.

Every new subterm created by a rotation is a parallel product of a consecutive
subword of $v_1,\ldots,v_n$. Its rank sum is at most
\[
 \sum_{i=1}^{n}(p_i+q_i),
\]
which is the rank sum of the literal subterm $V$ and is therefore at most
$\patw(U)$. The internal subterms of the $v_i$ are unchanged, and every
ancestor retains its type. Hence no rotation increases pattern width.
\end{proof}

\subsection{A two-edge path}
\label{supp:sec:two-edge-path}
Let $G$ be the closed path $v_0v_1v_2$ with edge occurrences $e_{01}$ and
$e_{12}$, both labeled by a letter of rank $(1,1)$. The binary
edge-partition tree has two leaves. Both leaf boundaries are $\{v_1\}$, and
the root boundary is empty. The compiler closes $v_0$ and $v_2$ inside the
leaf wrappers and produces cores $C_{01},C_{12}:0\to1$, whose protected ports
are named $v_1$. The root term is
\[
 (C_{01}\bsum C_{12})\circ
 \operatorname{Can}\bigl(\{\{v_1^{\rm L},v_1^{\rm R}\}\},0\bigr):0\to0.
\]
The router merges the two representatives and closes the result. There is no
isolated-vertex $z$-tail. The construction therefore never exposes all four edge
incidences at once. The boundary load is one and the edge-arity load is two,
so the main text's reduced-realization audit applies with $r=2$ and gives
the deliberately coarse bound $4r+4=12$. Relative spider straightening sends
every alternative discrete staging with these protected ports to the same
router.

\section{Encoding details for the equational word problem}
\label{supp:sec:gi-encoding}

The graph-isomorphism equivalence in the main text uses expanded arities.
Thus $e_n$ may be replaced by a fixed sum of $n$ copies of $e_1$, and an
interface of length $n$ occupies $\Theta(n)$ input symbols. Structural
evaluation creates the disjoint unions and positional quotient maps recorded
at the syntax nodes, so its output hypergraph has size polynomial in the
encoded expression. Labels, directions, begin positions, and end positions
can be represented by constant-size colors and incidence gadgets, reducing
the resulting pointed-hypergraph isomorphism instance to ordinary graph
isomorphism.

\begin{proposition}[isomorphism-invariant reverse encoding]
\label{supp:prop:gi-encoding}
There is a polynomial-time map from a finite simple undirected graph $G$ to a
closed expression $p_G$ such that
\[
 G\cong H
 \quad\Longleftrightarrow\quad
 \val(p_G)\cong\val(p_H)
 \quad\Longleftrightarrow\quad
 p_G\xleftrightarrow{*}_{\mathcal B\cup\E}p_H.
\]
The construction makes no orientation choice that depends on a vertex
ordering.
\end{proposition}

\begin{proof}
Replace every undirected edge $\{u,v\}$ by the two arcs $u\to v$ and
$v\to u$. Denote the resulting symmetric directed graph by $D(G)$ and retain
all vertices, including isolated ones. This construction is functorial under
vertex bijections, and
\[
 G\cong H\quad\Longleftrightarrow\quad D(G)\cong D(H).
\]
In particular, it makes no orientation choice depending on a vertex order.

Let $V(G)=\{1,\ldots,n\}$ in the input encoding. For distinct $u,v$, choose
a fixed adjacent-transposition word $P_{u,v}:n\to n$ whose output interface is
\[
 (u,v,r_1,\ldots,r_{n-2}),
\]
where $r_1,\ldots,r_{n-2}$ are the remaining vertex wires in their original
relative order. Let $P_{u,v}^{-1}$ be the reversed adjacent-transposition
word. Define
\[
 E_{\to}=(\io12\bsum e_1)\circ(e_1\bsum a\bsum e_1)
          \circ(e_1\bsum\io21):2\longrightarrow2
\]
and
\[
 T_{u,v}=P_{u,v}\circ(E_{\to}\bsum e_{n-2})\circ P_{u,v}^{-1}
 :n\longrightarrow n.
\]
The first factor of $E_{\to}$ sends the two live vertex wires $(u,v)$ to
$(u,u,v)$. The middle factor places one $a$-edge from the second copy of $u$
to a fresh node $x$, producing $(u,x,v)$. The final factor identifies $x$
with $v$ and returns the ordered pair $(u,v)$. Thus $T_{u,v}$ adds precisely
the arc $u\to v$ while returning every vertex wire to its original position.

List the $2m$ arcs of $D(G)$ in a deterministic order as
$\alpha_j=u_j\to v_j$, where $m=|E(G)|$. For $n\ge1$, put
\[
 p_G=\io01^{\bsum n}\circ
 T_{u_1,v_1}\circ\cdots\circ T_{u_{2m},v_{2m}}
 \circ\io10^{\bsum n}:0\longrightarrow0,
\]
using fixed parenthesizations throughout. For $n=0$, put $p_G=e_0$. The
initial sum creates one persistent wire for every vertex, and the final sum
closes those wires without deleting their vertices. Consequently isolated
vertices are retained and $\val(p_G)\cong D(G)$.

Moving two selected wires to the first two positions takes $O(n)$ adjacent
transpositions. A transposition padded to rank $(n,n)$ has expanded raw size
$O(n)$, so each $P_{u,v}$ and $P_{u,v}^{-1}$ has size $O(n^2)$. The padded
edge gadget has size $O(n)$. Hence
\[
 |p_G|=O(1+n+mn^2),
\]
and $p_G$ is computable in polynomial time. Although the routing words use the
input encoding order, the closed value is the order-independent graph $D(G)$.
It follows that $G\cong H$ if and only if
$\val(p_G)\cong\val(p_H)$. Completeness of the finite-graphoid equations
\cite[Cor.~3]{BK04} gives the final equivalence in the raw syntax.
\end{proof}

\section{Expanded finite-saturation audit and certificate semantics}
\label{sec:saturation-details}\label{supp:sec:saturation}

This section specifies the exact finite transition system used by the clique
certificates. Throughout the section we use the default rank-$(1,1)$ alphabet
$\Sigma=\{a\}$. Thus every target edge occurrence is represented by an
$a$-leaf. We keep two target modes distinct. In directed mode, the target is a
finite directed multigraph $\vec F=(V,E)$, where $E$ consists of directed edge
occurrences and parallel occurrences remain distinct. In underlying-simple
mode, the target is a finite simple undirected graph $H=(V,E)$, where $E$
consists of unordered pairs. Write $\operatorname{Ori}(H)$ for the directed graphs
obtained by choosing exactly one orientation for every edge of $H$. An
expression realizes $H$ in underlying-simple mode when its directed value lies
in $\operatorname{Ori}(H)$. Thus the orientation is not fixed in advance. In either
mode, for $v\in V$ let $I_F(v)\subseteq E$ be the set of target occurrences
incident with $v$, using $F$ generically for the target. The clique
implementation uses underlying-simple mode. Under the convention of the main
article,
\[
 \patw(H)=\min_{\vec H\in\operatorname{Ori}(H)}\patw(\vec H).
\]
One target and one of these two modes remain fixed throughout every saturation
run. In particular, directed occurrences and unordered occurrences are never
mixed in a state set.

\subsection{Raw states and legality}

\begin{definition}[raw state]\label{supp:def:raw-state}
A width-$W$ raw state is a tuple
\[
 s=(P,C,\kappa,b,e),
\]
where $P\subseteq E$, $C\subseteq V$,
$\kappa=(\kappa_0,\ldots,\kappa_{r-1})\in V^r$, and
\[
 b\in\{0,\ldots,r-1\}^{p},
 \qquad
 e\in\{0,\ldots,r-1\}^{q},
 \qquad p+q\le W.
\]
The integers $0,\ldots,r-1$ are the live-piece indices and $\kappa_i$ is
the target vertex represented by the piece with index $i$. Every live-piece
index must occur in $b$ or $e$. Repeated indices in a boundary tuple are
allowed. Distinct pieces may also have the same target label.

The state is \emph{legal} when
\begin{enumerate}
\item $C\cap\operatorname{im}(\kappa)=\varnothing$,
\item $I_F(v)\subseteq P$ for every $v\in C$,
\item for every occurrence $\alpha\in P$, all its endpoint labels belong to
$C\cup\operatorname{im}(\kappa)$. In directed mode, an occurrence
$\alpha:u\to v$ has the ordered endpoint labels $(u,v)$. In underlying-simple
mode, an occurrence $\alpha=\{u,v\}$ has the unordered endpoint pair
$\{u,v\}$.
\end{enumerate}
\end{definition}

A \emph{realization} of a legal state is a partial expression equipped with a
target-label map from its nodes to $V$. Its generated edge occurrences are
matched bijectively with $P$. In directed mode, if the matched target occurrence
is $\alpha:u\to v$, the tail and head of the generated occurrence have labels
exactly $u$ and $v$, respectively. In underlying-simple mode, if the matched
target occurrence is $\alpha=\{u,v\}$, the generated occurrence is one directed
$a$-edge whose tail and head labels are $u$ and $v$ in one of the two orders.
This exact endpoint-incidence condition is part of the realization invariant. The live
boundary nodes are the indexed pieces and have labels $\kappa$. The internal
nodes map bijectively to $C$. Thus a target vertex may have several live
pieces, but it can never have two internal pieces. The ordered begin and end
tuples are exactly $b$ and $e$. Condition~(3) above is the state-level
coverage condition. The stronger invariant used below is that every reached
state has a realization with the exact endpoint incidences just specified.

For a target with at least one vertex, the nullary unit $e_0$ is omitted from
the seeds. Every occurrence of $e_0$ in a term for a value with at least one
vertex can be removed by the unit laws without increasing width. A vertexless
target is accepted
separately by $e_0$.
For $r\ge2$, the structural unit law replaces a leaf $e_r$ by a fixed
parenthesized sum of $r$ copies of $e_1$. Every subterm of this replacement
has rank sum at most $2r$, the rank sum of the original leaf. Thus only $e_1$
is required among the positive-rank unit seeds, without increasing the width
cap.

\begin{definition}[generator seeds]\label{supp:def:state-seeds}
For every $v\in V$, the seed set contains the following structural states
subject to the width cap:
\begin{align*}
s_{01,v}&=(\varnothing,\varnothing,(v),(),(0)),\\
s_{10,v}&=(\varnothing,\varnothing,(v),(0),()),\\
s_{12,v}&=(\varnothing,\varnothing,(v),(0),(0,0)),\\
s_{21,v}&=(\varnothing,\varnothing,(v),(0,0),(0)),\\
e_{1,v}&=(\varnothing,\varnothing,(v),(0),(0)).
\end{align*}
Their rank sums are $1,1,3,3,2$, respectively.
For every ordered pair $(u,v)\in V^2$, it also contains
\[
 \pi_{u,v}=(\varnothing,\varnothing,(u,v),(0,1),(1,0))
\]
when its rank sum $4$ is at most $W$. In directed mode, for every occurrence
$\alpha:u\to v$ it contains exactly
\[
 a_{\alpha,u,v}=(\{\alpha\},\varnothing,(u,v),(0),(1))
\]
when its rank sum $2$ is at most $W$. In underlying-simple mode, for every edge
$\alpha=\{u,v\}$ it contains both $a_{\alpha,u,v}$ and $a_{\alpha,v,u}$ when
the same cap condition holds. The entries of all displayed boundary tuples are
live-piece indices. Since $P$ is a set and the successor
guards require disjoint placed sets, an underlying-simple edge is nevertheless
used at most once in any realization. Loops occur only in directed mode and use
the same formula with two distinct live pieces carrying the common label. An
isolated vertex $v$ is created and finalized by composing the seeds
$s_{01,v}$ and $s_{10,v}$. No special isolated-vertex seed is needed.
\end{definition}

\subsection{Successor operations}

For two states, first rename their piece sets to be disjoint. Put
$L(s)=\operatorname{im}(\kappa)$. Both operations require
\begin{equation}\label{supp:eq:disjointness-guards}
 P_1\cap P_2=\varnothing,
 \quad C_1\cap C_2=\varnothing,
 \quad C_1\cap L(s_2)=\varnothing,
 \quad C_2\cap L(s_1)=\varnothing.
\end{equation}

\begin{definition}[sum successor]\label{supp:def:sum-successor}
If the guards in \eqref{supp:eq:disjointness-guards} hold, define
$s_1\bsum s_2$ by taking $P=P_1\cup P_2$, $C=C_1\cup C_2$, concatenating
$\kappa_1$ and $\kappa_2$, and concatenating $b_1,b_2$ and $e_1,e_2$ after
shifting the piece indices of the second state. The successor is retained
exactly when it is legal and its rank sum is at most $W$.
\end{definition}

\begin{definition}[composition successor]\label{supp:def:composition-successor}
Suppose that the guards in \eqref{supp:eq:disjointness-guards} hold and
$|e_1|=|b_2|$. Form the disjoint union of the two live-piece sets and identify
the piece indexed by the $i$th entry of $e_1$ with the piece indexed by the
$i$th entry of $b_2$. Reject the operation if any identified pair has
different target labels. Let $\sim$ be the equivalence relation generated by
these positional identifications.

The new boundary consists of the $\sim$-classes referenced by $b_1$ or $e_2$.
Let $D$ be the set of $\sim$-classes absent from the new boundary. Each class
in $D$ has a common target label because identifications between differently
labeled pieces were rejected. For every $v\in V$, let $R_v$ be the set of
$\sim$-classes whose pieces have label $v$. For every $B\in D$ with label
$v$, require both conditions
\begin{equation}\label{supp:eq:finalization-test}
 R_v=\{B\}
 \qquad\text{and}\qquad
 I_F(v)\subseteq P_1\cup P_2.
\end{equation}
Reject the operation if either condition fails for any $B\in D$. If all tests
pass, take the finalized set to be
\[
 C_1\cup C_2\cup
 \{v:\text{some }B\in D\text{ has label }v\}.
\]
Delete precisely the classes in $D$, renumber the remaining classes by first
occurrence in the new begin-then-end tuple, and retain the resulting state
exactly when it is legal and has rank sum at most $W$.
\end{definition}

Condition \eqref{supp:eq:finalization-test} is the point at which the
internal-node lemma enters the finite system. Once one piece of target class
$v$ becomes internal, no distinct live piece of class $v$ can be merged with
it later. Every incident target edge must also have been placed already.

\begin{lemma}[endpoint-incidence preservation]
\label{supp:lem:endpoint-incidence}
Every generator seed has a realization satisfying the exact endpoint-incidence
invariant. If two states have such realizations and a sum or composition
successor is retained, then the successor also has such a realization.
\end{lemma}

\begin{proof}
The structural seeds have no edge occurrence. In directed mode, the $a$-seed
for $\alpha:u\to v$ has one occurrence whose tail and head carry exactly the
labels $u$ and $v$. In underlying-simple mode, each of the two $a$-seeds for
$\alpha=\{u,v\}$ has one occurrence with the corresponding ordered endpoint
labels and hence with unordered endpoint pair $\{u,v\}$.

For a sum, take the disjoint union of the two realizations. The guard
$P_1\cap P_2=\varnothing$ makes the two edge-occurrence bijections combine to
a bijection with $P_1\cup P_2$, and every endpoint label is unchanged. For a
composition, take the positional quotient. Only equally labeled pieces are
identified, so the target-label map descends to the quotient. Composition
neither creates an edge occurrence nor changes its endpoint nodes. Hence its
edge-occurrence bijection and exact endpoint labels are preserved. Finally,
for each $B\in D$ with label $v$, the condition $R_v=\{B\}$ makes $B$ the
unique remaining class with that label, while
$I_F(v)\subseteq P_1\cup P_2$ makes it complete. Moving this class from the
live boundary to the internal set therefore preserves the full realization
invariant.
\end{proof}

The accepting state is
\begin{equation}\label{supp:eq:accepting-state}
 s_{\rm acc}=(E,V,(),(),()).
\end{equation}

\begin{proposition}[finite saturation in either target mode]\label{supp:prop:finiteclosure}
\label{supp:prop:finite-saturation}
Fix an integer $W\ge0$. For either a finite directed target $\vec F$ or a
finite underlying-simple target $H$ over the default alphabet
$\Sigma=\{a\}$, the least legal state set in the corresponding mode that
contains every generator seed of rank sum at most $W$ and is closed under the
two successor operations is finite. In directed mode, if
$V(\vec F)\ne\varnothing$, this set contains $s_{\rm acc}$ exactly when
$\patw(\vec F)\le W$. In underlying-simple mode, if
$V(H)\ne\varnothing$, this set contains $s_{\rm acc}$ exactly when
\[
 \patw(\vec H)\le W\quad\text{for some }\vec H\in\operatorname{Ori}(H),
\]
equivalently exactly when $\patw(H)\le W$ under the convention above. Every
expression reconstructed from an accepting state in this mode has underlying
simple graph $H$. A vertexless target is accepted by the separate term $e_0$.
\end{proposition}

\begin{proof}
Assume first that the target vertex set is nonempty. In directed mode, start
with a width-$W$ expression $t$ whose exact value is $\vec F$. In
underlying-simple mode, start with a width-$W$ expression $t$ whose value is some
$\vec H\in\operatorname{Ori}(H)$ and match each of its directed occurrences with the
unique underlying edge of $H$. By the preceding unit reductions, $t$ may be
chosen with no $e_0$ leaf and with $e_1$ as its only positive-rank unit leaf.
For a subterm $u$, let $P$ be the matched target occurrences generated inside
$u$. Index the distinct boundary nodes of $\val(u)$, label them by their target
vertices, and record the ordered begin and end occurrences. By
Lemma~\ref{supp:lem:internal}, an internal node
cannot later acquire an incident edge or be identified with an outside node.
It therefore determines one member of $C$, all its incident target edges
belong to $P$, and two internal nodes cannot have the same target image. The
canonical map sends the endpoints of every generated occurrence to the
endpoints of its matched target occurrence. It preserves the prescribed order
in directed mode and the prescribed unordered pair in underlying-simple mode.
Thus every subterm gives a legal state together with the full realization
invariant above. Leaves give the appropriate directed or underlying-simple
seeds, and the two parse operations give exactly the two successors above. The
root gives \eqref{supp:eq:accepting-state}.

Conversely, retain a derivation tree for every state discovered by saturation.
Lemma~\ref{supp:lem:endpoint-incidence}, applied inductively to this tree,
reconstructs a partial expression satisfying the full realization invariant
for every reached state. In a reached accepting state, the vertex labels are
exactly $V$ and the matched edge occurrences are exactly $E$. In directed
mode, the endpoint-incidence invariant gives their exact incidence and
directions, so the closed expression has exact value $\vec F$. In
underlying-simple mode, it gives one directed occurrence over each edge of $H$
and no other occurrence. The reconstructed value is therefore some
$\vec H\in\operatorname{Ori}(H)$, and its underlying simple graph is exactly $H$.
The cap bounds the rank sum of every reconstructed subterm by $W$.

For arbitrary $F$, finiteness is immediate. There are finitely many choices
for $P$ and $C$.
There are at most $W$ live pieces, each labeled in the finite set $V$, and
the two boundary tuples have total length at most $W$. Renaming live pieces by
first boundary occurrence gives a finite list of representatives.
\end{proof}

\subsection{Symmetry quotient and static certificates}
\label{supp:sec:symmetry-certificates}

The proposition above does not require symmetry reduction. In underlying-simple
mode, for the clique target $H=K_n$, a state with $r$ live pieces is acted on by
\[
 \Gamma_r=S_n\times C_2\times S_r.
\]
A vertex permutation acts on $P$, $C$, and every entry of $\kappa$, using its
induced permutation of unordered edge occurrences. The nontrivial element of
$C_2$ swaps $b$ and $e$. It exchanges the discrete seed kinds $01$ with $10$
and $12$ with $21$. Since clique edge occurrences are unordered, it leaves
$P$ unchanged. A piece permutation $\rho\in S_r$ leaves $P$ and $C$
unchanged, maps every live-piece index in $b,e$ by $\rho$, and replaces
$\kappa$ by $\kappa\circ\rho^{-1}$.

For a general directed target $\vec F$, let $A_{\vec F}$ be a chosen group of
automorphisms $\phi=(\phi_V,\phi_E)$ acting by
\[
 T_\phi(P,C,\kappa,b,e)
 =\bigl(\phi_E(P),\phi_V(C),\phi_V\circ\kappa,b,e\bigr).
\]
Begin-end duality may be used
only after choosing involutions $\delta_V$ on $V$ and $\delta_E$ on $E$ such
that every occurrence $\alpha:u\to v$ is sent to an occurrence
$\delta_E(\alpha):\delta_V(v)\to\delta_V(u)$. Its action sends
$(P,C,\kappa,b,e)$ under the map
\[
 R_\delta(P,C,\kappa,b,e)
 =(\delta_E(P),\delta_V(C),\delta_V\circ\kappa,e,b).
\]
Write $Q_\rho$ for the piece-renaming map determined by $\rho\in S_r$ above.
Let $\Gamma_r(\vec F)$ be the finite group generated by the maps
$T_\phi$ for $\phi\in A_{\vec F}$, $Q_\rho$ for $\rho\in S_r$, and $R_\delta$
when a reversal involution has been chosen. If no reversal involution is
specified, duality is omitted.

In underlying-simple mode, choose $A_H\le\operatorname{Aut}(H)$. Every
$\phi\in A_H$ acts on vertices and unordered edge occurrences by the same
formula $T_\phi$. Let $R_{\mathrm{id}}$ swap $b$ and $e$ while fixing $P$,
$C$, and $\kappa$. This map reverses the orientation chosen by a realization
without changing its underlying-simple target. Let $\Gamma_r(H)$ be generated
by these maps, the piece-renaming maps, and $R_{\mathrm{id}}$. For the clique
with $A_H=S_n$, this recovers the displayed $S_n\times C_2\times S_r$. We write
$\Gamma_r(F)$ generically for the resulting group in either mode and suppress
$F$ when the target is fixed.
The action
preserves the seed set up to piece renaming, legality, the width cap, and the
accepting state. It preserves sum
successors. Duality sends a composition successor of $(s,t)$ to the
composition successor of the dual pair in the reverse operand order:
$R_\delta(s\circ t)=R_\delta(t)\circ R_\delta(s)$.

Normalize piece indices by order of first occurrence in $b$ followed by $e$.
Write $\operatorname{norm}$ for this operation. Piece renaming is a congruence
for both successors. We may therefore normalize every retained successor. The
resulting quotient reaches the normalized accepting state exactly when the raw
saturation does. We use normalized states from this point onward. Fix the
lexicographic encoding $(P,C,\kappa,b,e)$ and put
\[
 \operatorname{can}(s)=\min_{\rm lex}
 \{\operatorname{norm}(g\mathbin{\cdot}s):g\in\Gamma_r(F)\}.
\]

Write
\[
 \operatorname{Orb}(s)=
 \{\operatorname{norm}(g\mathbin{\cdot}s):g\in\Gamma_{r(s)}(F)\},
 \qquad
 U_R=\bigcup_{s\in R}\operatorname{Orb}(s),
\]
where $r(s)$ is the number of live pieces in $s$.

\begin{lemma}[reduced orbit-closure test]
\label{supp:lem:reduced-orbit-closure}
Let $R$ be a set of canonical states. Suppose that, for every $s,t\in R$ and
every $x\in\operatorname{Orb}(s)$, each legal successor within the width cap
among
\[
 x\bsum t,\qquad t\bsum x,\qquad x\circ t,\qquad t\circ x
\]
has its canonical representative in $R$. Then $U_R$ is closed under both
normalized successor operations.
\end{lemma}

\begin{proof}
Take $y\in\operatorname{Orb}(s)$ and $z\in\operatorname{Orb}(t)$ for
$s,t\in R$. Local piece renaming commutes with either successor after
normalization, so discard that part of a symmetry taking $t$ to $z$. Let $h$
be the remaining word in target automorphisms and reversal, let
$\epsilon(h)$ be its reversal parity, and put
$x=\operatorname{norm}(h^{-1}\mathbin{\cdot}y)\in\operatorname{Orb}(s)$.
Simultaneous application of $h^{-1}$ gives
\[
 \operatorname{norm}\bigl(h^{-1}\mathbin{\cdot}(y\bsum z)\bigr)
   =\operatorname{norm}(x\bsum t)
\]
and
\[
 \operatorname{norm}\bigl(h^{-1}\mathbin{\cdot}(y\circ z)\bigr)
 =
 \begin{cases}
   \operatorname{norm}(x\circ t),&\epsilon(h)=0,\\
   \operatorname{norm}(t\circ x),&\epsilon(h)=1.
 \end{cases}
\]
The generator identities above and induction on $h$ prove these formulas and
preserve definedness, legality, and the cap. The hypothesis places the
transformed successor in $U_R$, and orbit invariance returns the original one
to $U_R$.
\end{proof}

\begin{definition}[certificate invariant]\label{supp:def:certificate}
Fix either a finite directed target or a finite underlying-simple target $F$
with $V(F)\ne\varnothing$. A width-$W$ lower-bound certificate in that target
mode is a finite set $R$ of canonical states. The following five checks are
required.
\begin{enumerate}
\item Every member of $R$ is legal and within the cap.
\item Every generator seed of rank sum at most $W$ has its canonical
representative in $R$.
\item For every $s,t\in R$ and every $x\in\operatorname{Orb}(s)$, each legal
successor of rank sum at most $W$ among
$x\bsum t$, $t\bsum x$, $x\circ t$, and $t\circ x$ has its canonical
representative in $R$.
\item Every $s\in R$ is canonical, meaning
$s=\operatorname{can}(s)$. Consequently, if $r(s)=r(t)=r$ and $s\ne t$,
then $s$ and $t$ belong to distinct $\Gamma_r(F)$-orbits.
\item The accepting orbit is absent, equivalently
$\operatorname{can}(s_{\rm acc})\notin R$.
\end{enumerate}
\end{definition}

\begin{proposition}[certificate soundness]\label{supp:prop:certificate-soundness}
If $V(\vec F)\ne\varnothing$ and a certificate invariant exists for the
directed target $\vec F$ at width $W$, then $\patw(\vec F)>W$. If
$V(H)\ne\varnothing$ and a certificate invariant exists for the
underlying-simple target $H$ at width $W$, then
\[
 \patw(\vec H)>W\quad\text{for every }\vec H\in\operatorname{Ori}(H),
\]
equivalently $\patw(H)>W$.
\end{proposition}

\begin{proof}
The union of the orbits represented by $R$ contains every seed. By
Lemma~\ref{supp:lem:reduced-orbit-closure}, the third certificate check makes
this union closed under both normalized successor operations. It therefore
contains the
least normalized reachable state set. The accepting orbit is absent, so
the appropriate nonempty-target equivalence in
Proposition~\ref{supp:prop:finiteclosure} gives the directed conclusion for
$\vec F$. In underlying-simple mode, every edge has both oriented seeds. The
same equivalence therefore rules out an accepting expression for every
orientation of $H$ and gives $\patw(H)>W$.
\end{proof}

In the JSON files, $P$ and $C$ are bit masks. The field
\texttt{PAIRS} fixes the unordered edge-bit indexing for the clique target.
Both oriented $a$-seeds are checked for every indexed pair, so the resulting
certificate conclusion concerns the underlying simple clique and rules out all
its orientations. The remaining fields are the lists
\texttt{classes}, \texttt{begin}, and \texttt{end}, which encode
$\kappa,b,e$. The standalone checker reconstructs the clique group action,
generates all seed orbits, and verifies that the accepting orbit is absent. For
the closure test, one operand ranges through a full orbit and the other ranges
through the canonical representatives. The checker applies both operations in
both operand orders. Lemma~\ref{supp:lem:reduced-orbit-closure} proves that this
is equivalent to testing the whole orbit union. The checker does not import
the search engine that proposed $R$.

\clearpage
\section*{Declarations}

\noindent
Funding: No external funding was received for this work.

\noindent
Competing interests: The author declares no competing interests.

\noindent
Ethics approval, consent to participate, and consent for publication: Not
applicable.

\noindent
Data and code availability: The code, witnesses, certificates, manifests,
archived logs, replay scripts, and licenses supporting the computational
results are archived in reproducibility package version 1.1.0 on Zenodo at
\href{https://doi.org/10.5281/zenodo.22297410}{\nolinkurl{doi:10.5281/zenodo.22297410}}.

\noindent
Author contributions: A.~Kalampakas is the sole author and was responsible for
the conceptualization, methodology, formal analysis, software, validation,
writing, and approval of the manuscript.

\bibliographystyle{amsplain}
\bibliography{mybib}
\end{document}